\documentclass[a4paper,fleqn]{cas-dc}

\usepackage{amsmath,amssymb,amsfonts,amsthm}
\usepackage{mathtools}
\usepackage{bm}
\usepackage{float}
\usepackage{wrapfig}
\usepackage{graphicx}
\usepackage{booktabs}
\usepackage{multirow}
\usepackage{array}
\usepackage{adjustbox} % adjustbox for fit-to-column tables/figures % prevents floats crossing sections
\usepackage{enumitem}
\usepackage[section]{placeins}
\PassOptionsToPackage{numbers,sort&compress}{natbib}
\usepackage[numbers]{natbib}
\providecommand{\eprint}[2][]{\url{#2}}

\AtBeginDocument{\hypersetup{colorlinks=true,linkcolor=blue,citecolor=blue,urlcolor=blue}}

\usepackage{algorithm}
\usepackage{algpseudocode}

\newcommand{\SelectBibliographyStyle}{%
  \IfFileExists{cas-model1-num.bst}{\bibliographystyle{cas-model1-num}}{%
    \IfFileExists{elsarticle-num.bst}{\bibliographystyle{elsarticle-num}}{%
      \bibliographystyle{unsrtnat}}}%
}
\SelectBibliographyStyle

\usepackage{tikz}
\usetikzlibrary{arrows.meta,positioning,fit,calc,backgrounds}
\pgfdeclarelayer{background}
\pgfsetlayers{background,main}

\newcommand{\R}{\mathbb{R}}

\newcommand{\diag}{\mathrm{diag}}
\newcommand{\sat}{\mathrm{sat}}
\newcommand{\Lcal}{\mathcal{L}}
\newcommand{\Hcal}{\mathcal{H}}
\usepackage{xcolor}
\newcommand{\rev}[1]{#1}

\newtheorem{assumption}{Assumption}
\newtheorem{lemma}{Lemma}
\newtheorem{theorem}{Theorem}

\newtheorem{corollary}{Corollary}
\newtheorem{definition}{Definition}
\newtheorem{remark}{Remark}
\newtheorem{problem}{Problem}

\begin{document}
%==============================================================
% Front matter
%==============================================================
\let\WriteBookmarks\relax
\def\floatpagepagefraction{1}
\def\textpagefraction{.001}

\shorttitle{Observer-Assisted Relative-Velocity Compensation with LPV-$\Hcal_\infty$ Correction for AUVs}
% \shortauthors{Anonymized}

\title[mode = title]{Observer-Assisted Relative-Velocity Compensation
with LPV-$\mathcal{H}_\infty$ Robust Correction for 3D Trajectory
Tracking of Underactuated Non-Minimum-Phase AUVs under Ocean Currents}

\author[1]{Mohammad Sabouri}[orcid=0000-0002-2568-3253]
\cormark[1]
\ead{S5659227@studenti.unige.it}
\credit{Conceptualization, Methodology, Software, Formal analysis,
Investigation, Validation, Visualization, Writing -- original draft,
Writing -- review \& editing}

\affiliation[1]{organization={Department of Informatics, Bioengineering,
                Robotics and Systems Engineering (DIBRIS), University of
                Genoa},
                addressline={Via all'Opera Pia 13},
                city={Genoa},
                postcode={16145},
                country={Italy}}

\cortext[1]{Corresponding author.}
\shortauthors{M. Sabouri}

\begin{abstract}
This paper develops an observer-assisted control architecture for
3D trajectory tracking of torpedo-type underactuated AUVs with
non-minimum-phase sway/heave dynamics under unknown ocean currents.
A three-stage state--current observer provides relative-velocity
estimates to a nonlinear feedforward term for dominant current
rejection and to an LMI-certified LPV-$\Hcal_\infty$ correction layer.
Because feedback-linearising cancellation makes the input matrix
constant, the scheduled correction is synthesised by convex LMIs
without pairwise cross terms. A residual-level break-even law links the
effective surge disturbance to the current-estimation error, and a
singular-perturbation argument establishes local practical uniform
ultimate boundedness on the embedded LPV model. \rev{The observer-assisted relative-velocity
feedforward is the primary performance mechanism; the
LPV-$\Hcal_\infty$ layer supplies certified robustness and a
conditional scheduling margin under reduced feedback authority.}
REMUS simulations over three trajectories and
four current scenarios show $89$--$96\%$ current-estimation reduction,
about $99\%$ translational residual reduction, and RMS error reduction
from $4.04$~m to $0.24$~m.
\end{abstract}


\maketitle

%==============================================================
\section{Introduction}\label{sec:introduction}
%==============================================================

\subsection{Motivation}\label{subsec:motivation}

Autonomous underwater vehicles (AUVs) are central to
oceanographic survey, subsea inspection, and environmental monitoring
\citep{fossen2011handbook,wynn2014autonomous}. The dominant
torpedo-type class is inherently underactuated and
non-minimum-phase: rudder and stern-plane actuators
generate yaw and pitch moments while inducing sway and
heave accelerations, giving non-zero off-diagonal entries in
the input matrix \citep{prestero2001verification,godhavn1996nonlinear,
li2023trajectory}. This coupling complicates recursive
control design and rules out a direct extension of horizontal-plane
results to three-dimensional trajectory tracking
\citep{do2002underactuated,jiang2002global,ashrafiuon2008sliding}.

The ocean current --- the
dominant exogenous disturbance --- enters the dynamics through the
\emph{relative} velocity $\bm\nu_r=\bm\nu-\bm\nu_c$ and thus acts
multiplicatively on the damping, lift, and Coriolis terms
\citep{fossen2011handbook,ahmed2023survey}. Treating it as an additive
lumped disturbance, as most robust and adaptive
designs do, discards this structural information.

\subsection{Present study}\label{subsec:present_study}

Three-dimensional path-following and trajectory-tracking control for
torpedo-type AUVs has been pursued via robust adaptive backstepping
\citep{wang2019command,li2023trajectory}, neural-network feedback
\citep{li2020neural}, finite-time control \citep{yu2019globally},
prescribed-performance disturbance-observer control
\citep{sun2024prescribed}, and recently approximate-optimal adaptive control with prescribed performance \citep{chen2024prescribed}.
Closest to the present work, \citet{li2023trajectory} introduces a
spherical-coordinate transformation
that reduces the 3D kinematics to a three-input three-output feedback
form together with an exponential modification of orientation (EMO)
that removes the singularity when the position error vanishes. We
adopt that geometric backbone and develop a controller for it.

Ocean disturbances have been addressed through disturbance-observer-based
control \citep{guerrero2020adaptive,he2024nonlinear,lei2025constrained},
indirect adaptive disturbance observers in the line-of-sight framework
\citep{du2023improved}, extended-state-observer and geometric line-of-sight guidance for marine craft \citep{song2024cascaded,du2025improved,yuan2025geometric},
tube-based model predictive control under saturation
\citep{jimoh2024tube}, and event-triggered disturbance-observer control
for surface vehicles \citep{ning2024event}. Direct ocean-current
estimation has been studied in
\citep{kim2018current,kim2020current,refsnes2007model}. In a
kinematic setting, \citep{sabouri2026kinematic} proposed a 3D
path-following controller for an underactuated
AUV under a constant unknown current, using a current estimator based on
the relative position between the vehicle and a (possibly virtual)
moving target together with a Lyapunov-based switching guidance law;
the present paper develops the dynamic-level, observer-assisted
relative-velocity counterpart with a certified LPV-$\Hcal_\infty$
correction.

Linear parameter-varying control bridges linear $\Hcal_\infty$
synthesis and nonlinear dynamics through gain scheduling on measurable
parameters \citep{shamma1990analysis,apkarian1995self,wu1995induced,
scherer2000linear}, with marine applications including path following
\citep{zhang2020dolpv}, depth control \citep{silvestre2007depth}, and
course keeping \citep{chang2020gain}.

Two limitations cut across these lines. \emph{First}, fixed-gain
robust and adaptive designs do not exploit the structured variation
of the vehicle dynamics across speed, attitude, and depth; the
gain-scheduling line that does \citep{zhang2020dolpv,silvestre2007depth,
chang2020gain} predominantly addresses fully actuated or planar
cases. \emph{Second}, what is missing is not
current observers or gain scheduling per se but an explicit
\emph{residual-level break-even law} quantifying \emph{when}
estimated-relative-velocity scheduling improves over absolute-velocity
compensation, together with the observer-side excitation condition
that determines its applicability, in the 3D underactuated
non-minimum-phase setting.

\subsection{Objectives and contributions}\label{subsec:contributions}

Building on the geometric formulation of \citet{li2023trajectory},
this work designs an observer-assisted architecture for 3D trajectory
tracking of underactuated non-minimum-phase AUVs under unknown ocean
currents, combining a nonlinear relative-velocity feedforward term
for the dominant hydrodynamic current compensation with an
LMI-certified scheduled LPV-$\mathcal H_\infty$ correction layer
for the robust performance margin. The main contributions
are as follows.

First, a three-stage state–current observer estimates
the unmeasured sway/heave velocities and the body-frame ocean current
from actuator-coupled dynamic residuals and kinematic innovations; its
sufficient excitation condition is expressed through a two-dimensional
directional Gramian.

Second, a residual-level break-even law is derived for the
relative-velocity feedforward: after
compensation at the estimated relative velocity, the effective
surge-channel disturbance scales with the current-estimation error
rather than with the current itself.

Third, a polytopic LPV-$\mathcal H_\infty$ correction layer is
formulated on the pre-stabilised tracking model. Because the
feedback-linearising cancellation yields a constant input matrix, the
scheduled synthesis remains convex and avoids the pairwise cross-term
LMIs that would arise with a parameter-dependent input matrix.

Fourth, a singular-perturbation analysis establishes local
practical uniform ultimate boundedness of the embedded
observer–controller interconnection under the stated excitation,
embedding-margin, unsaturated-cone, and small-gain conditions.

Finally, a structured REMUS simulation study separates the roles of
observer, feedforward, scheduling, and LPV correction, identifying the
observer-assisted relative-velocity feedforward as the
dominant performance mechanism and the LPV-$\mathcal H_\infty$
layer as the certified robustness mechanism supplying a conditional
scheduling margin under reduced feedback authority.

\rev{Table~\ref{tab:novelty} separates the tools \emph{adapted} from the
literature from the parts \emph{developed here}. The spherical-coordinate
reduction and the EMO backbone are taken from
\citet{li2023trajectory}; the bounded-real LMI, the high-gain
differentiator, and the persistent-excitation
\citep{narendra1989stable} and singular-perturbation
\citep[Ch.~11]{khalil2002nonlinear} machinery are standard. The new elements are the joint
three-stage state--current observer with the directional-Gramian
condition, the surge break-even law, the observation that the
feedback-linearising cancellation yields a constant input matrix (making
the scheduled correction convex without pairwise cross-term LMIs), and
the composite observer--controller practical-UUB certificate.}

\rev{%
\begin{table}[!htbp]\centering
\caption{Adapted tools versus contributions of this paper.}
\label{tab:novelty}
\begin{adjustbox}{max width=\columnwidth}
\setlength{\tabcolsep}{4pt}\small
\begin{tabular}{@{}ll@{}}
\toprule
\textbf{Adapted (cited)} & \textbf{New in this paper}\\
\midrule
Spherical reduction, EMO \citep{li2023trajectory} & Three-stage joint state--current observer\\
Bounded-real LMI \citep{scherer2000linear} & Directional-Gramian excitation condition\\
High-gain differentiator \citep{khalil2002nonlinear} & Surge residual break-even law\\
PE \citep{narendra1989stable} / singular-perturbation \citep[Ch.~11]{khalil2002nonlinear} theory & Constant-$B_u$ convex LPV-$\Hcal_\infty$ correction\\
 & Composite small-gain practical-UUB certificate\\
\bottomrule
\end{tabular}
\end{adjustbox}
\end{table}}

% \subsection{Paper organisation}\label{subsec:organisation}

% Section~\ref{sec:problem} introduces foundational concepts, the
% layered analytical roadmap, the AUV system description, the
% spherical-coordinate reduction, and the problem statement.
% Section~\ref{sec:main} develops the three-stage observer, the
% scheduling map, the LPV embedding, the LMI synthesis, and the
% closed-loop stability theorem; all proofs appear inline with the
% corresponding theorem. Section~\ref{sec:simulation} reports the
% numerical study, including conditional-benefit validation, robustness
% verification under noise and parameter uncertainty, and comparison
% with \citet{li2023trajectory}. Sections~\ref{sec:discussion} and
% \ref{sec:conclusion} close with discussion and conclusion.
\subsection{Paper organisation}\label{subsec:organisation}
Section~\ref{sec:problem} presents the foundational concepts,
design overview, AUV description, spherical-coordinate reduction,
and problem statement. Section~\ref{sec:main} develops the
three-stage observer, scheduling map, LPV embedding, LMI synthesis,
and closed-loop stability theorem, \rev{with the longer proofs collected
in Appendix~\ref{app:proofs}}.
\rev{Section~\ref{sec:implementation} gives the per-cycle implementation
form.}
Section~\ref{sec:simulation} reports the numerical study
(conditional-benefit validation, noise/parameter robustness, and
comparison with \citet{li2023trajectory}).
Section~\ref{sec:conclusion} concludes the paper.

%==============================================================
\section{System Description and Problem Formulation}\label{sec:problem}
\subsection{Overview of the proposed design}
\label{subsec:framework}

The development proceeds through five layers, used
throughout (the real-time closed-loop signal flow is given
in Section~\ref{subsec:closed_loop_architecture}):
$\mathcal P_0$ is the full 5-DOF underactuated AUV model with
input $\bm\tau\in\R^3$, state $(\bm\eta,\bm\nu)\in\R^{10}$, and
exogenous current $\bm V_c^n$; $\mathcal P_1$ is the reduced
spherical-coordinate tracking model, whose feedback-linearised
pre-stabilisation yields the LPV form
$\dot{\bm x}_e=A_{\rm pre}(\bm\rho)\bm x_e+
B_u\delta\bm v_{\Hcal_\infty}+E(\bm\rho)\bm w$ with constant
input matrix $B_u$; $\mathcal O$ is the joint state--current
observer producing $(\hat{\bm\nu}_r,\hat{\bm\nu}_c)$; $\mathcal S$
is the estimated scheduling map
$\hat{\bm\rho}=\Phi(\hat{\bm\nu}_r,\theta,p_e)$; and $\mathcal C$
is the scheduled LPV-$\mathcal H_\infty$ correction layer with
gain $K(\bm\rho)=(\sum_i\lambda_i(\bm\rho)W_i)Y^{-1}$.

\subsection{5-DOF underactuated AUV system description}
\label{subsec:dof5}

Consider a torpedo-type AUV in five degrees of freedom (roll
stabilised by design \citep{fossen2011handbook}). With
$\bm\eta=[x,y,z,\theta,\psi]^{\!\top}$,
$\bm\nu=[u,v,w,q,r]^{\!\top}$, and control input
$\bm\tau=[\tau_u,\tau_q,\tau_r]^{\!\top}\in\R^3$, the kinematics and
dynamics read
\begin{equation}\label{eq:kindyn}
\dot{\bm\eta}=J(\theta,\psi)\bm\nu,\quad
\dot{\bm\nu}=\bm f(\bm\nu_r,\theta)+G\bm\tau+\bm d,
\end{equation}
with $\bm d\in\R^5$ bounded residual disturbance, and
\begin{equation}\label{eq:Gmat}
G=\begin{bmatrix}
g_u & 0   & 0\\
0   & 0   & \epsilon_r\\
0   & \epsilon_q & 0\\
0   & g_q & 0\\
0   & 0   & g_r
\end{bmatrix}.
\end{equation}
The positive entries $g_u,g_q,g_r$ are the direct surge, pitch, and
yaw actuator coefficients (acceleration-normalised), and
$\epsilon_q,\epsilon_r\neq 0$ are the non-minimum-phase
coupling terms characteristic of torpedo-type AUVs
\citep{prestero2001verification}: the rudder ($\tau_r$) produces a sway
acceleration $\epsilon_r\tau_r$, and the stern plane ($\tau_q$) a heave
acceleration $\epsilon_q\tau_q$.

\begin{remark}[Non-minimum-phase coupling]\label{rem:nmp}
The presence of $\epsilon_q,\epsilon_r\neq 0$ alone does not establish
non-minimum-phase behaviour. The output zero dynamics associated with
$\bm y=(x,y,z)$ are obtained after the kinematic chain
$\dot{\bm\eta}=J(\theta,\psi)\bm\nu$ is cascaded with~\eqref{eq:Gmat}
and contain unstable internal modes driven by the sway/heave channels,
as established in \citep[Sec.~II]{li2023trajectory}. We adopt the
terminology ``non-minimum-phase coupling'' for $(\epsilon_q,\epsilon_r)$.
\rev{Three distinct notions are involved here: \emph{underactuation} (three inputs
for five degrees of freedom), \emph{actuator coupling} (the off-diagonal
$\epsilon_q,\epsilon_r$ in $G$), and \emph{non-minimum-phase behaviour}
(instability of the zero dynamics of the position output). The first two
do not imply the third; the non-minimum-phase property here follows
specifically from the sway/heave-driven internal dynamics of the
$(x,y,z)$ output for the $G$ of~\eqref{eq:Gmat}.}
\end{remark}

\textit{Ocean current via relative velocity.}
Let $\bm V_c^n\in\R^3$ be the navigation-frame current and
$R(\theta,\psi)\in\mathrm{SO}(3)$ the body-to-navigation rotation
under the roll-stabilised 5-DOF representation. The body-frame
current is $\bm\nu_c=R^{\!\top}(\theta,\psi)\bm V_c^n=[u_c,v_c,w_c]^{\!\top}$,
and the water-relative velocity is
$\bm\nu_r=[u_r,v_r,w_r,q,r]^{\!\top}$ with $u_r=u-u_c$, $v_r=v-v_c$,
$w_r=w-w_c$.

\rev{The three velocities use distinct symbols throughout
(Table~\ref{tab:notation}): the
body/ground velocity $\bm\nu$, the body-frame current
$\bm\nu_c=R^{\!\top}\bm V_c^n$, and the water-relative velocity
$\bm\nu_r=\bm\nu-\bm\nu_c$ that drives the hydrodynamics.}

\rev{%
\begin{table}[!htbp]\centering
\caption{Velocity notation used throughout.}\label{tab:notation}
\begin{adjustbox}{max width=\columnwidth}
\setlength{\tabcolsep}{4pt}\small
\begin{tabular}{@{}lll@{}}
\toprule
\textbf{Symbol} & \textbf{Meaning} & \textbf{Role}\\
\midrule
$\bm V_c^n$ & navigation-frame current & exogenous, unknown\\
$\bm\nu$ & body/ground velocity & plant state\\
$\bm\nu_c=R^{\!\top}\bm V_c^n$ & body-frame current & rotated $\bm V_c^n$\\
$\bm\nu_r=\bm\nu-\bm\nu_c$ & water-relative velocity & drives hydrodynamics\\
$\hat{\bm\nu}_r,\hat{\bm\nu}_c$ & observer estimates & used for control\\
\bottomrule
\end{tabular}
\end{adjustbox}
\end{table}}

Even under a constant navigation-frame current, the body-frame
current varies through vehicle rotation. The body-frame angular
velocity vector is
\begin{equation}\label{eq:omega_b}
\bm\omega_b:=[0,q,r]^{\!\top}\in\R^3,
\end{equation}
where $(q,r)$ are body-frame angular-velocity components, not Euler
rates \citep[Sec.~2]{fossen2011handbook}. The body-frame current
evolves as
\begin{equation}\label{eq:nuc_dot}
\dot{\bm\nu}_c=-\bm\omega_b\times\bm\nu_c+R^{\!\top}\dot{\bm V}_c^n.
\end{equation}
A steady navigation-frame current thus generates time-varying
body-frame disturbances proportional to $(q,r)$, while truly
time-varying currents contribute the second term in~\eqref{eq:nuc_dot};
this decomposition underlies the observer of
Section~\ref{subsec:observer}.

\begin{assumption}[Current]\label{ass:current}
$\bm V_c^n$ satisfies one of:
\textnormal{(C1)} $\dot{\bm V}_c^n=\bm 0$, $\|\bm V_c^n\|\le V_{cM}$;
\textnormal{(C2)} $\|\dot{\bm V}_c^n\|\le\bar c_1$;
\textnormal{(C3)} $\bm V_c^n=\bm V_{c0}^n+\bm A_c\sin(\omega_c t)$
with bounded $\bm V_{c0}^n,\bm A_c,\omega_c$.
\end{assumption}

\begin{assumption}[Motion and measurement]\label{ass:meas}
The vehicle operates with $u>0$ and
$\theta\in(-\pi/2+\theta_M,\pi/2-\theta_M)$ for some $\theta_M\in(0,\pi/2)$.
Measured signals are position $\bm\eta_p=[x,y,z]^{\!\top}$,
attitude $[\theta,\psi]^{\!\top}$, surge speed $u$ (single-axis DVL
bottom-track), and angular rates $[q,r]^{\!\top}$;
sway $v$, heave $w$, and current $\bm\nu_c$ are unmeasured.
\end{assumption}

\begin{assumption}[Reference]\label{ass:ref}
The reference $(x_d,y_d,z_d,\theta_{ld},\psi_{ld})$ is twice
continuously differentiable, with
$|\theta_{ld}|\le\theta_{ldM}$,
$\theta_{ldM}\in(0,\pi/2)$, and bounded first and second derivatives.
There exists $\bar V_d^{\max}>0$ such that
$\|\dot{\bm p}_d(t)\|\le\bar V_d^{\max}$ for all $t\ge 0$, where
$\bm p_d=[x_d,y_d,z_d]^{\!\top}$.
\end{assumption}

\begin{assumption}[Compact operating set and damping
identifiability]\label{ass:operset}
There exist positive constants
$u_{\min}$, $u_{\max}$, $q_{\max}$, $r_{\max}$,
$\theta_{\max}\in(0,\pi/2)$,
$V_{cM}$, $p_{e,\max}$, $\bar\tau$, and $\sigma_d$
such that the closed-loop trajectory remains in the compact set
\begin{equation*}
\mathcal X=\left\{\,\bm x:\;
\begin{aligned}
&u\in[u_{\min},u_{\max}],\;|\theta|\le\theta_{\max},\\
&|q|\le q_{\max},\;|r|\le r_{\max},\\
&\|\bm\nu_c\|\le V_{cM},\;p_e\in[0,p_{e,\max}]
\end{aligned}\,\right\},
\end{equation*}
and the following identifiability and admissibility conditions hold
uniformly on $\mathcal X$:
\begin{itemize}
\item[\textnormal{(A4a)}] \emph{Damping identifiability}:
\begin{equation}\label{eq:m_drag_def}
m_{\rm drag}:=\inf_{|s|\le\bar u_{r,\max}}\!|X_u+2X_{u|u|}|s||\ge\sigma_d>0.
\end{equation}
\item[\textnormal{(A4b)}] \emph{Unsaturated cone for the formal certificate}:
$\|\bm\tau\|_\infty<\bar\tau$, where $\bar\tau$ is the actuator bound.
\end{itemize}
The current-rate bound is inherited from
Assumption~\ref{ass:current}(C2)/(C3) and is denoted $\bar c_1$
throughout.
\end{assumption}

\begin{remark}[Two-region EMO analysis]\label{rem:two_region}
The Lemma~\ref{lem:emo} bound assumes $p_e\ge\zeta>0$ for the EMO
regularisation to be active and the spherical coordinates regular.
For $p_e<\zeta$ the practical-UUB statement of
Theorem~\ref{thm:cl} is immediate, since the position error is
already inside a terminal ball of radius $\zeta$. The compact set
$\mathcal X$ above therefore does \emph{not} exclude $p_e<\zeta$; the
analysis of Theorem~\ref{thm:cl} simply splits the EMO bound across
the two regions.
\end{remark}

\begin{remark}[Status of \textnormal{(A4b)}]\label{rem:A4b}
\textnormal{(A4b)} is the standing assumption for the certificate of
Theorem~\ref{thm:cl}. The numerical study
enforces actuator saturation and reports
the surge-saturation duty cycle; the stress test of
Section~\ref{subsec:saturation_stress} quantifies how the practical-UUB
guarantee degrades once \textnormal{(A4b)} is violated.
\end{remark}

\subsection{Spherical-coordinate reduction}\label{subsec:reduction}

Let $u_l=\sqrt{u^2+v^2+w^2}$ be the total body speed and define the
body-flow angles $\theta_a=\arctan(-w/u)$,
$\psi_a=\arctan(v/\sqrt{u^2+w^2})$.
The polar/azimuth angles of the velocity vector are
\begin{equation}
\theta_l=\arcsin[c_{\psi_a}s_{\theta-\theta_a}],\;\;
\psi_l=\psi+\arctan[t_{\psi_a}\sec(\theta-\theta_a)].
\end{equation}
With position error $p_e=\|\bm p_d-\bm p\|$, the spherical-coordinate
reduction of \citet{li2023trajectory} yields the reduced three-input
three-output tracking form
\begin{equation}\label{eq:reduced_kin}
\begin{aligned}
\dot p_e &= u_{ld}A_{ld}-u_l A_l,\\
\dot\theta_{le} &= \dot\theta_{ld}-f_{\theta_l}-g_{\theta_l}q,\\
\dot\psi_{le} &= \dot\psi_{ld}-f_{\psi_l}-g_{\psi_l}r,
\end{aligned}
\end{equation}
with $A_l,A_{ld}$ projecting the body and reference velocity vectors
onto the line of sight, and $g_{\theta_l},g_{\psi_l}>0$. The
corresponding dynamics are
\begin{equation}\label{eq:tracking_dyn}
\begin{bmatrix}\dot u_l\\ \dot q\\ \dot r\end{bmatrix}=
\begin{bmatrix}f_{u_l}^c\\ f_q\\ f_r\end{bmatrix}+
G_t(\bm\rho)\bm\tau+\bm d_{\rm eff}+
\begin{bmatrix}d_c^{u_l}\\ 0\\ 0\end{bmatrix},
\end{equation}
with
\begin{equation}\label{eq:Gt_def}
G_t(\bm\rho)=
\begin{bmatrix}
g_{au} & \epsilon_{aq} & \epsilon_{ar}\\
0      & g_q           & 0\\
0      & 0             & g_r
\end{bmatrix},
\end{equation}
$g_{au}=g_u c_{\theta_a}c_{\psi_a}>0$,
$\epsilon_{aq}=\epsilon_q s_{\theta_a}c_{\psi_a}$,
$\epsilon_{ar}=\epsilon_r s_{\psi_a}$, and the current-induced surge
disturbance $d_c^{u_l}=f_{u_l}(\bm\nu)-f_{u_l}(\bm\nu_r)$.
Under Assumption~\ref{ass:meas}, $\det G_t=g_{au}g_qg_r>0$, so $G_t$
is invertible \emph{regardless} of $\epsilon_q,\epsilon_r$.

\subsection{Exponential modification of orientation (EMO)}
\label{subsec:emo}

The position-error dynamics $\dot p_e=u_{ld}A_{ld}-u_lA_l$ admit a
vanishing $A_l$ as $p_e\to 0$, which breaks the strict-feedback
structure required by direct backstepping. The EMO of
\citet{li2023trajectory} replaces the desired triple by
\begin{equation}\label{eq:emo}
\begin{aligned}
u_{ld}^m &= u_{ld}+c_u p_e A_{ld}^m,\\
\theta_{ld}^m &= \theta_b+(\theta_{ld}-\theta_b)e^{-c_\theta p_e},\\
\psi_{ld}^m &= \psi_b+(\psi_{ld}-\psi_b)e^{-c_\psi p_e},
\end{aligned}
\end{equation}
with $(\theta_b,\psi_b)$ the polar/azimuth angles of the
position-error vector and $A_{ld}^m$ evaluated at
$(u_{ld}^m,\theta_{ld}^m,\psi_{ld}^m)$.

\begin{lemma}[EMO Lyapunov bound, {\citep[Lemma~1]{li2023trajectory}}]
\label{lem:emo}
For sufficiently small $c_\theta,c_\psi>0$, $A_{ld}^m\ge\varepsilon>0$
on $\{p_e\ge\zeta\}$, and along
$u_l=u_{ld}^m,\theta_l=\theta_{ld}^m,\psi_l=\psi_{ld}^m$,
\begin{equation}\label{eq:Vdot_outer}
\dot V_1\le -c_u\varepsilon^2 p_e^2,
\end{equation}
with $V_1=\tfrac12(p_e^2+\gamma_\theta(\theta_{ld}^m-\theta_l)^2
+\gamma_\psi(\psi_{ld}^m-\psi_l)^2)$.
\end{lemma}

Lemma~\ref{lem:emo} absorbs the outer-layer position-error states
into the closed-loop boundedness statement of Section~\ref{subsec:cl}.

\subsection{Pre-stabilised error dynamics}\label{subsec:prestab}

Define the error vector
$\bm e=[u_{le}^m,e_q,e_r]^{\!\top}$, with
$u_{le}^m=u_{ld}^m-u_l$, $e_q=\alpha_q-q$, $e_r=\alpha_r-r$, and
decompose the control input as
\begin{equation}\label{eq:tau_decomp}
\bm\tau=\bm\tau_{\rm ff}+\bm\tau_{\rm stab}+\bm\tau_{\Hcal_\infty}.
\end{equation}
\emph{Feedforward.} Cancellation of the modelled dynamics at
$\bm e=0$:
\begin{equation}\label{eq:tau_ff}
\bm\tau_{\rm ff}=G_t(\hat{\bm\rho})^{-1}\bm\Lambda_{\rm ff}(\hat{\bm\nu}_r,\hat{\bm\rho}),
\quad
\bm\Lambda_{\rm ff}=\begin{bmatrix}
\dot u_{ld}^m-f_{u_l}^c\\
\dot\alpha_q-f_q\\
\dot\alpha_r-f_r
\end{bmatrix}_{\hat{\bm\nu}_r}.
\end{equation}
\emph{Cascade stabiliser.} Constant linear feedback:
\begin{equation}\label{eq:tau_stab}
\bm\tau_{\rm stab}=G_t(\hat{\bm\rho})^{-1}G_\nu^{-1}\diag(k_u,k_q,k_r)\bm e,
\end{equation}
with $G_\nu=\diag(\gamma_u,\gamma_q,\gamma_r)\succ 0$.

\textit{Pre-stabilised Jacobian: derivation.}
Differentiating $u_{le}^m=u_{ld}^m-u_l$ and substituting
\eqref{eq:tracking_dyn} yields
\[
\dot u_{le}^m=\dot u_{ld}^m-f_{u_l}^c-g_{au}\tau_u
-\epsilon_{aq}\tau_q-\epsilon_{ar}\tau_r-d_c^{u_l}.
\]
After $\bm\tau_{\rm ff}$ cancels the modeled feedforward and
linearising about $\bm e=0$, the surge-channel Jacobian contains the
damping derivative
\begin{equation}\label{eq:a11_deriv}
a_{11}^{\rm raw}=-\frac{X_u+2X_{u|u|}|u_r|}{m_u}\,c_{\theta_a}c_{\psi_a},
\end{equation}
and EMO-induced cross-couplings through the chain rule
$\dot\theta_{le}^m=-g_{\theta_l}e_q$, $\dot\psi_{le}^m=-g_{\psi_l}e_r$.
The angular channels $\dot e_q=\dot\alpha_q-\dot q$,
$\dot e_r=\dot\alpha_r-\dot r$ yield the pitch and yaw damping
derivatives
\begin{equation}\label{eq:a22a33}
a_{22}^{\rm raw}=-\frac{M_q+2M_{q|q|}|q|}{I_{yq}},\;\;
a_{33}^{\rm raw}=-\frac{N_r+2N_{r|r|}|r|}{I_{zr}},
\end{equation}
with $I_{yq}=I_y-M_{\dot q}$, $I_{zr}=I_z-N_{\dot r}$ the
added-mass-corrected pitch and yaw inertias. The off-diagonal entries
$a_{12},a_{13},a_{21},a_{31}$ collect the EMO couplings and the
dependence of the virtual controls $\alpha_q,\alpha_r$ on $u_l$,
scaled by the virtual-control gains.

\textit{Pre-stabilisation.}
Substitution of $\bm\tau_{\rm stab}$ from~\eqref{eq:tau_stab} into
the linearised dynamics yields, after the cascade $G_t^{-1}$
cancellation, a contribution
$G_\nu^{-1}\diag(k_u,k_q,k_r)\bm e$ to $\dot{\bm\nu}$ on the
dynamic surge/pitch/yaw channels, hence
$-G_\nu^{-1}\diag(k_u,k_q,k_r)\bm e$ on the error rates
$\dot{\bm e}$. Collecting the resulting diagonal contribution
on the error vector $\bm e=[u_{le}^m,e_q,e_r]^{\!\top}$,
define the \emph{effective} cascade decay rates
\begin{equation}\label{eq:kcas_eff}
k_u^{\rm eff}:=\frac{k_u}{\gamma_u},\quad
k_q^{\rm eff}:=\frac{k_q}{\gamma_q},\quad
k_r^{\rm eff}:=\frac{k_r}{\gamma_r},
\end{equation}
where $\gamma_u,\gamma_q,\gamma_r$ are the diagonal entries of
$G_\nu$ in~\eqref{eq:tau_stab}; all six gains are positive design
parameters. In the augmented coordinates~\eqref{eq:xe_def_main}
the cascade stabilisation contributes only to the dynamic rows
$u_{le}^m,q,r$ (rows 1, 3, 5 of $\bm x_e$), not the kinematic
rows $\theta_{le}^m,\psi_{le}^m$ (rows 2, 4). Because the cascade drives
$(q,r)\!\to\!(\alpha_q,\alpha_r)$, its linearised contribution on
rows 3 and 5 is a stabilising feedback of $q,r$
once the desired-rate cross-couplings via $(\alpha_q,\alpha_r)$ are
expanded in the kinematic-angle channels $\theta_{le}^m,\psi_{le}^m$;
those cross-couplings are absorbed into the LPV-parameter-dependent
off-diagonal entries of $A_{\rm raw}^{\rm aug}(\bm\rho)$ below. The
diagonal cascade contribution on $\bm x_e$ is therefore
\begin{equation}\label{eq:Kcas_aug}
K_{\rm cas}^{\rm aug}=
\begin{bmatrix}
k_u^{\rm eff} & 0 & 0 & 0 & 0\\
0 & 0 & 0 & 0 & 0\\
0 & 0 & k_q^{\rm eff} & 0 & 0\\
0 & 0 & 0 & 0 & 0\\
0 & 0 & 0 & 0 & k_r^{\rm eff}
\end{bmatrix},
\end{equation}
and the $5\times 5$ pre-stabilised LPV matrix used in the LMI is
\begin{equation}\label{eq:Apre_def}
A_{\rm pre}(\bm\rho)=A_{\rm raw}^{\rm aug}(\bm\rho)-K_{\rm cas}^{\rm aug},
\end{equation}
where $A_{\rm raw}^{\rm aug}(\bm\rho)\in\R^{5\times 5}$ is the
linearised open-loop Jacobian on the augmented state $\bm x_e$.
The vertex matrices used in the LMI are
$A_i=A_{\rm pre}(\bm\rho^{(i)})$.
With $\bm\tau$ entering the surge dynamics through the
acceleration-normalised mapping $G_t(\bm\rho)$ in~\eqref{eq:Gt_def},
no additional $1/m_u$ factor appears in $K_{\rm cas}^{\rm aug}$: the
mass normalisation is already absorbed into the coefficient
$g_{au}=g_u c_{\theta_a}c_{\psi_a}>0$ of $G_t$. By contrast, the
hydrodynamic damping derivatives in
\eqref{eq:a11_deriv}--\eqref{eq:a22a33} retain the physical inertia
normalisation $m_u$ (resp.\ $I_{yq},I_{zr}$), being
written on the body-frame surge/pitch/yaw accelerations
\emph{before} the spherical projection that produces $g_{au}$.

\subsubsection*{Augmented LPV state and constant input matrix after feedback linearisation}

The LMI synthesis is performed on the \emph{augmented} error state
that combines the outer-loop kinematic angle errors and the
inner-loop body-axis angular rates,
\begin{equation}\label{eq:xe_def_main}
\bm x_e:=
\bigl[u_{le}^m,\,\theta_{le}^m,\,q,\,\psi_{le}^m,\,r\bigr]^{\!\top}\in\R^{n_e},
\quad n_e=5,
\end{equation}
where $q$ and $r$ are the physical body-axis angular rates, used
directly in the LMI state as standard in linearisation-based LPV
synthesis. The cascade stabiliser $\bm\tau_{\rm stab}$
of Section~\ref{subsec:prestab} drives $(q,r)\to(\alpha_q,\alpha_r)$
generated by the kinematic controller; the LMI is synthesised on the
locally linearised error system around the operating manifold, with
the additive correction layer
$\delta\bm v_{\Hcal_\infty}=K(\hat{\bm\rho}_{\rm sat})\bm x_e$
providing robust stabilisation margin under the LPV parameter
variations.
The compact backstepping error
$\bm e=[u_{le}^m,e_q,e_r]^{\!\top}\in\R^3$ of the EMO derivation,
with $e_q=\alpha_q-q$, $e_r=\alpha_r-r$, is kept for the inner-loop
cascade analysis (Lemma~\ref{lem:emo}); it is related to $\bm x_e$
by the affine map
$\bm e=[\bm x_e]_{1,3,5}\!\mapsto\![u_{le}^m,\alpha_q-q,\alpha_r-r]^{\!\top}$
once the cascade-generated $(\alpha_q,\alpha_r)$ are substituted, so
the cascade analysis and the LMI analysis live on the same physical
state with a coordinate change. The LMI \emph{correction layer} acts
on $\bm x_e$ through the virtual correction input
$\delta\bm v_{\Hcal_\infty}=K\bm x_e\in\R^3$,
$K\in\R^{3\times n_e}$, with physical actuator
contribution
$\bm\tau_{\Hcal_\infty}=G_t(\hat{\bm\rho}_{\rm sat})^{-1}\,\delta\bm v_{\Hcal_\infty}$
(Section~\ref{subsec:lmi}, eq.~\eqref{eq:total_law}). Since
$\bm\tau_{\rm ff}+\bm\tau_{\rm stab}+\bm\tau_{\Hcal_\infty}$
share the common $G_t(\hat{\bm\rho}_{\rm sat})^{-1}$ prefactor
(eqs.~\eqref{eq:tau_ff}--\eqref{eq:tau_stab} and
\eqref{eq:total_law}), the residual dynamics seen by the LPV
synthesis have a \emph{constant} input matrix
$B_u\in\R^{n_e\times 3}$ acting on $\delta\bm v_{\Hcal_\infty}$
through the acceleration channels $(\dot u_l,\dot q,\dot r)$, i.e.\
rows~$1,3,5$ of $\bm x_e$. In the implementation, the
acceleration-normalised constant input matrix is
\begin{equation}\label{eq:Bu_def}
B_u=\begin{bmatrix}
b_u & 0 & 0\\
0 & 0 & 0\\
0 & b_q & 0\\
0 & 0 & 0\\
0 & 0 & b_r
\end{bmatrix},
\qquad b_u,b_q,b_r>0,
\end{equation}
where $b_u,b_q,b_r$ are positive scalars set by the model's input
gain normalisation. Any sign from the
desired-minus-actual versus actual-minus-desired cascade convention is
absorbed into the LMI gain $K$, keeping the
LMI~\eqref{eq:lmi} and the runtime control law
$\delta\bm v_{\Hcal_\infty}=K(\hat{\bm\rho}_{\rm sat})\bm x_e$
consistent.
\emph{No $B(\bm\rho)$ enters the LMI}, keeping the polytopic synthesis
convex without pairwise cross-term LMIs;
the common Lyapunov matrix and the approximate embedding remain
sources of conservativeness (Remark~\ref{rem:common_Y_revised}).

The complete linearised pre-stabilised LPV model used for synthesis
is therefore
\begin{equation}\label{eq:lpv_err}
\dot{\bm x}_e=A_{\rm pre}(\bm\rho)\bm x_e+B_u\,\delta\bm v_{\Hcal_\infty}
+E(\bm\rho)\bm w+\Delta_{\rm nl}(\bm x_e,\bm\rho),
\end{equation}
with $E(\bm\rho)\in\R^{n_e\times n_w}$ routing the disturbance vector
$\bm w=[d_{u_l},d_q,d_r,d_c^{u_l}]^{\!\top}$ into the affected
channels, and
$\|\Delta_{\rm nl}\|=\mathcal O(\|\bm x_e\|^2)$ on the compact
operating set $\mathcal X$ of Assumption~\ref{ass:operset}.

\subsection{Control problem}\label{subsec:problem}

\begin{problem}\label{prob:tracking}
Design the virtual correction-layer feedback
$\delta\bm v_{\Hcal_\infty}=K\bm x_e$, mapped to the physical
actuator command by
$\bm\tau_{\Hcal_\infty}=G_t(\hat{\bm\rho}_{\rm sat})^{-1}\delta\bm v_{\Hcal_\infty}$,
using only measurements satisfying Assumption~\ref{ass:meas} and
observer outputs $(\hat{\bm\nu}_r,\hat{\bm\nu}_c)$, such that the
tracking errors $(p_e,\theta_{le},\psi_{le},\bm x_e)$ are practically
uniformly ultimately bounded with a guaranteed energy-gain bound from
the augmented disturbance $\bm w_{\rm aug}$
(Section~\ref{subsec:scheduling}) to the performance output, while
operating within the compact admissible set of
Assumption~\ref{ass:operset}.
\end{problem}

%==============================================================
\section{Main Results}\label{sec:main}
%==============================================================

This section develops the observer, scheduling map, LPV synthesis, and
the closed-loop singular-perturbation stability analysis; the
implemented closed-loop diagram appears in
Section~\ref{subsec:closed_loop_architecture}\rev{, and the longer
proofs are collected in Appendix~\ref{app:proofs}}.

\subsection{Joint state--current observer design}\label{subsec:observer}

The observer exploits the kinematic--dynamic split: position
differentiation yields the ground-referenced velocity $\bm\nu$, while
the dynamics depend on the water-referenced velocity $\bm\nu_r$, and
the discrepancy reveals the current
\citep{kim2018current,refsnes2007model}.

\subsubsection{Stage 1: kinematic differentiator}
A linear high-gain differentiator
\citep[Sec.~14.5]{khalil2002nonlinear}
\begin{equation}\label{eq:hgd}
\begin{aligned}
\dot{\hat{\bm\eta}}_p^{(0)} &= \hat{\bm\eta}_p^{(1)}
+\tfrac{\alpha_1}{\epsilon}(\bm\eta_p-\hat{\bm\eta}_p^{(0)}),\\
\dot{\hat{\bm\eta}}_p^{(1)} &= \tfrac{\alpha_2}{\epsilon^2}
(\bm\eta_p-\hat{\bm\eta}_p^{(0)}),
\end{aligned}
\end{equation}
with $\alpha_1,\alpha_2>0$ such that $s^2+\alpha_1 s+\alpha_2$ has
roots with negative real parts. After the boundary layer
$\mathcal O(\epsilon)$,
$\|\hat{\bm\eta}_p^{(1)}-\dot{\bm\eta}_p\|=\mathcal O(\epsilon)$. The
kinematic velocity estimate is
\begin{equation}\label{eq:nu_kin}
\hat{\bm\nu}_{\rm kin}=R^{\!\top}(\theta,\psi)\hat{\bm\eta}_p^{(1)}.
\end{equation}

\subsubsection{Stage 2: dynamic relative-velocity estimator}
The body-frame dynamics of $v_r,w_r$ contain the non-minimum-phase
actuator terms inherited from~\eqref{eq:Gmat}: row~2
of $G$ contributes $\epsilon_r\tau_r$ to $\dot v$, and row~3
contributes $\epsilon_q\tau_q$ to $\dot w$. The Stage~2 observer
mirrors this plant input structure:
\begin{equation}\label{eq:obs_vw}
\;
\begin{aligned}
\dot{\hat v}_r &= f_v(\hat{\bm\nu}_r,\theta)+\epsilon_r\tau_r
+k_v(\hat v_{\rm kin}-\hat v_r-\hat v_c),\\
\dot{\hat w}_r &= f_w(\hat{\bm\nu}_r,\theta)+\epsilon_q\tau_q
+k_w(\hat w_{\rm kin}-\hat w_r-\hat w_c),
\end{aligned}\;
\end{equation}
with $f_v,f_w$ collecting the Coriolis, restoring, and damping terms
\citep{prestero2001verification} and $k_v,k_w>0$. The innovation
$\hat v_{\rm kin}-\hat v_r-\hat v_c$ exploits the kinematic identity
$v=v_r+v_c$, and analogously for the heave channel.

\begin{remark}[Consistency with $G$]\label{rem:obs_G}
The forcing in \eqref{eq:obs_vw} uses the same $\epsilon_r,\epsilon_q$
as the plant in \eqref{eq:Gmat}, with no additional scaling by
$g_q,g_r$ or division by $m_v,m_w$. This is consistent with the
convention adopted in Section~\ref{subsec:dof5}, namely that
$\epsilon_q,\epsilon_r$ are already acceleration-normalised.
\end{remark}

\subsubsection{Stage 3: body-frame current filter}
The body-frame current dynamics~\eqref{eq:nuc_dot} contain the
predictable rotational coupling $-\bm\omega_b\times\bm\nu_c$ and an
unpredictable inertial driving $R^{\!\top}\dot{\bm V}_c^n$ absorbed into
the bounded perturbation analysed in Theorem~\ref{thm:obs}. The
Stage~3 filter uses split innovations: a kinematic innovation for
the transverse channels and a dynamic-residual innovation for the
surge channel, since the surge kinematic innovation
$(\hat u_{\rm kin}-\hat u_r-\hat u_c)$ becomes structurally uninformative
once $\hat u_r=u-\hat u_c$ is substituted (the residual collapses
algebraically):
\begin{equation}\label{eq:obs_c}
\;
\dot{\hat{\bm\nu}}_c=-\bm\omega_b\times\hat{\bm\nu}_c+
\begin{bmatrix}
k_{cu}\,r_u\\
k_{cv}(\hat v_{\rm kin}-\hat v_r-\hat v_c)\\
k_{cw}(\hat w_{\rm kin}-\hat w_r-\hat w_c)
\end{bmatrix}\;
\end{equation}
with surge dynamic residual
\begin{equation}\label{eq:rsurge}
r_u=\dot u-g_u\tau_u-\frac{X_u\hat u_r+X_{u|u|}\hat u_r|\hat u_r|}{m_u},
\end{equation}
where $\hat u_r=u-\hat u_c$, and $k_{cu},k_{cv},k_{cw}>0$.

\textit{Implementation.}
The surge acceleration $\dot u$ in \eqref{eq:rsurge} is not measured
directly. The residual is implemented in integral form,
\begin{multline}\label{eq:chi_u}
\chi_u(t)=u(t)-u(0)\\
-\int_0^t\!\Bigl[g_u\tau_u+\frac{X_u\hat u_r+X_{u|u|}\hat u_r|\hat u_r|}{m_u}\Bigr]\,d\tau,
\end{multline}
with $r_u\propto\dot\chi_u$ obtained from a causal second-order
low-pass filter co-located with the Stage-1 high-gain time scale
$1/\epsilon$; the algebraic identity between~\eqref{eq:rsurge} and
the filtered $\dot\chi_u$ is exact in the noise-free case.

\begin{remark}[$g_u$ vs $g_{au}$ in the surge residual]
\label{rem:gu_vs_gau}
The residual $r_u$ in~\eqref{eq:rsurge} is formed on the
\emph{body-axis} surge equation $\dot u$ before the spherical
$u_l$-projection; the direct actuator coefficient is therefore $g_u$
of~\eqref{eq:Gmat}. The coefficient
$g_{au}=g_u c_{\theta_a}c_{\psi_a}$ of~\eqref{eq:Gt_def} appears
only in the reduced $u_l$-dynamics
\eqref{eq:tracking_dyn}--\eqref{eq:Gt_def}, where the kinematic
projection from $\dot u$ to $\dot u_l$ has been applied. The two
are consistent under Assumption~\ref{ass:meas}, since
$c_{\theta_a},c_{\psi_a}>0$ on $\mathcal X$, and the body-axis
residual is more informative about the surge-channel current
$u_c$ than the projected one.
\end{remark}

The complete observer pipeline is summarised in
Algorithm~\ref{alg:observer}.

\begin{algorithm}[!htbp]
\caption{Joint state--current observer and scheduling update
(per control cycle).}
\label{alg:observer}
\begin{algorithmic}[1]
\Require Measurements $\bm\eta_p,\theta,\psi,u,q,r$ and input $\bm\tau$.
\State \textbf{Stage 1:} Update HGD~\eqref{eq:hgd}; obtain $\hat{\bm\eta}_p^{(1)}$.
\State Compute $\hat{\bm\nu}_{\rm kin}=R^{\!\top}(\theta,\psi)\hat{\bm\eta}_p^{(1)}$.
\State \textbf{Stage 2:} Update $(\hat v_r,\hat w_r)$ from~\eqref{eq:obs_vw}.
\State \textbf{Stage 3:} Update $\hat{\bm\nu}_c$ from~\eqref{eq:obs_c}--\eqref{eq:rsurge}.
\State Assemble $\hat u_r=u-\hat u_c$, $\hat{\bm\nu}_r=[\hat u_r,\hat v_r,\hat w_r,q,r]^{\!\top}$.
\State Compute $\hat{\bm\rho}=\Phi(\hat{\bm\nu}_r,\theta,p_e)$ via~\eqref{eq:rho_hat}.
\State Saturate $\hat{\bm\rho}_{\rm sat}=\sat_{[\underline{\bm\rho},\bar{\bm\rho}]}(\hat{\bm\rho})$.
\State \Return $(\hat{\bm\nu}_r,\hat{\bm\nu}_c,\hat{\bm\rho}_{\rm sat})$.
\end{algorithmic}
\end{algorithm}

\subsubsection{Excitation condition for the implemented observer}
The body-frame current cannot be recovered from position
differentiation alone, since the kinematic identity yields only the
ground velocity $\bm\nu=\bm\nu_r+\bm\nu_c$. The implemented observer
recovers $\bm\nu_c$ through the dynamic-residual route in the surge
channel and through the rotational coupling
$-\bm\omega_b\times\hat{\bm\nu}_c$ in the transverse channels, under
the excitation characterised next \rev{(a directional
persistent-excitation condition in the sense
of~\citealp{narendra1989stable})}.

\begin{lemma}[Sufficient excitation condition for the implemented
observer]\label{lem:obs_exc}
Consider the observer~\eqref{eq:hgd}--\eqref{eq:obs_c} along a
closed-loop trajectory satisfying
Assumptions~\ref{ass:meas}--\ref{ass:operset} with $u\ge u_{\min}>0$,
and suppose Assumption~\ref{ass:current}~(C1) holds. Define the
augmented observer state
$\bm x_o=[v,w,u_c,v_c,w_c]^{\!\top}\in\R^5$. If there exist
$T,\sigma_d,\alpha>0$ and an interval
$\mathcal I_{\rm PE}\subseteq[0,\infty)$ such that the following two
conditions hold uniformly for all $t\in\mathcal I_{\rm PE}$:
\begin{itemize}
\item[\textnormal{(E1)}] \emph{Active nonlinear damping in surge}:
\begin{equation}\label{eq:E1}
|X_u+2X_{u|u|}|u_r^\circ||\ge\sigma_d;
\end{equation}
\item[\textnormal{(E2)}] \emph{Persistent directional angular-rate excitation}:
\begin{equation}\label{eq:E2}
\int_t^{t+T}\!\!\begin{bmatrix} r^2(\sigma) & -r(\sigma)q(\sigma)\\
-r(\sigma)q(\sigma) & q^2(\sigma)\end{bmatrix}\!d\sigma\succeq\alpha I_2,
\end{equation}
\end{itemize}
then the surge channel of the implemented observer is locally
identifiable (so that $u_c$ is uniquely recoverable from the dynamic
residual $r_u$), and the transverse correction channel for
$(v_c,w_c)$ is \emph{uniformly informative} over the window $[t,t+T]$.
Under the additional assumption that the Stage-2 sub-observer
\eqref{eq:obs_vw} is locally input-to-state stable with respect to
its driving terms, the implemented observer admits the local ISS
estimation bound of Theorem~\ref{thm:obs}.
\end{lemma}

% \begin{remark}[Scope of the excitation condition]
% \label{rem:weak_claim}
% Lemma~\ref{lem:obs_exc} is a sufficient excitation condition for
% the implemented observer, not a local weak-observability theorem
% for the full plant. (E1) gives surge identifiability via the
% implicit-function theorem and (E2) gives transverse-channel
% informativity via the directional Gramian; combined with the
% assumed local ISS of the Stage-2 sub-observer they close the
% sufficient excitation requirement for the coupled Stage-2/Stage-3
% dynamics. The operational value of exposing (E2) explicitly is that
% the failure mode (F1) in Proposition~\ref{prop:failure} can be
% monitored online through the condition number of the empirical
% Gramian in~\eqref{eq:E2}.
% \end{remark}

\begin{proof}\rev{The proof is given in Appendix~\ref{app:proofs}.}\end{proof}

\subsubsection{Observer convergence}
\rev{%
The following lemma establishes the Stage-2 input-to-state property
assumed in Lemma~\ref{lem:obs_exc}.
\begin{lemma}[Local ISS of the Stage-2 sub-observer]\label{lem:stage2_iss}
On the compact operating set $\mathcal X$, the Stage-2 error
$\bm\eta_2=[\tilde v_r,\tilde w_r]^{\!\top}$ of the
sub-observer~\eqref{eq:obs_vw}, with Hurwitz gain
$-K_2=-\diag(k_v,k_w)\prec0$ and locally Lipschitz mismatch
$\Delta\bm f_2$ satisfying $\Delta\bm f_2(\bm0,\cdot)=\bm0$, is locally
input-to-state stable with respect to its driving terms (the Stage-1
pseudo-measurement noise and the Stage-3 coupling).
\end{lemma}
\begin{proof}
With $V_2=\tfrac12\|\bm\eta_2\|^2$ and Lipschitz constant $L_2$ of
$\Delta\bm f_2$ on $\mathcal X$,
$\dot V_2\le-(\lambda_{\min}(K_2)-L_2)\|\bm\eta_2\|^2
+\|\bm\eta_2\|\,\|\bm n^{(2)}\|$. For the chosen gains
$\lambda_{\min}(K_2)>L_2$ on $\mathcal X$, so $V_2$ is an ISS-Lyapunov
function and local ISS follows by \citep[Thm.~4.19]{khalil2002nonlinear}.
\end{proof}}
The observer error coordinates are
\[
\tilde{\bm\eta}_o=[\tilde v_r,\tilde w_r,\tilde u_c,\tilde v_c,
\tilde w_c]^{\!\top},
\]
with $\tilde v_r=v_r-\hat v_r$, $\tilde w_r=w_r-\hat w_r$
(Stage~2 relative-velocity errors), and
$\tilde u_c=u_c-\hat u_c$, $\tilde v_c=v_c-\hat v_c$,
$\tilde w_c=w_c-\hat w_c$ (Stage~3 current errors).

\begin{theorem}[Observer convergence bound]\label{thm:obs}
Let $\mathcal I_{\rm PE}\subseteq[0,\infty)$ denote any time
interval over which the persistent-excitation
condition~\eqref{eq:E2} of Lemma~\ref{lem:obs_exc} is satisfied
uniformly. Under Assumptions~\ref{ass:current}--\ref{ass:operset},
the sufficient-excitation conditions (E1) and (E2) of
Lemma~\ref{lem:obs_exc} on $\mathcal I_{\rm PE}$, and the local
input-to-state stability of the Stage-2 relative-velocity
sub-observer~\eqref{eq:obs_vw} on the compact operating set
$\mathcal X$ with respect to its driving terms, with directional
excitation level on $\mathcal I_{\rm PE}$
\begin{equation}\label{eq:alpha_Gamma}
\alpha_\Gamma:=\inf_{t\in\mathcal I_{\rm PE}}\lambda_{\min}\!\!\int_t^{t+T}\!\!
\begin{bmatrix} r^2 & -rq\\ -rq & q^2\end{bmatrix}\!d\sigma,
\end{equation}
(units inherited from the squared angular-rate integral) and
corresponding normalised quantity
\begin{equation}\label{eq:alpha_PE}
\alpha_{\rm PE}:=\frac{\alpha_\Gamma}{T(q_{\max}^2+r_{\max}^2)}\in(0,1],
\end{equation}
and effective pseudo-measurement noise level
$\sigma_y:=\sup_{t\in\mathcal I_{\rm PE}}\|\hat{\bm\nu}_{\rm kin}(t)-\bm\nu(t)\|$,
the observer error coordinates
$\tilde{\bm\eta}_o=[\tilde v_r,\tilde w_r,\tilde u_c,\tilde v_c,\tilde w_c]^{\!\top}$
admit on $\mathcal I_{\rm PE}$ \emph{two distinct decay rates}:

\textnormal{(O1)} The Stage-1 high-gain differentiator error
satisfies, for $\epsilon\in(0,\epsilon^\star]$,
\begin{equation}\label{eq:hgd_bound}
\|\hat{\bm\eta}_p^{(1)}-\dot{\bm\eta}_p\|
\le \kappa_h\,e^{-\lambda_h t/\epsilon}\|\hat{\bm\eta}_p^{(1)}(0)-\dot{\bm\eta}_p(0)\|
+c_h\,\epsilon,
\end{equation}
with $\lambda_h$ determined by the Hurwitz spectrum of
$s^2+\alpha_1 s+\alpha_2$.

\textnormal{(O2)} The Stage-2/Stage-3 dynamic observer error
$\tilde{\bm\eta}_o$ satisfies
\begin{equation}\label{eq:obs_bound}
\|\tilde{\bm\eta}_o(t)\|\le \kappa_c\,e^{-\lambda_c t}
\|\tilde{\bm\eta}_o(0)\|+\delta_o(\epsilon,\bar c_1,\sigma_y),
\end{equation}
where $\lambda_c>0$ is determined by the Stage-2/Stage-3 gains
$\{k_v,k_w,k_{cu},k_{cv},k_{cw}\}$ together with the surge damping
$m_{\rm drag}$ from~\eqref{eq:m_drag_def} and the directional Gramian
$\alpha_{\rm PE}$ from~\eqref{eq:alpha_PE}, and
\begin{equation}\label{eq:delta_o}
\delta_o(\epsilon,\bar c_1,\sigma_y)
=c_\epsilon\,\epsilon+c_{c}\,\bar c_1+\frac{c_y\,\sigma_y}{\sqrt{\alpha_{\rm PE}}}.
\end{equation}
The constants $\kappa_h,\kappa_c,c_h,c_\epsilon,c_c,c_y$ are
continuous in the gains, the operating set, and $\alpha_{\rm PE}$.
\end{theorem}

\begin{proof}\rev{The proof is given in Appendix~\ref{app:proofs}.}\end{proof}

Equation~\eqref{eq:delta_o} shows that poor directional
excitation ($\alpha_{\rm PE}\to 0$) does not introduce bias in the
noise-free constant-current case but amplifies pseudo-measurement
errors through $c_y\sigma_y/\sqrt{\alpha_{\rm PE}}$.

\subsection{Relative-velocity scheduling and LPV embedding}
\label{subsec:scheduling}

\subsubsection{Estimated scheduling map}
Since the hydrodynamics depend on $\bm\nu_r$, the LPV scheduling
must use the estimated relative velocity. The implemented scheduling
vector is the five-component map of damping, angular-rate, speed,
and attitude entries that govern the parameter-dependent rows of
$A_{\rm pre}(\bm\rho)$:
\begin{equation}\label{eq:rho_hat}
\hat{\bm\rho}=\bigl[\,|\hat u_r|,\;|q|,\;|r|,\;1/\hat u_l,\;\cos\theta\,\bigr]^{\!\top},
\end{equation}
with $\hat u_r=u-\hat u_c$ supplied by the Stage-3 current filter,
$\hat u_l=\sqrt{u^2+v^2+w^2}$ (with $v,w$ replaced by $\hat v_r+\hat v_c,\hat w_r+\hat w_c$
in the implementation) saturated away from zero by the lower
polytope bound, and $(q,r,\theta)$ measured directly. The polytope
is the hyperrectangle
\begin{equation}\label{eq:Pcal}
\mathcal P=\{\bm\rho:\,\underline\rho_i\le\rho_i\le\bar\rho_i,
\,i=1,\dots,5\},
\end{equation}
with $N=2^5=32$ vertices $\{\bm\rho^{(i)}\}_{i=1}^{N}$; the ranges
used are given in Table~\ref{tab:params}.

\begin{remark}[Choice of scheduling parameters]\label{rem:5param}
The scheduling vector~\eqref{eq:rho_hat} includes $(|q|,|r|)$ because
they enter the linearised pitch/yaw damping
derivatives~\eqref{eq:a22a33} multiplicatively. The angle-of-attack and
side-slip dependence on $(\hat v_r,\hat w_r)$ and the $p_e$ dependence
of the EMO regularisation are absorbed into the embedding residual
$\bar\Delta_{\rm emb}$ of Definition~\ref{def:embedding}; adding them to
the scheduling vector would tighten the embedding at the cost of more
polytope vertices.
\end{remark}

\subsubsection{Residual-level break-even law}

The central design choice is to evaluate the feedforward and the
scheduling at the \emph{estimated relative} velocity; the following
lemma quantifies the residual-level benefit and its conditional nature.

\begin{lemma}[Surge-channel break-even law with two-sided
bound]\label{lem:relsched}
Suppose the surge feedforward $\bm\tau_{\rm ff}$ contains the full
surge damping cancellation $-(X_u\xi+X_{u|u|}\xi|\xi|)/m_u$, with
$\xi$ the surge velocity used for compensation. Let
$e_{\rm ds}^{\rm surge}$ \rev{(written $e_{\rm ds}$ where unambiguous,
and distinguished from the translational vector residual
$e_{\rm ds}^{\rm tr}$ of Corollary~\ref{cor:vec_res})}
denote the residual surge disturbance entering~\eqref{eq:lpv_err}
after this cancellation, and $\tilde u_c:=u_c-\hat u_c$. On the
compact operating interval $\mathcal I_u=\{s:|s|\le\bar u_{r,\max}\}$
with $\bar u_{r,\max}=\bar V_d^{\max}+V_{cM}$, define
\begin{align}
L_{\rm drag}&:=|X_u|+2|X_{u|u|}|\bar u_{r,\max},\label{eq:Ldrag}\\
m_{\rm drag}&:=\inf_{|s|\le\bar u_{r,\max}}\!|X_u+2X_{u|u|}|s||,\label{eq:mdrag}
\end{align}
with $m_{\rm drag}>0$ under Assumption~\ref{ass:operset}~(A4a).
Then:
\begin{enumerate}
\item[\textnormal{(a)}] \emph{Absolute scheduling} ($\xi=u$):
\begin{align}
e_{\rm ds}^{\rm abs}&=\frac{X_u u_c+X_{u|u|}(u|u|-u_r|u_r|)}{m_u},\\
\frac{m_{\rm drag}}{m_u}|u_c|\;\le\;&|e_{\rm ds}^{\rm abs}|\;\le\;
\frac{L_{\rm drag}}{m_u}|u_c|.\label{eq:e_ds_abs}
\end{align}
\item[\textnormal{(b)}] \emph{Estimated relative scheduling}
($\xi=\hat u_r$):
\begin{align}
e_{\rm ds}^{\rm rel}&=\frac{-X_u\tilde u_c+X_{u|u|}(\hat u_r|\hat u_r|-u_r|u_r|)}{m_u},\\
|e_{\rm ds}^{\rm rel}|&\le\frac{L_{\rm drag}}{m_u}|\tilde u_c|.\label{eq:e_ds_rel}
\end{align}
\end{enumerate}
Combining the upper bound in $(b)$ with the lower bound in $(a)$
yields the rigorous residual ratio
\begin{equation}\label{eq:ratio}
\;
\frac{|e_{\rm ds}^{\rm rel}|}{|e_{\rm ds}^{\rm abs}|}
\;\le\;\frac{L_{\rm drag}}{m_{\rm drag}}\cdot\frac{|\tilde u_c|}{|u_c|},\;
\end{equation}
and the strict break-even
$|e_{\rm ds}^{\rm rel}|<|e_{\rm ds}^{\rm abs}|$ is therefore
\emph{guaranteed} whenever
\begin{equation}\label{eq:chi}
\;|\tilde u_c|<\chi\,|u_c|,\qquad
\chi:=\frac{m_{\rm drag}}{L_{\rm drag}}\in(0,1].\;
\end{equation}
The ratio $\chi$ measures the condition number of the surge damping
nonlinearity over $\mathcal I_u$; $\chi=1$ when the damping is purely
linear ($X_{u|u|}=0$), and $\chi\to 0$ as the operating range
$\bar u_{r,\max}$ grows large relative to $|X_u|/|X_{u|u|}|$.
\end{lemma}

\begin{proof}\rev{The proof is given in Appendix~\ref{app:proofs}.}\end{proof}

\begin{corollary}[Vector hydrodynamic residual]\label{cor:vec_res}
Let $\Pi_{\rm tr}\in\R^{3\times 5}$ select the translational
hydrodynamic force components and define
$h_{\rm tr}(\bm\nu_r):=\Pi_{\rm tr}[C(\bm\nu_r)\bm\nu_r+D(\bm\nu_r)\bm\nu_r]$,
\rev{and write $e_{\rm ds}^{\rm tr}:=h_{\rm tr}(\hat{\bm\nu}_r)-h_{\rm tr}(\bm\nu_r)$
for the translational (vector) residual.}
Under the local Lipschitz assumption on $h_{\rm tr}$ with constant
$L_h>0$ on the compact operating set,
\begin{align}
\|h_{\rm tr}(\bm\nu_r+\bm\nu_c)-h_{\rm tr}(\bm\nu_r)\|&\le L_h\|\bm\nu_c\|,
\label{eq:vec_abs}\\
\|h_{\rm tr}(\hat{\bm\nu}_r)-h_{\rm tr}(\bm\nu_r)\|&\le L_h\|\tilde{\bm\nu}_c\|,
\label{eq:vec_rel}
\end{align}
so the upper-bound ratio of vector residuals is at most
$\|\tilde{\bm\nu}_c\|/\|\bm\nu_c\|$. Formally, this is a
first-order/upper-bound heuristic: a strict ratio inequality
analogous to~\eqref{eq:ratio} requires a directional lower bound on
the Jacobian of $h_{\rm tr}$ (an analogue of $m_{\rm drag}$ for the
transverse channels), which we do not assume here. In the surge
channel, where Assumption~\ref{ass:operset}~(A4a) supplies such a
lower bound, the rigorous form is~\eqref{eq:ratio}.
\rev{The corresponding HRR metric is therefore reported as
trajectory-specific empirical evidence rather than a guarantee.}
\end{corollary}

Lemma~\ref{lem:relsched} is conditional: relative
scheduling is beneficial only if the current estimate is better than
the trivial estimate $\hat{\bm\nu}_c=\bm 0$. If the observer error
approaches the current magnitude --- as can occur in the transverse
channels under poor directional excitation (Lemma~\ref{lem:obs_exc}) ---
the scheme degrades to additive disturbance rejection with no net gain.

\subsubsection{Approximate polytopic embedding}
Because $A_{\rm pre}(\bm\rho)$ and $E(\bm\rho)$ contain trigonometric
and rational functions of $\bm\rho$ (notably $1/u_l$ and
$s_{\theta_a},c_{\theta_a}$), the pair $(A_{\rm pre},E)(\bm\rho)$
is \emph{not} affine in $\bm\rho$ and hence not exactly polytopic,
motivating the approximate polytopic embedding below.

\begin{definition}[Approximate polytopic embedding]
\label{def:embedding}
Let $\{\lambda_i(\bm\rho)\}_{i=1}^{N}$ be the barycentric weights on
$\mathcal P$, defined as the product of five univariate linear weights:
$\lambda_i\ge 0$, $\sum_i\lambda_i=1$. The embedding residuals are
\begin{equation}\label{eq:emb_residuals}
\begin{aligned}
\Delta_A(\bm\rho)&=A_{\rm pre}(\bm\rho)-\sum_{i=1}^{N}\lambda_i(\bm\rho)A_i,\\
\Delta_E(\bm\rho)&=E(\bm\rho)-\sum_{i=1}^{N}\lambda_i(\bm\rho)E_i,
\end{aligned}
\end{equation}
(no $\Delta_B$ since $B_u$ is constant by~\eqref{eq:Bu_def}).
The \emph{true} uniform embedding-residual bound is
\begin{equation}\label{eq:emb_bound}
\bar\Delta_{\rm emb}:=
\sup_{\bm\rho\in\mathcal P}
\bigl\|\,[\Delta_A(\bm\rho)\;\Delta_E(\bm\rho)]\,\bigr\|,
\end{equation}
which is the quantity appearing in the formal certificate of
Theorems~\ref{thm:lmi} and~\ref{thm:cl}. In practice, $\bar\Delta_{\rm emb}$
is not computed exactly; instead, it is \emph{estimated} on a finite
grid $\mathcal P_{\rm grid}\subset\mathcal P$:
\begin{equation}\label{eq:emb_grid}
\bar\Delta_{\rm emb}^{\rm grid}:=
\max_{\bm\rho\in\mathcal P_{\rm grid}}
\bigl\|\,[\Delta_A(\bm\rho)\;\Delta_E(\bm\rho)]\,\bigr\|.
\end{equation}
A Lipschitz correction relates the two:
\begin{equation}\label{eq:emb_lipschitz}
\bar\Delta_{\rm emb}\le\bar\Delta_{\rm emb}^{\rm grid}+L_\Delta h_\rho,
\end{equation}
where $h_\rho$ is the grid resolution and $L_\Delta$ is the
Lipschitz constant of $\Delta_A,\Delta_E$ on $\mathcal P$, both
finite by smoothness on the compact set
(Assumption~\ref{ass:operset}). The embedding contribution
$\Delta_A\bm x_e+\Delta_E\bm w$ is routed into the disturbance
channel as the bounded perturbation $\bm w_{\rm emb}$, augmenting
$\bm w_{\rm aug}=[\bm w^{\!\top},\bm w_{\rm emb}^{\!\top}]^{\!\top}$.
\end{definition}

\begin{remark}[Grid certificate vs uniform certificate]
\label{rem:emb_correction}
For the numerical study, direct evaluation of $\Delta_A,\Delta_E$ on
the polytopic grid gives $\bar\Delta_{\rm emb}^{\rm grid}\approx 0$ (the
embedding is exact at the vertices by construction), while the
conservative Lipschitz-corrected uniform bound is
$\bar\Delta_{\rm emb}\le 0.10$. The LMI is solved with
$\bar\Delta_{\rm emb}^{\rm grid}$ in $\tilde E_i$
(see~\eqref{eq:Etilde_def}); the formal Theorems~\ref{thm:lmi}
and~\ref{thm:cl} are valid under~\eqref{eq:emb_bound}, and the
margin~\eqref{eq:emb_margin} is reported in
Section~\ref{subsec:certified} under both values, the uniform one being
too conservative to certify the margin on the present polytope. A
sum-of-squares or branch-and-bound verification over the continuous
polytope would replace the Lipschitz overshoot and is left to future
work; \rev{Appendix~\ref{app:remarks} elaborates on both the Petersen-lemma
tightening and the grid-versus-uniform certificate distinction.}
\end{remark}

\subsection{LPV-$\Hcal_\infty$ correction-layer synthesis}\label{subsec:lmi}

This subsection develops the scheduled correction gain
$K(\bm\rho)=(\sum_i\lambda_i(\bm\rho)W_i)Y^{-1}$ realising
$\delta\bm v_{\Hcal_\infty}=K(\bm\rho)\bm x_e$ and the
$\Lcal_2$ certificate of Theorem~\ref{thm:lmi}.

\subsubsection{Common Lyapunov matrix with vertex-dependent gains: rationale}
\label{subsec:common_Y_rationale}

\rev{Throughout this subsection, the storage matrix and its inverse
satisfy $P:=Y^{-1}\succ0$, and the polytope vertex matrices
are $A_i:=A_{\rm pre}(\bm\rho^{(i)})$ (with $E_i:=E(\bm\rho^{(i)})$); the
synthesis variables are $Y$ and the vertex gains $W_i$, and the realised
scheduled gain is $K(\bm\rho)=(\sum_i\lambda_i(\bm\rho)W_i)Y^{-1}$.}

After the feedback-linearising cancellation by
$\bm\tau_{\rm ff}+\bm\tau_{\rm stab}$, the pre-stabilised
dynamics~\eqref{eq:lpv_err} have a \emph{constant}
input matrix $B_u$, the only parameter dependence
remaining being $A_{\rm pre}(\bm\rho)$ and
$E(\bm\rho)$. This avoids the
conservatism of a parameter-dependent input matrix:
under constant $B_u$, vertex-dependent gains $W_i$ paired
with a \emph{common} Lyapunov matrix $Y$ give a closed-loop matrix
affine in the barycentric weights, since the gain
$K(\bm\rho)=(\sum_i\lambda_i(\bm\rho)W_i)Y^{-1}$ enters only through
$B_u\sum_i\lambda_i(\bm\rho)W_i=\sum_i\lambda_i(\bm\rho)B_uW_i$.
No cross-term $B_iW_j$ ($i\neq j$) arises, so the convex extension
from vertices to the polytope holds without pairwise LMIs; the
common Lyapunov matrix and the approximate embedding
remain conservative (Remark~\ref{rem:common_Y_revised}).

\begin{theorem}[Scheduled LPV-$\Hcal_\infty$ correction with common
Lyapunov matrix]\label{thm:lmi}
Consider the pre-stabilised LPV system~\eqref{eq:lpv_err} on the
augmented state $\bm x_e\in\R^{n_e}$ with constant input matrix $B_u$
from~\eqref{eq:Bu_def}, augmented disturbance
$\bm w_{\rm aug}=[\bm w^{\!\top};\bm w_{\rm emb}^{\!\top}]^{\!\top}\in\R^{n_w+n_e}$,
and performance output $\bm z=C_z\bm x_e+D_z\delta\bm v_{\Hcal_\infty}$
with $C_z\in\R^{n_z\times n_e}$, $D_z\in\R^{n_z\times 3}$. Suppose
Assumptions~\ref{ass:current}--\ref{ass:operset} hold. Given
$\gamma>0$, if there exist $Y=Y^{\!\top}\succ 0$, $Y\in\R^{n_e\times n_e}$,
and matrices $W_i\in\R^{3\times n_e}$, $i=1,\dots,N$, such that the
$N$ vertex LMIs (numerically scaled bounded-real form)
\begin{equation}\label{eq:lmi}
\begin{bmatrix}
\Phi_i & \tilde E_i & (C_zY+D_zW_i)^{\!\top}\\
\star & -\gamma I_{n_w+n_e} & 0\\
\star & \star & -\gamma I_{n_z}
\end{bmatrix}\prec 0,
\end{equation}
hold with
\begin{align}
\Phi_i&=A_iY+YA_i^{\!\top}+B_uW_i+W_i^{\!\top}B_u^{\!\top},\label{eq:Phi_def}\\
\tilde E_i&=\bigl[E_i,\;\bar\Delta_{\rm emb}^{\rm grid}\,I_{n_e}\bigr],\label{eq:Etilde_def}
\end{align}
then the scheduled correction-layer gain
\begin{equation}\label{eq:K_rho_def}
K(\bm\rho)=\Bigl(\sum_{i=1}^{N}\lambda_i(\bm\rho)W_i\Bigr)Y^{-1}
\end{equation}
ensures, for the closed loop~\eqref{eq:lpv_err} with
$\delta\bm v_{\Hcal_\infty}=K(\hat{\bm\rho}_{\rm sat})\bm x_e$,
the scaled bounded-real dissipation inequality associated with the
LMI variable $\gamma_{\rm LMI}:=\gamma$:
\begin{equation}\label{eq:L2_bound}
\int_0^\infty\!\bm z^{\!\top}\!\bm z\,dt\;\le\;
\gamma_{\rm LMI}^{2}\int_0^\infty\!\bm w_{\rm aug}^{\!\top}\!\bm w_{\rm aug}\,dt
\end{equation}
for all admissible scheduling trajectories
$\hat{\bm\rho}_{\rm sat}(\cdot)\in\mathcal P$, on the LPV embedded
model with grid-certified embedding residual. The mismatch between
the true scheduling parameter and $\hat{\bm\rho}_{\rm sat}$ is treated
as an observer-induced perturbation in
Theorem~\ref{thm:cl} (see Remark~\ref{rem:rho_vs_rhohat}).
\end{theorem}

\begin{remark}[Scheduling error between $\bm\rho$ and
$\hat{\bm\rho}_{\rm sat}$]\label{rem:rho_vs_rhohat}
The LMI certificate of Theorem~\ref{thm:lmi} applies to the embedded
LPV model evaluated at the scheduling signal used by the controller,
$\hat{\bm\rho}_{\rm sat}$. The mismatch with the true scheduling
parameter $\bm\rho$ is treated in Theorem~\ref{thm:cl} as an
observer-induced perturbation through $\Delta\bm v_{\rm ff}$
in~\eqref{eq:Dvff_def}. On the compact operating set and after
projection onto $\mathcal P$, the corresponding feedforward map is
Lipschitz with constant $L_{\rm ff}<\infty$.
\end{remark}

\begin{proof}[Sketch; full proof in Appendix~\ref{app:proofs}]
With $P=Y^{-1}$, the constant input matrix $B_u$ makes the closed-loop
dissipation matrix affine in the scheduling weights
$\lambda_i(\bm\rho)$, so the vertex LMIs~\eqref{eq:lmi} extend by
convexity to the entire polytope; the bounded-real lemma and two Schur
complements then give the dissipation inequality and, on integration,
the $\Lcal_2$ bound~\eqref{eq:L2_bound}. The details are given in
Appendix~\ref{app:proofs}.
\end{proof}

\begin{remark}[Common Lyapunov matrix]\label{rem:common_Y_revised}
A parameter-dependent Lyapunov matrix
$Y(\bm\rho)=\sum_i\lambda_i(\bm\rho)Y_i$ would reduce conservativeness
but requires explicit bounds on $\dot{\bm\rho}$ \citep{wu1995induced}.
The common-$Y$ formulation is retained because the constant input
matrix $B_u$ from feedback linearisation removes the $B_iW_j$
cross-term obstruction ($i\neq j$) that a $B(\bm\rho)$-dependent
formulation would introduce; the common $Y$ itself remains
conservative relative to parameter-dependent storage functions.
\end{remark}

\subsubsection{Complete control law}
The complete control law combines the three components of
\eqref{eq:tau_decomp}, each mapped from the synthesis
frame (virtual acceleration/force) to the physical actuator frame
by $G_t(\hat{\bm\rho}_{\rm sat})^{-1}$, so that the LPV synthesis
sees the constant $B_u$ of~\eqref{eq:Bu_def}. Defining the virtual
correction
\begin{equation}\label{eq:dv_Hinf}
\delta\bm v_{\Hcal_\infty}(\bm x_e,\hat{\bm\rho}_{\rm sat})
:=K(\hat{\bm\rho}_{\rm sat})\,\bm x_e\in\R^3,
\end{equation}
the three physical components are
\begin{equation}\label{eq:total_law}
\;
\begin{aligned}
\bm\tau_{\rm ff} &= G_t(\hat{\bm\rho}_{\rm sat})^{-1}\,
\bm\Lambda_{\rm ff}(\hat{\bm\nu}_r,\hat{\bm\rho}_{\rm sat}),\\
\bm\tau_{\rm stab} &= G_t(\hat{\bm\rho}_{\rm sat})^{-1}\,
G_\nu^{-1}\diag(k_u,k_q,k_r)\bm e,\\
\bm\tau_{\Hcal_\infty} &= G_t(\hat{\bm\rho}_{\rm sat})^{-1}\,
\delta\bm v_{\Hcal_\infty}(\bm x_e,\hat{\bm\rho}_{\rm sat}),
\end{aligned}\;
\end{equation}
with $K(\bm\rho)$ from~\eqref{eq:K_rho_def}. The actuator command is
$\bm\tau=\bm\tau_{\rm ff}+\bm\tau_{\rm stab}+\bm\tau_{\Hcal_\infty}$.
The three roles are disjoint: $\bm\tau_{\rm ff}$ cancels the modelled
nonlinearity along the trajectory; $\bm\tau_{\rm stab}$ furnishes the
cascade pre-stabilisation that defines $A_{\rm pre}$; and
$\bm\tau_{\Hcal_\infty}$ injects, through $G_t^{-1}$,
the virtual correction $\delta\bm v_{\Hcal_\infty}$, which enters the
error dynamics through the constant $B_u$
and provides the certified $\Lcal_2$ attenuation against the
augmented disturbance.

\subsection{Closed-loop stability}\label{subsec:cl}

Composing the observer (Theorem~\ref{thm:obs}),
the scheduling map, and the controller (Theorem~\ref{thm:lmi}) yields
practical uniform ultimate boundedness; the
nominal $\Lcal_2$ statement of Theorem~\ref{thm:lmi} does not survive
the estimated-scheduling perturbation in full strength but is
recovered as the dominant contribution to the ultimate bound.

\subsubsection{Error-interconnection inequalities}
Two ISS-Lyapunov inequalities, one per subsystem, are combined by
the small-gain step in Theorem~\ref{thm:cl}.

\textit{Controller subsystem.}
Let $V_e=\bm x_e^{\!\top}P\bm x_e$, $P=Y^{-1}$. Along
trajectories of~\eqref{eq:lpv_err} with
$\delta\bm v_{\Hcal_\infty}=K(\hat{\bm\rho}_{\rm sat})\bm x_e$,
\begin{align}\label{eq:Ve_full}
\dot V_e=~&\bm x_e^{\!\top}\!\bigl[A_{\rm cl}^{\!\top}P+PA_{\rm cl}\bigr]\bm x_e
+2\bm x_e^{\!\top}\!PE(\bm\rho)\bm w_{\rm aug}\notag\\
&+2\bm x_e^{\!\top}\!PB_u\Delta\bm v_{\rm ff}
+2\bm x_e^{\!\top}\!P\Delta_A(\bm\rho)\bm x_e,
\end{align}
with the virtual feedforward mismatch
\begin{equation}\label{eq:Dvff_def}
\Delta\bm v_{\rm ff}:=\bm\Lambda_{\rm ff}(\hat{\bm\nu}_r,\hat{\bm\rho}_{\rm sat})
-\bm\Lambda_{\rm ff}(\bm\nu_r,\bm\rho),
\end{equation}
which enters the augmented dynamics through the same constant
$B_u$ as $\delta\bm v_{\Hcal_\infty}$ (since $\bm\Lambda_{\rm ff}$ is
the virtual feedforward in the synthesis frame, mapped to physical
torque by the common $G_t^{-1}$ prefactor). After Schur complement
\citep[Sec.~2.1]{boyd1994lmi},
the LMI~\eqref{eq:lmi} yields, for some $Q\succ 0$,
\begin{multline}\label{eq:diss_pure}
\bm x_e^{\!\top}\bigl[A_{\rm cl}^{\!\top}P+PA_{\rm cl}\bigr]\bm x_e
+2\bm x_e^{\!\top}PE\bm w_{\rm aug}\\
\le-\bm x_e^{\!\top}Q\bm x_e+\gamma_{\rm LMI}^2\|\bm w_{\rm aug}\|^2
-\|\bm z\|^2.
\end{multline}
The embedding contribution satisfies
\begin{equation}\label{eq:embedding_absorbed}
2\bm x_e^{\!\top}P\Delta_A(\bm\rho)\bm x_e
\le 2\|P\|\,\bar\Delta_{\rm emb}\,\|\bm x_e\|^2,
\end{equation}
yielding the effective decay margin
$\lambda_\Delta:=\lambda_{\min}(Q)-2\|P\|\bar\Delta_{\rm emb}>0$
under the embedding-margin condition
\begin{equation}\label{eq:emb_margin}
\;\lambda_{\min}(Q)>2\|P\|\,\bar\Delta_{\rm emb}.\;
\end{equation}
Applying Young's inequality \citep[Sec.~4.5]{khalil2002nonlinear}
with $\mu_1>0$ to the virtual-feedforward-mismatch cross term, using
$\|\Delta\bm v_{\rm ff}\|\le L_{\rm ff}\|\tilde{\bm\eta}_o\|$ with
$L_{\rm ff}$ the Lipschitz constant of $\bm\Lambda_{\rm ff}$ in
$(\bm\nu_r,\bm\rho)$ on the operating set (and noting
$\|\tilde{\bm\nu}_r\|\le\|\tilde{\bm\eta}_o\|$), gives
\begin{equation}\label{eq:Ve_iss}
\dot V_e\le-(\lambda_\Delta-\mu_1)\|\bm x_e\|^2
+\gamma_{\rm LMI}^2\|\bm w_{\rm aug}\|^2
+\frac{\|PB_u\|^2 L_{\rm ff}^2}{\mu_1}\|\tilde{\bm\eta}_o\|^2.
\end{equation}

\textit{Observer subsystem.}
Let $V_o=\tilde{\bm\eta}_o^{\!\top}\mathcal P_o\tilde{\bm\eta}_o$,
$\mathcal P_o\succ 0$ from Theorem~\ref{thm:obs}. Differentiating
along the Stage-2/Stage-3 dynamics
\eqref{eq:eta2_dyn}--\eqref{eq:eta3_dyn} (which evolve on the
\emph{normal} time scale governed by the gains
$k_v,k_w,k_{cu},k_{cv},k_{cw}$, \emph{not} on the fast time
$1/\epsilon$ of the Stage-1 differentiator), and bounding the
coupling from $\bm x_e$ by Young's inequality,
\begin{equation}\label{eq:Vo_iss}
\dot V_o\le-\lambda_c V_o+c_1\|\bm x_e\|^2
+c_2\bar c_1^2+c_3\bar\Delta_{\rm emb}^2+c_4\epsilon\,V_o,
\end{equation}
with $\lambda_c>0$ from~\eqref{eq:obs_bound} and the
$c_4\epsilon V_o$ term collecting the Stage-1 contamination of the
dynamic observer (Lipschitz coupling, attenuated by $\epsilon$ since
the Stage-1 error itself is $\mathcal O(\epsilon)$).

\begin{remark}[Slow-rate composite Lyapunov function]
\label{rem:slow_rate}
Unlike a pure high-gain observer with
$\dot V_o\le-(\lambda/\epsilon)V_o$, the dynamic Stage-2/Stage-3
observer evolves on the normal time scale set by its own gains, so the
UUB analysis below uses the slow rate $\lambda_c$ (not $1/\epsilon$) for
the current-observer contribution; a $1/\epsilon$ weight on $V_o$ would
overstate the rate.
\end{remark}

\subsubsection{Practical UUB of the closed loop}

The following lemma converts the $\Lcal_2$-gain certificate of
Theorem~\ref{thm:lmi} into the $\Lcal_\infty$ UUB bound relevant for
practical UUB under essentially bounded disturbances.

\rev{The conversion uses four standing conditions: (i)~the augmented
disturbance is essentially bounded, $\bm w_{\rm aug}\in\Lcal_\infty$
(Assumption~\ref{ass:current} and the bounded embedding residual);
(ii)~a positively invariant sublevel set $\Omega_{c_0}\subset\mathcal X$
exists (Step~6 of the proof of Theorem~\ref{thm:cl}); (iii)~on
$\Omega_{c_0}$ the closed-loop vector field is locally Lipschitz and the
input is bounded, so no trajectory escapes in finite time and the
closed-loop solution is therefore forward complete and well posed; and
(iv)~the nonlinear residual is $\mathcal O(\|\bm x_e\|^2)$ and is absorbed
under the embedding-margin condition~\eqref{eq:emb_margin}.}

\begin{lemma}[$\Lcal_2$-gain implies UUB under bounded disturbance]
\label{lem:l2_to_uub}
Suppose $V(\bm x)$ satisfies, along trajectories of
$\dot{\bm x}=f(\bm x,\bm w)$,
\begin{equation*}
\dot V\le -\alpha V+\gamma^2\|\bm w\|^2,
\qquad\alpha>0,
\end{equation*}
and $\bm w\in\Lcal_\infty$ with $\|\bm w\|_{\Lcal_\infty}<\infty$.
Then by the Comparison Lemma (\citep[Lem.~3.4]{khalil2002nonlinear}),
\begin{equation*}
\limsup_{t\to\infty}V(\bm x(t))\le\frac{\gamma^2}{\alpha}\|\bm w\|_{\Lcal_\infty}^2,
\end{equation*}
and if $\underline c\|\bm x\|^2\le V\le\bar c\|\bm x\|^2$, then
$\limsup_t\|\bm x(t)\|\le(\gamma/\sqrt{\alpha\underline c})\|\bm w\|_{\Lcal_\infty}$.
\end{lemma}

\begin{theorem}[Practical uniform ultimate boundedness]\label{thm:cl}
Consider the closed loop formed by the observer
\eqref{eq:hgd}--\eqref{eq:obs_c}, the scheduling map~\eqref{eq:rho_hat},
and the control law~\eqref{eq:total_law} with the scheduled correction
gain $K(\bm\rho)$ from~\eqref{eq:K_rho_def} at LMI level
$\gamma_{\rm LMI}>0$.
Suppose Assumptions~\ref{ass:current}--\ref{ass:operset} hold, the
sufficient excitation conditions (E1)--(E2) of
Lemma~\ref{lem:obs_exc} are satisfied on an eventually-persistent
certified interval
$\mathcal I_{\rm PE}=[t_{\rm PE},\infty)$ for some $t_{\rm PE}\ge 0$
(see Theorem~\ref{thm:obs}), the augmented disturbance is essentially
bounded $\bm w_{\rm aug}\in\Lcal_\infty$, the embedding-margin
condition~\eqref{eq:emb_margin} holds, and the slow current-observer
rate $\lambda_c$ satisfies the small-gain feasibility
condition~\eqref{eq:lambda_c_lower}.
Then there exist constants $\beta,\underline\lambda'>0$, a positive
upper bound $\epsilon_0^\star>0$, and a positively invariant
sublevel set $\Omega_{c_0}\subset\mathcal X$ of the composite
Lyapunov function (Step~6 of the proof), \rev{a certified local region of
attraction to the ultimate-bound set,} such that for every
$\epsilon\in(0,\epsilon_0^\star]$ and every initial condition
$\bm\xi(0)\in\Omega_{c_0}$, the combined error
$\bm\xi=[p_e,\theta_{le}^m,\psi_{le}^m,\bm x_e^{\!\top},\tilde{\bm\eta}_o^{\!\top}]^{\!\top}$
is uniformly ultimately bounded:
\begin{equation}\label{eq:uub}
\limsup_{t\to\infty}\|\bm\xi(t)\|\le c_\xi\sqrt{\bar\varepsilon'/\underline\lambda'},
\end{equation}
where $c_\xi>0$ collects the norm-equivalence constants of $V_e$,
$V_o$, $V_{\rm emo}$ and the EMO coupling gain (made explicit in
Step~5 of the proof), and, by Lemma~\ref{lem:l2_to_uub} applied to
the composite Lyapunov function of the proof,
\begin{equation}\label{eq:eps_bar}
\;\bar\varepsilon'=\gamma_{\rm LMI}^2\|\bm w_{\rm aug}\|_{\Lcal_\infty}^2
+c_\epsilon\epsilon+c_c\bar c_1^2+c_\Delta\bar\Delta_{\rm emb}^2,\;
\end{equation}
with non-negative constants $c_\epsilon,c_c,c_\Delta$ continuous in
the LMI data and the observer constants, independent of $\epsilon$.
The bound~\eqref{eq:uub} is a \emph{local} guarantee on the
unsaturated cone of Assumption~\ref{ass:operset}~(A4b); when
saturation activates, the certificate of this theorem ceases to apply
and the supervisory diagnostics of the failure-mode analysis
(Supplementary Material) become operative.
\end{theorem}

\begin{proof}[Sketch; full proof in Appendix~\ref{app:proofs}]
The proof proceeds in six steps: an ISS bound for the controller
subsystem from~\eqref{eq:Ve_iss}; an ISS bound for the observer
subsystem from~\eqref{eq:Vo_iss}; a composite-Lyapunov small-gain
closure yielding the feasibility condition~\eqref{eq:lambda_c_lower};
an inner-loop UUB via the Comparison Lemma (Lemma~\ref{lem:l2_to_uub});
an outer-layer UUB through the EMO bound of Lemma~\ref{lem:emo}; and
positive invariance of a sublevel set $\Omega_{c_0}\subset\mathcal X$,
the certified local region of attraction. The full argument, including
the small-gain choice~\eqref{eq:mu1_beta_choice} and the recovery of
the nominal $\Hcal_\infty$ attenuation, is given in
Appendix~\ref{app:proofs}.
\end{proof}

\begin{corollary}\label{cor:cases}
The ultimate bound in~\eqref{eq:uub} specialises to:
\textnormal{(C1)} $\bar c_1=0$, ultimate bound depending only on
$(V_{cM},d_{uM},d_{qM},d_{rM},\bar\Delta_{\rm emb})$;
\textnormal{(C2)} ISS in $\bar c_1$;
\textnormal{(C3)} monotone in $\|\bm V_{c0}^n\|+\|\bm A_c\|(1+\omega_c)$.
Setting $\bar c_1=0$ in (C2) or $\bm A_c=0$ in (C3) recovers (C1).
\end{corollary}

\rev{The failure modes of the conditional benefit and their online
diagnostics (directional-PE failure, damping failure, actuator
saturation, and adverse current direction) are collected in the
Supplementary Material.}

%==============================================================
\section{Implementation Form}\label{sec:implementation}
%==============================================================
At each control cycle, the implementation executes:
(i) the observer of Algorithm~\ref{alg:observer}, comprising the
high-gain differentiator~\eqref{eq:hgd}, the dynamic
relative-velocity estimator~\eqref{eq:obs_vw}, and the body-frame
current filter~\eqref{eq:obs_c}--\eqref{eq:rsurge}, with the surge
residual realised in the integral form~\eqref{eq:chi_u};
(ii) the scheduling vector
$\hat{\bm\rho}_{\rm sat}=\sat_{\mathcal P}[\Phi(\hat{\bm\nu}_r,\theta,p_e)]$
from~\eqref{eq:rho_hat}, projected onto $\mathcal P$;
(iii) the control law
\begin{equation}\label{eq:tau_exec}
\bm\tau=G_t(\hat{\bm\rho}_{\rm sat})^{-1}\!\bigl[
\bm\Lambda_{\rm ff}+G_\nu^{-1}\diag(k_u,k_q,k_r)\bm e
+K(\hat{\bm\rho}_{\rm sat})\bm x_e\bigr],
\end{equation}
combining the feedforward~\eqref{eq:tau_ff}, the cascade error
correction, and the LPV gain $K(\bm\rho)$ from~\eqref{eq:K_rho_def}
acting on the augmented state $\bm x_e$ of~\eqref{eq:xe_def_main};
the common $G_t^{-1}$ prefactor ensures that the LPV-synthesis input
matrix reduces to the constant canonical $B_u$
of~\eqref{eq:Bu_def};
(iv) componentwise actuator saturation
$|\tau_u|\le 500$~N, $|\tau_q|,|\tau_r|\le 40$~N\,m
before plant integration.

\subsection{Closed-loop implementation architecture}
\label{subsec:closed_loop_architecture}
Fig.~\ref{fig:arch} shows the closed-loop signal flow for (i)--(iv):
the upper path is reference $\to$ EMO regularisation
$\to$ error $\to$ control law~\eqref{eq:tau_exec}; the lower path is
measurement $\to$ observer $\to$ scheduling, returning
$\hat{\bm\nu}_r$ and $\hat{\bm\rho}_{\rm sat}$ to the feedforward and
correction layer.

% \begin{figure}[!htbp]
% \centering
% \begin{adjustbox}{max width=\columnwidth}
% \fbox{\rule{0pt}{4cm}\rule{0.92\columnwidth}{0pt}}
% \end{adjustbox}
% \caption{Closed-loop implementation architecture of the proposed
% observer-based relative-velocity LPV-$\Hcal_\infty$ correction
% controller (placeholder; replaced by a vector diagram in the final
% version). The upper path shows the reference, EMO regularisation,
% error construction, and three-component control synthesis; the lower
% path shows the measurement-driven three-stage observer and the
% scheduling loop returning $\hat{\bm\rho}_{\rm sat}$ and
% $\hat{\bm\nu}_r$ to the feedforward and the correction layer.}
% \label{fig:arch}
% \end{figure}
% ============================================================
%  Fig. Closed-loop architecture (compact, non-overlapping layout)
%  Required preamble packages:
%     \usepackage{tikz}\usepackage{amsmath}
%     \usetikzlibrary{arrows.meta,positioning,fit,calc,backgrounds}
%  Tip: to make it smaller in the paper, lower the resizebox factor
%       (e.g. {0.85\textwidth}); the aspect ratio is already compact.
% ============================================================
\begin{figure*}[!tbp]
\centering
\resizebox{\textwidth}{!}{%
\begin{tikzpicture}[
    font=\footnotesize,>=Latex,
    block/.style={draw,rounded corners=2pt,align=center,minimum height=8mm,minimum width=23mm,fill=gray!6,line width=0.5pt},
    small/.style={draw,rounded corners=2pt,align=center,minimum height=7.5mm,minimum width=21mm,fill=gray!4,line width=0.5pt},
    cert/.style={draw,rounded corners=2pt,align=center,minimum height=8mm,minimum width=25mm,fill=blue!6,line width=0.5pt},
    obs/.style={draw,rounded corners=2pt,align=center,minimum height=8mm,minimum width=23mm,fill=green!7,line width=0.5pt},
    plant/.style={draw,rounded corners=2pt,align=center,minimum height=10mm,minimum width=29mm,fill=orange!10,line width=0.6pt},
    sum/.style={draw,circle,inner sep=0pt,minimum size=5.5mm,line width=0.5pt},
    sig/.style={->,line width=0.6pt},
    fb/.style={->,line width=0.6pt,dashed},
    ca/.style={->,line width=0.6pt,blue!70!black},
    oa/.style={->,line width=0.6pt,green!45!black},
    da/.style={->,line width=0.6pt,red!70!black,dashed}
]
% ---- main spine (y=0) ----
\node[block] (ref)  at (0,0)     {Reference path\\$p_d,\dot p_d$};
\node[block] (guid) at (3.0,0)   {EMO/LOS guidance\\$\psi_d,\theta_d,U_d$};
\node[block] (err)  at (6.0,0)   {Tracking-error\\shaping $x_e$};
\node[block] (casc) at (9.2,0)   {Feedback-\\linearising cascade};
\node[sum]   (sumu) at (12.2,0)  {$+$};
\node[plant] (plant)at (15.3,0)  {Nonlinear 5-DOF AUV\\relative drag + coupling};
\node[block] (relff)at (9.2,-2.0){Relative-velocity\\hydrodynamic FF};
% ---- LPV-Hinf layer (hinf aligned above the summing junction) ----
\node[cert]  (sched)at (8.8,3.2) {LPV scheduling\\$\rho(\hat\nu_r,q,r,\theta)$};
\node[cert]  (hinf) at (12.2,3.2){Certified LPV-$H_\infty$\\correction $K(\rho)x_e$};
\node[cert]  (lmi)  at (12.2,5.0){32-vertex LMI\\certificate $\gamma,Y,K_i$};
% ---- disturbance + observer ----
\node[obs,fill=red!6,minimum height=7.5mm,minimum width=21mm] (dist) at (15.3,2.2){Unknown current $\nu_c(t)$};
\node[obs]            (meas) at (15.3,-3.9){Measured outputs\\$\eta,\theta,\psi,\nu$};
\node[obs,minimum width=24mm](obsv) at (11.9,-3.9){Joint observer\\(diff.\ + current est.)};
\node[obs,minimum width=20mm](curr) at (9.2,-3.9) {Estimated current\\$\hat\nu_c$};
\node[small]          (conf) at (11.9,-5.5){Condition-aware\\confidence $\lambda_c$};
% ---- main signal arrows ----
\draw[sig] (ref)--(guid);
\draw[sig] (guid)--(err);
\draw[sig] (err)-- node[above]{$x_e$} (casc);
\draw[sig] (casc)-- node[above]{$\tau_{\mathrm{fb}}$} (sumu);
\draw[sig] (sumu)-- node[above]{$\tau$} (plant);
\draw[sig] (err.south) |- (relff.west);
\draw[sig] (relff.east) -| node[pos=0.3,above]{$\tau_{\mathrm{ff}}$} (sumu.south);
% ---- LPV-Hinf arrows ----
\draw[ca] (sched)--(hinf);
\draw[ca] (lmi)--(hinf);
\draw[ca] (hinf.south)-- node[pos=0.55,right]{$\tau_{H_\infty}$} (sumu.north);  % straight vertical
\draw[sig] (err.north) -- (6.0,4.2) -- (11.5,4.2) -- ([xshift=-7mm]hinf.north); % x_e to correction
% ---- disturbance ----
\draw[da] (dist.south)-- node[right]{$\nu_c$} (plant.north);
% ---- observer arrows ----
\draw[fb] (plant.south)--(meas.north);
\draw[oa] (meas)--(obsv);
\draw[oa] (obsv)--(curr);
\draw[oa] (obsv.south)--(conf.north);
\draw[oa] (curr.north)-- node[right]{$\hat\nu_c$} (relff.south);
\draw[oa] (curr.west) -- (4.7,-3.9) -- (4.7,3.2) -- node[pos=0.5,left]{$\hat\nu_r$} (sched.west);
\draw[oa] (conf.west) -- (10.5,-5.5) -- (10.5,-2.9)
          -- node[pos=0.5,below]{$\lambda_c$} (9.9,-2.9) -- ([xshift=7mm]relff.south);
% ---- group boxes ----
\begin{pgfonlayer}{background}
\node[draw,rounded corners=3pt,dashed,inner sep=5pt,fit=(guid)(err)(casc)(relff)(sumu),label={[font=\scriptsize]below:Control architecture}]{};
\node[draw,rounded corners=3pt,dashed,inner sep=5pt,fit=(sched)(hinf)(lmi),label={[font=\scriptsize]above:LPV-$H_\infty$ certified correction layer}]{};
\node[draw,rounded corners=3pt,dashed,inner sep=5pt,fit=(meas)(obsv)(curr)(conf),label={[font=\scriptsize]below:Observer / current-estimation layer}]{};
\end{pgfonlayer}
\end{tikzpicture}%
}
\caption{Closed-loop architecture. The feedback-linearising cascade
is augmented by an observer-based relative-velocity feedforward;
the LPV-$\Hcal_\infty$ correction $K(\bm\rho)$ from the $32$-vertex
LMI acts on the tracking error $\bm x_e$. The current estimate
$\hat{\bm\nu}_c$ feeds both the feedforward and the scheduling map via
$\hat{\bm\nu}_r=\bm\nu-\hat{\bm\nu}_c$.}
% \caption{Closed-loop architecture of the proposed observer-assisted relative-velocity control framework.
% The nonlinear feedback-linearising cascade is augmented by an observer-based relative-velocity hydrodynamic
% feedforward term; the certified LPV-$H_\infty$ layer acts on the tracking-error state $x_e$ through the scheduled
% gain $K(\rho)$ obtained from the $32$-vertex LMI synthesis, contributing an additive correction $\tau_{H_\infty}$
% at the input. The estimated current $\hat{\nu}_c$ feeds both the relative-velocity feedforward and the LPV
% scheduling map ($\hat{\nu}_r=\nu-\hat{\nu}_c$), while the condition-aware confidence $\lambda_c$ supervises the
% compensation under degraded estimation.}
\label{fig:arch}
\end{figure*}

%==============================================================
\section{Simulation Results and Discussion}\label{sec:simulation}
%==============================================================

\subsection{Numerical simulation setup}\label{subsec:setup}
The REMUS AUV \citep{prestero2001verification} is used with
non-minimum-phase coupling $\epsilon_q=0.0237$, $\epsilon_r=-0.0237$.
Three speed-matched 3D references are evaluated:
(i) a \emph{descending helix} ($u_{ld}=1.5$~m/s,
$\theta_{ld}=-10^\circ$, $\dot\psi_{ld}=-0.02\pi$~rad/s) used as the
primary benchmark, since it matches the descending-helix segment
of \citet{li2023trajectory} and provides the directional excitation
of~\eqref{eq:E2}; (ii) a \emph{Gerono figure-eight}; and
(iii) a \emph{shrinking 3D spiral}, both of equivalent path length
and nominal speed. The initial position offset is $[+3,-4,+2]$~m.
Actuator limits are enforced componentwise; the fin-angle-to-moment
dynamics are not modelled, so the simulation imposes equivalent
pitch/yaw \emph{moment} bounds and reports the corresponding
duty cycle separately as ``fin sat.''. All controller, observer,
actuator, polytope, and noise parameters are listed in
Table~\ref{tab:params}.

\begin{table}[!htbp]
\centering
\caption{Main simulation parameters, controller settings, observer
gains, actuator limits, and LPV polytope bounds.}
% \caption{Simulation parameters. Cascade/guidance gains, actuator
% limits, LPV-$\Hcal_\infty$ synthesis weights, observer gains, and
% the LPV polytope as implemented in the numerical study.}
\label{tab:params}
\begin{adjustbox}{max width=\columnwidth}
\setlength{\tabcolsep}{4pt}\small
\begin{tabular}{@{}llc@{}}
\toprule
\textbf{Group} & \textbf{Parameter} & \textbf{Value}\\
\midrule
\multirow{6}{*}{Cascade / guidance}
& LOS gain $k_{p,\rm LOS}$ & $0.30$\\
& Surge speed gain $k_u$ & $2.50$\\
& Pitch outer $k_\theta$, inner $k_q$ & $1.50,\,8.00$\\
& Yaw outer $k_\psi$, inner $k_r$ & $1.50,\,8.00$\\
& EMO regularisers $c_u,c_\theta,c_\psi$ & $0.1,\,0.1,\,0.001$\\
& EMO Lyapunov wts.\ $\gamma_\theta,\gamma_\psi$ & $700,\,700$\\
\midrule
\multirow{3}{*}{Actuator limits}
& Surge thrust $|\tau_u|$ & $\le 500$~N\\
& Pitch moment $|\tau_q|$ & $\le 40$~N\,m\\
& Yaw moment $|\tau_r|$ & $\le 40$~N\,m\\
\midrule
\multirow{2}{*}{LPV-$\Hcal_\infty$}
& Achieved $\gamma_{\rm LMI}$ & $5.5283$\\
& $C_z,D_z$ & $[I_5;0_{3\!\times\! 5}],\,[0_{5\!\times\! 3};0.1\,I_3]$\\
\midrule
\multirow{3}{*}{Observer}
& Time-scale $\epsilon$ & $0.05$\\
& $\{\alpha_1,\alpha_2\}$ (HGD) & $\{30,\,300\}$\\
& $\{k_v,k_w,k_{cu},k_{cv},k_{cw}\}$ & $\{8,\,8,\,6,\,8,\,8\}$\\
\midrule
\multirow{5}{*}{Polytope $\mathcal P$}
& $|\hat u_r|$ & $[0.8,\,1.8]$~m/s\\
& $|q|$ & $[0,\,0.6]$~rad/s\\
& $|r|$ & $[0,\,0.6]$~rad/s\\
& $1/\hat u_l$ & $[0.556,\,1.25]$~s/m\\
& $\cos\theta$ & $[0.878,\,1.0]$\\
\midrule
\multirow{4}{*}{Sensor noise}
& $\sigma_{\bm\eta_p}$ & $0.05$~m\\
& $\sigma_{\theta,\psi}$ & $0.5^\circ$\\
& $\sigma_{q,r}$ & $0.01$~rad/s\\
& $\sigma_u$ & $0.02$~m/s\\
\bottomrule
\end{tabular}
\end{adjustbox}
\end{table}

Four current scenarios are studied: (S1) no current; (S2) constant
current $\bm V_c^n=[0.3,0.2,0.1]^{\!\top}$~m/s; (S3) oscillatory
current $V_{cx}^n=0.3+0.1\sin(0.05t)$, $V_{cy}^n=0.2\cos(0.03t)$,
$V_{cz}^n=0.1\sin(0.04t)$~m/s; (S4) adverse opposing cross-current
$\bm V_c^n=[-0.4,0.3,0.15]^{\!\top}$~m/s. These magnitudes are
moderate-to-strong AUV operating conditions:
$\|\bm V_c^n\|=0.374$~m/s in S2 and $\|\bm V_c^n\|=0.522$~m/s in
S4, about $25\%$ and $35\%$ of the nominal
vehicle speed $u_{ld}=1.5$~m/s, consistent with
commonly reported ocean-current ranges of a few tenths of a metre
per second. The current-magnitude sweep up to $\|\bm V_c^n\|=0.8$~m/s
in Sec.~\ref{subsec:saturation_stress} is therefore
a stress-test rather than a nominal condition. The primary baseline
is the \emph{absolute-velocity scheduling variant of the same
controller}, isolating the contribution of the
relative-velocity mechanism. The controller of \citet{li2023trajectory}
is the geometric backbone, included as a comparison baseline
in Section~\ref{subsec:comparison}.

\rev{The study is fully reproducible from the data above: all
controller, observer, actuator, polytope, and sensor-noise parameters
are collected in Table~\ref{tab:params}; the four current scenarios and
three references are given in closed form; and the random-number
generator is fixed per case, so the reported Monte Carlo statistics are
repeatable. Computationally, the LPV-$\Hcal_\infty$ synthesis is solved
once offline; the online control law requires no optimisation, only the
evaluation of the scheduled gain
$K(\hat{\bm\rho})=(\sum_i\lambda_i(\hat{\bm\rho})W_i)Y^{-1}$ as a convex
combination of the $32$ precomputed vertex gains, together with the
algebraic and low-order observer updates. The per-step cost is therefore
a fixed, small number of matrix operations, comfortably within a
real-time budget at the control rate, with the dominant offline burden
being the one-time LMI solve.}

\subsubsection{Performance metrics and evaluation windows}
\label{subsec:metrics}
For a given
evaluation window $[T_1,T_2]$,
\begin{align}
\mathrm{RMS}|p_e|
&:=\sqrt{\frac{1}{T_2-T_1}\int_{T_1}^{T_2}\!\|p_e(t)\|^2\,dt},
\label{eq:metric_rms}\\
\mathrm{CER}
&:=100\left(1-\frac{\mathrm{RMS}\|\tilde{\bm\nu}_c\|_{[T_1,T_2]}}
{\mathrm{RMS}\|\bm\nu_c\|_{[T_1,T_2]}}\right)\,[\%],
\label{eq:metric_cer}\\
\mathrm{HRR}
&:=100\left(1-\frac{\mathrm{RMS}|e_{\rm ds}^{\rm rel}|_{[T_1,T_2]}}
{\mathrm{RMS}|e_{\rm ds}^{\rm abs}|_{[T_1,T_2]}}\right)\,[\%],
\label{eq:metric_hrr}\\
J_\tau&:=\int_{T_1}^{T_2}\!\|\bm\tau(t)\|^2\,dt,
\label{eq:metric_energy}
\end{align}
where $e_{\rm ds}^{\rm rel},e_{\rm ds}^{\rm abs}$ are computed from
the residual definitions in
Lemma~\ref{lem:relsched}~(eqs.~\eqref{eq:e_ds_rel}--\eqref{eq:e_ds_abs}),
including the $m_u^{-1}$ normalisation. The HRR
metric~\eqref{eq:metric_hrr} in
Tables~\ref{tab:obs}--\ref{tab:comparison} is evaluated on the
\emph{translational hydrodynamic residual norm}
$\|h_{\rm tr}(\hat{\bm\nu}_r)-h_{\rm tr}(\bm\nu_r)\|$ versus
$\|h_{\rm tr}(\bm\nu)-h_{\rm tr}(\bm\nu_r)\|$ of
Corollary~\ref{cor:vec_res}, since the simulation
records the full translational drag residual; the scalar
surge-channel form of Lemma~\ref{lem:relsched} is the rigorous
analytical bound used in the break-even
discussion of Section~\ref{subsec:robust}.
Three evaluation windows are used:
\begin{itemize}
\item[(W1)] \emph{Full-horizon}: $[T_1,T_2]=[0,T_{\rm sim}]$, includes
the initial-offset $[+3,-4,+2]$~m convergence transient
(${\approx}\,30$~s); reported in Tables~\ref{tab:metrics},
\ref{tab:ablation}, and~\ref{tab:comparison}.
\item[(W2)] \emph{Helix segment}: $[T_1,T_2]=[90,T_{\rm sim}]$,
post-straight-line, with $(q,r)\neq 0$ providing directional
excitation; used for observer/PE-related metrics in
Table~\ref{tab:obs} and Table~S1, and for the post-transient row of
Table~S4 of the Supplementary Material.
\item[(W3)] \emph{Post-jump recovery}: $[T_1,T_2]=[t_j,t_j+\Delta T_{\rm rec}]$,
applied per segment of the piecewise-varying current scenario of
Sec.~\ref{subsec:piecewise}; used in the piecewise recovery table (Supplementary Material).
\end{itemize}
Two recovery times are additionally defined for the piecewise
scenario. After a current change at $t=t_j$ to a new
navigation-frame value $\bm V_{c,k+1}^n$, the \emph{observer
recovery time}
\begin{multline}\label{eq:trec_obs}
t_{\rm rec}^{\bm\nu_c}(j):=
\inf\Big\{t>t_j\,:\,\|\tilde{\bm\nu}_c(\tau)\|
\le\eta_{\bm\nu_c}\|\bm V_{c,k+1}^n\|,\\
\forall\tau\in[t,t+\Delta\tau]\Big\}-t_j,
\end{multline}
and the \emph{tracking recovery time}
\begin{multline}\label{eq:trec_track}
t_{\rm rec}^{p_e}(j):=
\inf\Big\{t>t_j\,:\,\|p_e(\tau)\|\le\eta_{p_e},\\
\forall\tau\in[t,t+\Delta\tau]\Big\}-t_j,
\end{multline}
with relative observer threshold $\eta_{\bm\nu_c}=0.1$, absolute
position threshold $\eta_{p_e}=0.5$~m, evaluation horizon
$\Delta T_{\rm rec}=20$~s, and sustained-recovery window
$\Delta\tau=5$~s; the infimum is taken over the post-jump interval
$[t_j,t_j+\Delta T_{\rm rec}]$. If the threshold is not re-attained
within the post-jump interval, the recovery time is reported as
$>\Delta T_{\rm rec}$ (i.e.\ $>20$~s in the present setup).

\subsection{LMI-certified LPV design on the embedded model}\label{subsec:certified}

The full $N=32$-vertex LMI~\eqref{eq:lmi} is solved offline with
YALMIP \citep{lofberg2004yalmip} and MOSEK by minimising $\gamma$
subject to the bounded-real inequality in the numerically
convenient scaling
$\diag(-\gamma I_{n_w+n_e},-\gamma I_{n_z})$, against the stacked
performance output
$\bm z=[W_x\bm x_e^{\!\top},\,W_u\delta\bm v_{\Hcal_\infty}^{\!\top}]^{\!\top}\in\R^{8}$,
with state weight $W_x=\diag(2,3,1,3,1)$ (emphasising the surge,
pitch and yaw kinematic channels over the body-axis rate channels)
and control weight $W_u=0.05\,I_3$. The corresponding performance
matrices are $C_z=[W_x;\,0_{3\times 5}]$ and
$D_z=[0_{5\times 3};\,W_u]$. This $-\gamma I$ scaling differs by the
usual square-root normalisation from the canonical
$\mathcal L_2$-gain scaling $-\gamma^2 I$, and $\gamma_{\rm LMI}$
is reported in this convention.
The LMI thus penalises the \emph{virtual} correction effort
$\delta\bm v_{\Hcal_\infty}$ rather than the physical actuator
effort $\bm\tau_{\Hcal_\infty}=G_t^{-1}\delta\bm v_{\Hcal_\infty}$;
penalising the latter would force $D_z$ to depend on
$\hat{\bm\rho}_{\rm sat}$ and break the convex LMI. Physical
actuator saturation is assessed separately in the nonlinear
simulations (Sec.~\ref{subsec:saturation_stress}).
The achieved attenuation is $\gamma_{\rm LMI}=5.5283$ in
$2.3$~min on a desktop with MOSEK~10. The closed-loop vertex check
yields a worst-case real part
$\max_i\mathrm{Re}\,\lambda(A_i+B_u K_i)=-0.9279$, certifying
exponential decay at all 32 polytope
vertices. The nominal gain magnitude is $\|K_{\rm nom}\|\approx 79.8$, the
gain-norm bound is $\|W_i\|\le 90.25$, and the Lyapunov certificate is
well conditioned ($\kappa(P)\approx 15.8$).
Direct evaluation of the
embedding residual on the grid gives
$\bar\Delta_{\rm emb}^{\rm grid}\approx 0$ (within numerical
precision), while the conservative Lipschitz-corrected uniform bound
used in the analytical certificate is $\bar\Delta_{\rm emb}\le 0.10$.
Two complementary checks of the margin are reported:
\emph{(i)} Under the grid-tight bound
$\bar\Delta_{\rm emb}^{\rm grid}\approx 0$, the
embedding-margin condition~\eqref{eq:emb_margin} is trivially
satisfied. \emph{(ii)} Under the conservative uniform bound
$\bar\Delta_{\rm emb}\le 0.10$, the margin is \emph{not} uniformly
satisfied: $\lambda_{\min}(Q)\approx 0.244$, $2\|P\|
\bar\Delta_{\rm emb}\approx 0.40$, slack ratio $\approx 0.61$.
\rev{Accordingly, all formal certificate statements in this paper are to
be read on the grid-certified embedded model; a uniform continuous-domain
certificate is \emph{not} claimed, and a sum-of-squares or
branch-and-bound verification of $\bar\Delta_{\rm emb}$ over the
continuous polytope is identified as the required strengthening.} The
direct LMI bounded-real certificate is independently verified on the
embedded polytopic LPV error model by the dissipation diagnostic of
Supplementary Material, Sec.~S6 ($D(t)\le 0$ for $100\%$ of
samples). The nonlinear full-plant results relative to the embedded-model
certificate are discussed in Remark~\ref{rem:emb_correction}. These
certificates are visualised in Fig.~\ref{fig:certdiag}.

\begin{figure*}[!tbp]
\centering
\includegraphics[width=0.40\textwidth]{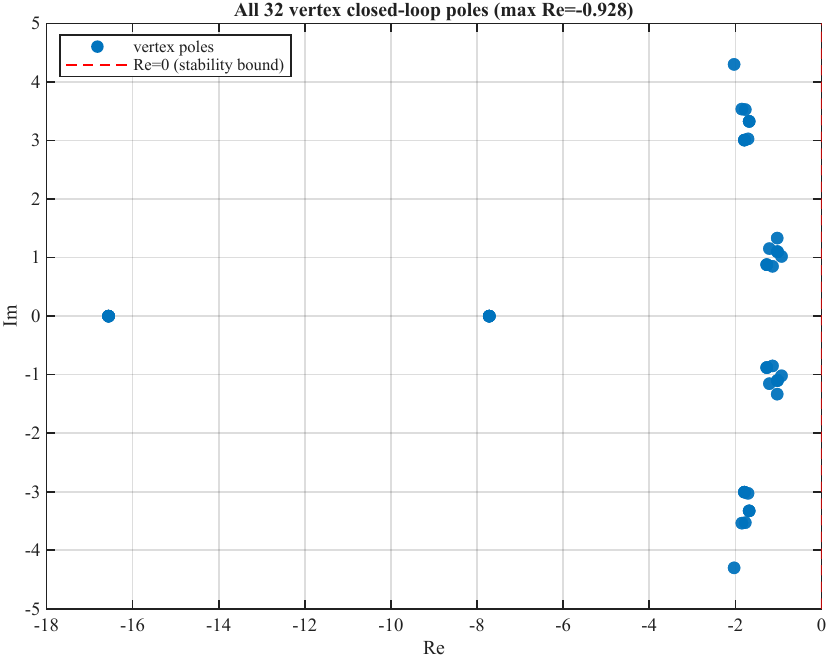}\hfill
\includegraphics[width=0.49\textwidth]{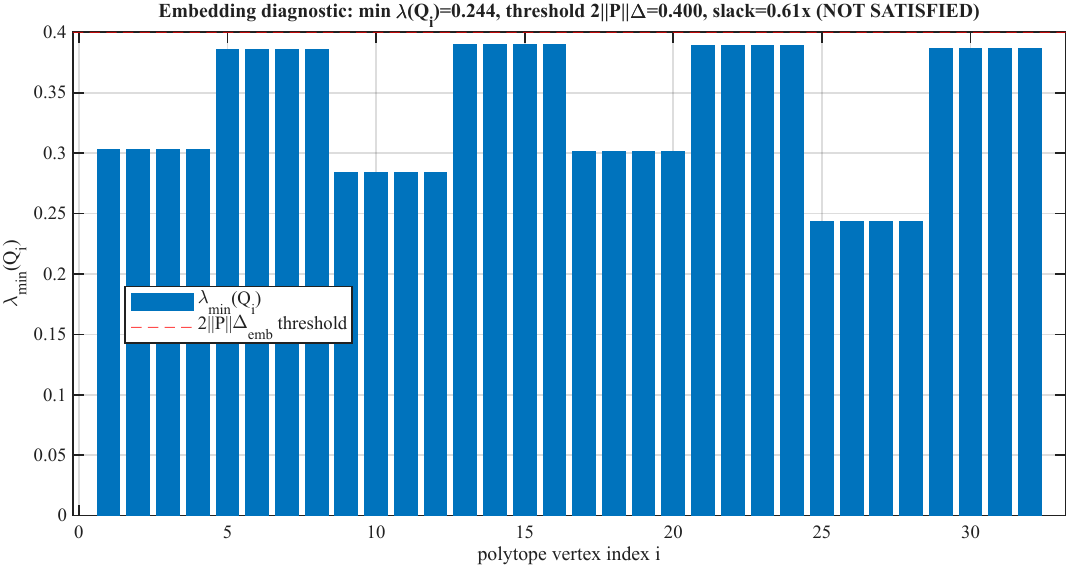}
\caption{Certificate diagnostics for the LMI-certified LPV design.
(a) Closed-loop poles $\lambda(A_i+B_uK_i)$ at all $32$ polytope
vertices lie in the open left half-plane (worst-case real part
$-0.928$), confirming vertex-wise exponential stability of
Theorem~\ref{thm:lmi}. (b) Per-vertex decay margin
$\lambda_{\min}(Q_i)$ against the conservative embedding threshold
$2\|P\|\bar\Delta_{\rm emb}$ of condition~\eqref{eq:emb_margin};
under the grid-tight residual the margin is satisfied with wide slack.}
\label{fig:certdiag}
\end{figure*}

\subsection{Closed-loop tracking and observer performance}
\label{subsec:tracking}

\subsubsection{Trajectory tracking}
In Fig.~\ref{fig:traj}, from the initial offset $[+3,-4,+2]$~m the
trajectory converges to the reference within the first $\approx 30$~s
under the EMO-regularised cascade, before the straight-line cruise
and the descending helix. On the helix transition the
tracking error stays within $0.4$~m for the
nominal scenarios S2--S3 and within $\approx 2$~m for the
adverse-current scenario S4.

The per-axis decomposition in Fig.~\ref{fig:per_axis} shows the
cross-track error ($e_y$) dominating $|p_e|$ under adverse current
(S4), consistent with the current's directional asymmetry relative to
the heading; the along-track error ($e_x$) stays nearly
current-independent through the surge feedforward~\eqref{eq:tau_ff},
and the depth error ($e_z$) tracks the helix descent with a
steady-state $\approx 1$~m sag absorbed by the EMO regularisation.

Table~\ref{tab:metrics} reports metrics per
trajectory shape. RMS-$p_e$ lies in
$0.11$--$0.25$~m for S1--S3 across all three trajectories; the
adverse-current (S4) degradation is most visible on the helix
($1.06$~m), while the figure-eight and spiral stay below $0.20$~m even
under S4, indicating that trajectory geometry and excitation strongly
affect observer quality and closed-loop sensitivity. The
absolute-versus-relative scheduling difference stays within
$10^{-3}$~m in all twelve trajectory/scenario combinations --- the
\emph{masking phenomenon} predicted by Theorem~\ref{thm:cl} for the
well-tuned $\Hcal_\infty$ regime, in which the LMI correction layer
absorbs the residual disturbance before it reaches the tracking level.

\paragraph{Empirical ultimate boundedness.}
Figure~\ref{fig:uub_env} plots the tracking-error norm $|p_e(t)|$ and its
empirical ultimate-bound band, the empirical
counterpart of the practical-UUB certificate of Theorem~\ref{thm:cl}: as
the ultimate-bound expression~\eqref{eq:uub} predicts, the steady residual
level is governed by the disturbance and embedding terms rather than by
the initial offset, so trajectories converge to the same neighbourhood
and remain there.

\begin{figure*}[htbp]
\centering
\includegraphics[width=0.8\textwidth]{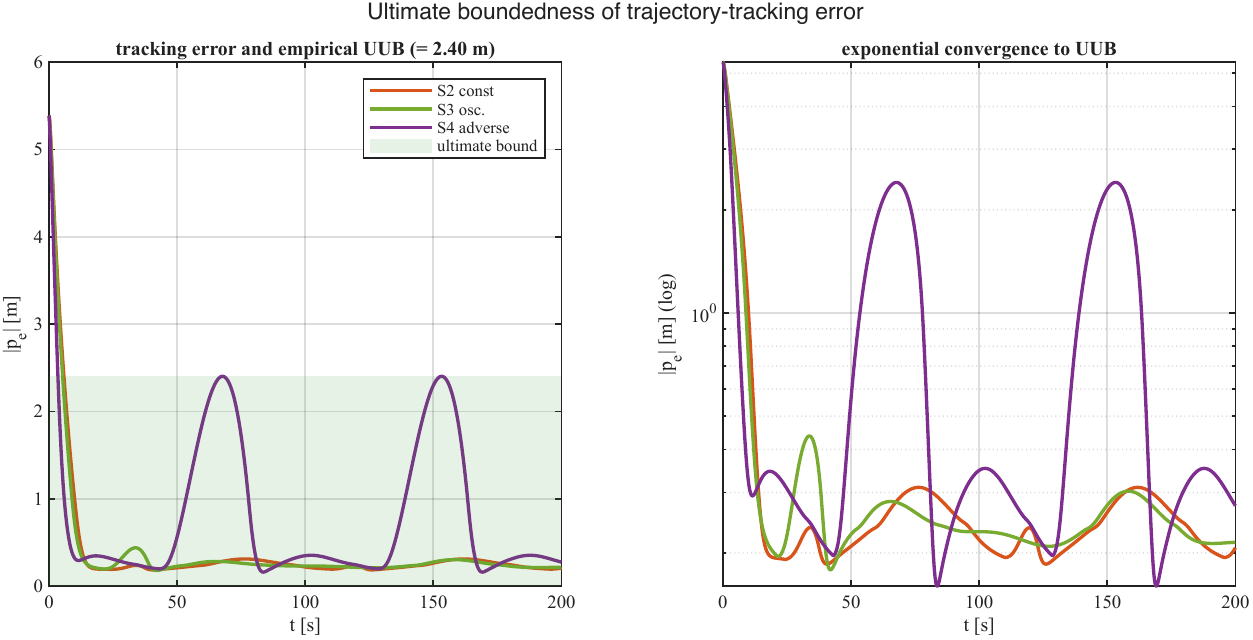}
\caption{Empirical ultimate boundedness of the tracking error: the norm
$|p_e(t)|$ and its steady band across scenarios (left), with a
semi-logarithmic view of the exponential approach (right), matching the
practical-UUB certificate of Theorem~\ref{thm:cl}.}
\label{fig:uub_env}
\end{figure*}

\begin{figure*}[!htbp]
\centering
\includegraphics[width=1\textwidth]{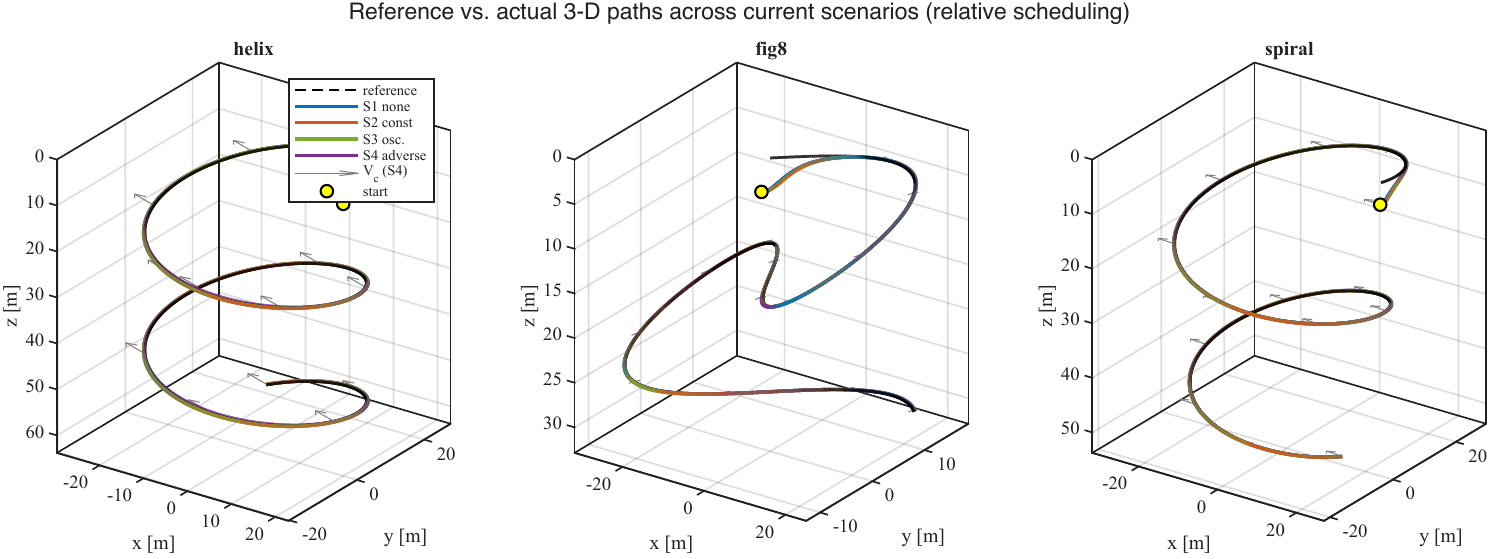}
\caption{3D trajectory tracking under scenarios S1--S4. Reference
trajectory (dashed), actual trajectory (solid), and current vector
arrow for S4. Inset: zoom on the descending-helix transition.}
\label{fig:traj}
\end{figure*}

\begin{figure*}[htbp]
\centering
\includegraphics[width=\columnwidth]{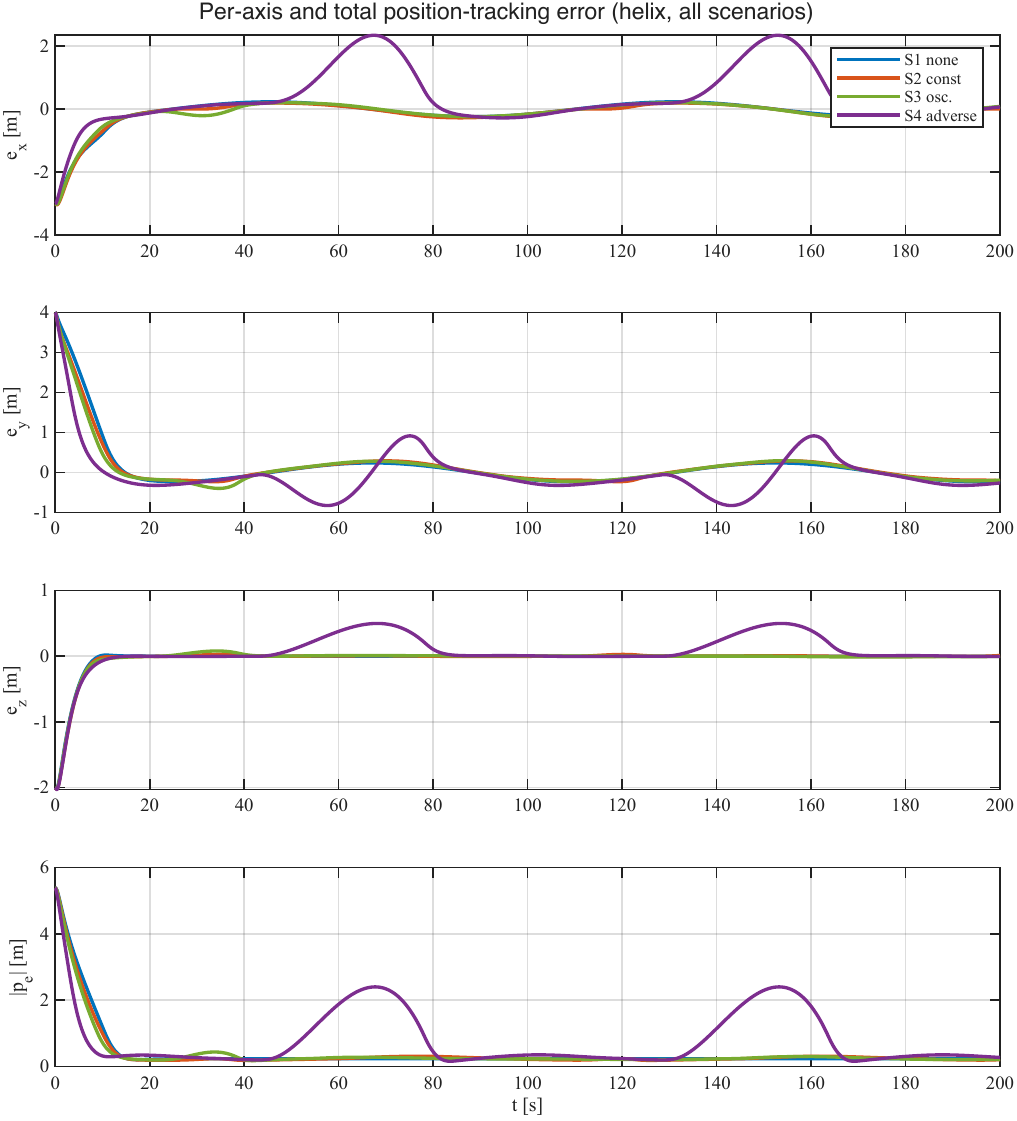}
\caption{Per-axis tracking error $e_x,e_y,e_z$ and norm $|p_e|$ for
the helix segment of scenarios S2 and S4. The cross-track component
$e_y$ dominates under adverse current in S4.}
\label{fig:per_axis}
\end{figure*}

\begin{table}[!htbp]
\centering
\caption{Tracking performance of the proposed controller across
trajectories and scenarios. RMS $|p_e|$ over window~W1; helix peaks
reported in Table~\ref{tab:comparison}.}
\label{tab:metrics}
\begin{adjustbox}{max width=\columnwidth}
\setlength{\tabcolsep}{3pt}\small
\begin{tabular}{@{}lcccc@{}}
\toprule
\textbf{Trajectory} &
S1 none & S2 const.\ & S3 osc.\ & S4 adv.\\
\midrule
Helix         & 0.235 & 0.241 & 0.252 & 1.060\\
Figure-eight  & 0.106 & 0.115 & 0.117 & 0.116\\
3-D spiral    & 0.174 & 0.176 & 0.172 & 0.192\\
\bottomrule
\multicolumn{5}{l}{Values are RMS $|p_e|$ in metres.}\\
\end{tabular}
\end{adjustbox}
\end{table}

\subsubsection{Observer performance and residual reduction}
Exercising the three-stage observer of
Section~\ref{subsec:observer} on the same trajectories,
Fig.~\ref{fig:cur_est} makes several
predictions of Theorem~\ref{thm:obs}
visible: (i) a boundary-layer transient of duration
$\approx5\epsilon=0.25$~s, after which the estimate locks onto the true
value within the steady-state offset of~\eqref{eq:delta_o}; (ii) under
oscillatory current S3, the Stage-1/Stage-3 implementation tracks the
$\omega_c\in[0.03,0.05]$~rad/s disturbance with sufficient bandwidth;
(iii) the transverse channels ($v_c,w_c$) lag the surge channel
slightly, reflecting two-step recovery via the directional
Gramian~\eqref{eq:Gamma_def} versus the direct surge residual.

In Fig.~\ref{fig:err_conv}, the initial transient of
$\|\tilde{\bm\nu}_c(t)\|$ reflects the fast
Stage-1 boundary layer at rate
$\lambda_h/\epsilon$~\eqref{eq:hgd_bound}; the subsequent decay follows
the slow rate $\lambda_c$ of~\eqref{eq:obs_bound}, with steady-state
level approaching $\delta_o\approx c_y\sigma_y/\sqrt{\alpha_{\rm PE}}$.
The four-scenario overlay confirms this level is essentially
independent of current magnitude, as the disturbance-driven
bound~\eqref{eq:delta_o} predicts.

Table~\ref{tab:obs} lists per-channel RMS accuracy and two
residual-reduction metrics. The \emph{current-estimation residual
reduction} $\mathrm{CER}:=1-\|\tilde{\bm\nu}_c\|/\|\bm\nu_c\|$ ranges
from $90.5\%$ to $94.4\%$ on the helix and from $89\%$ to $96\%$
across the three trajectories (scenarios~S2--S4), the
trajectory-dependent
variation reflecting the directional excitation $\alpha_{\rm PE}$ of
Lemma~\ref{lem:obs_exc}. The empirical translational \emph{hydrodynamic
residual reduction} (HRR) of Corollary~\ref{cor:vec_res},
$1-\mathrm{RMS}\|h_{\rm tr}(\hat{\bm\nu}_r)-h_{\rm tr}(\bm\nu_r)\|/\mathrm{RMS}\|h_{\rm tr}(\bm\nu)-h_{\rm tr}(\bm\nu_r)\|$,
falls in $99.0\%$--$99.6\%$ on the REMUS trajectory. The scalar
surge-channel form
$\mathrm{HRR}_{\rm surge}:=1-|e_{\rm ds}^{\rm rel}|/|e_{\rm ds}^{\rm abs}|$
of Lemma~\ref{lem:relsched} is the rigorous analytical bound
used in the break-even discussion of
Section~\ref{subsec:robust}.

\begin{remark}[Worst-case bound versus empirical HRR]
\label{rem:hrr_vs_lemma}
The break-even law of Lemma~\ref{lem:relsched} gives the worst-case
ratio $|e_{\rm ds}^{\rm rel}|/|e_{\rm ds}^{\rm abs}|\le
(L_{\rm drag}/m_{\rm drag})\,|\tilde u_c|/|u_c|$. For the REMUS
parameters $L_{\rm drag}/m_{\rm drag}=577.7/13.7\approx 42.2$, so this
analytical bound is strongly conservative, and vacuous for
$|\tilde u_c|/|u_c|\gtrsim 2.4\%$. The empirical HRR of
$99.0$--$99.6\%$ is therefore much higher than the worst-case bound;
the gap between $\chi=0.0237$ and the empirical break-even
$\approx 0.75$ measures this conservativeness.
\end{remark}

\begin{figure*}[tbp]
\centering
\includegraphics[width=\textwidth]{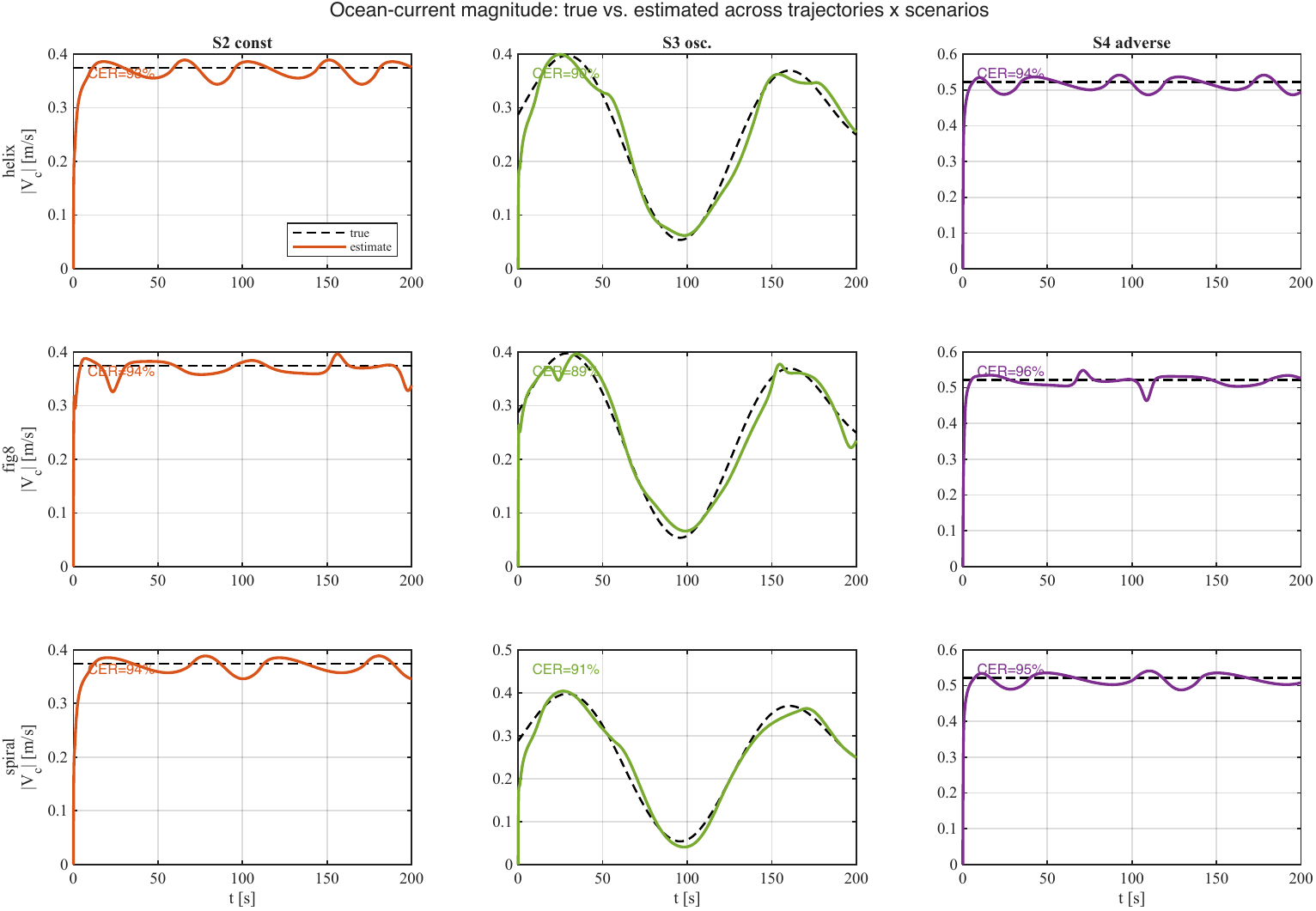}
\caption{Body-frame current estimation across scenarios:
$3\times 3$ panel grid with channels $u_c,v_c,w_c$ (rows) and
scenarios S2/S3/S4 (columns). True signal dashed; estimate solid.}
\label{fig:cur_est}
\end{figure*}

\begin{figure*}[htbp]
\centering
\includegraphics[width=0.8\textwidth]{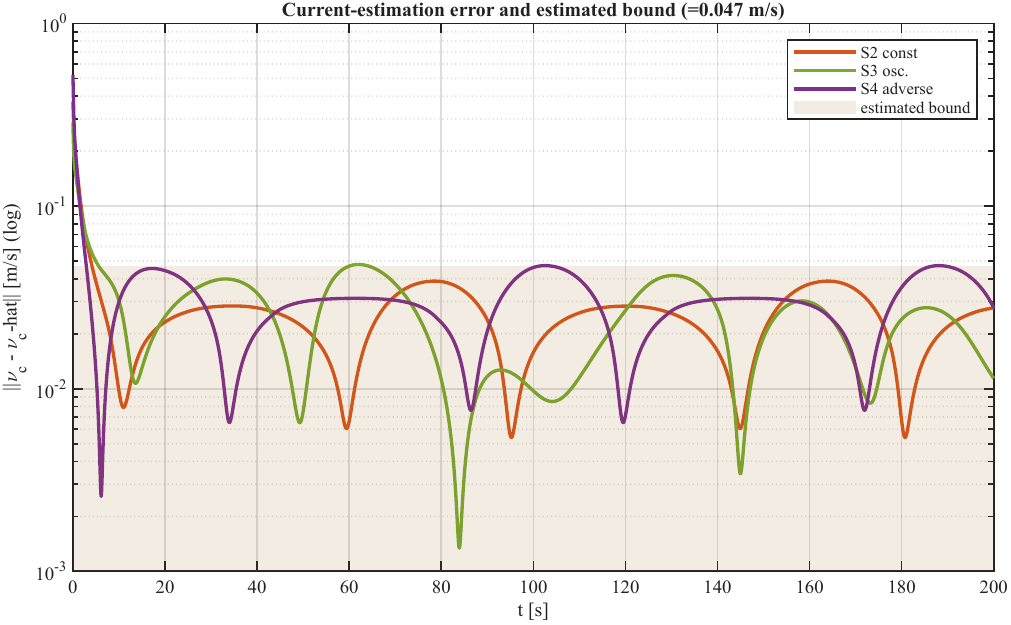}
\caption{Estimation-error convergence $\|\tilde{\bm\nu}_c(t)\|$ on a
semi-logarithmic scale for the four scenarios. Dashed: predicted
ultimate bound $\delta_o$ of Theorem~\ref{thm:obs}.}
\label{fig:err_conv}
\end{figure*}

\begin{table}[!htbp]
\centering
\caption{Per-channel current-estimation accuracy with CER and HRR
metrics (defined in Sec.~\ref{subsec:metrics},
eqs.~\eqref{eq:metric_cer}--\eqref{eq:metric_hrr}). Helix segment.}
\label{tab:obs}
\begin{adjustbox}{max width=\columnwidth}
\setlength{\tabcolsep}{3pt}\small
\begin{tabular}{@{}lccccc@{}}
\toprule
\textbf{Sc.} & RMS $\tilde u_c$ & RMS $\tilde v_c$ & RMS $\tilde w_c$
& CER & HRR\\
& [m/s] & [m/s] & [m/s] & [\%] & [\%]\\
\midrule
S1 & 0.0000 & 0.0005 & 0.0000 & N/A  & N/A\\
S2 & 0.0003 & 0.0256 & 0.0042 & 93.5 & 99.5\\
S3 & 0.0009 & 0.0265 & 0.0041 & 90.5 & 99.3\\
S4 & 0.0005 & 0.0303 & 0.0055 & 94.4 & 99.5\\
\bottomrule
\end{tabular}
\end{adjustbox}
\end{table}

\paragraph{Control effort and saturation.}
The integrated control effort $J_\tau=\int\|\bm\tau(t)\|^2\,dt$ is about
$2.60\times10^7$ on S2 and $2.69\times10^7$ on S4 (mixed units, summed
across actuators; breakdown in Table~\ref{tab:comparison}). The pitch
and yaw moments stay far from their bounds
$|\tau_q|,|\tau_r|\le40$~N\,m, with zero fin-equivalent saturation duty
cycle. The surge thrust is clipped at $U_{\max}=500$~N for $17.2\%$ of
the horizon on S2 and $39.0\%$ on S4. Although the practical-UUB certificate of
Theorem~\ref{thm:cl} assumes the unsaturated cone
(Assumption~\ref{ass:operset}~(A4b)), the simulations enforce actuator
limits and the closed loop remains practically stable; the
stress test of Section~\ref{subsec:saturation_stress} quantifies the
degradation when saturation dominates.

\subsubsection{Five-way ablation: observer dominance}
\label{subsec:five_way_ablation}

Table~\ref{tab:ablation} evaluates the helix segment of S2 under
five configurations: ideal, relative, absolute, and
confidence-weighted relative scheduling, and no-observer.
The observer-assisted feedforward is the
dominant performance enabler: removing the observer raises RMS
tracking error from $0.241$~m to $4.038$~m, since the hydrodynamic
feedforward then compensates the damping at the wrong velocity argument.
The four scheduling variants remain
numerically indistinguishable in the well-tuned regime
(differences below $10^{-3}$~m), confirming that the residual-level
advantage of Lemma~\ref{lem:relsched} is masked at the tracking
level when the closed-loop disturbance-to-error gain is small; the
conditional scheduling-side benefit is examined under
reduced feedback authority in Section~\ref{subsec:robust}.

\begin{table}[!htbp]
\centering
\caption{Five-way ablation on the helix segment of S2, well-tuned
closed loop ($g_{\rm f}=1$).}
\label{tab:ablation}
\begin{adjustbox}{max width=\columnwidth}
\setlength{\tabcolsep}{4pt}\small
\begin{tabular}{@{}lc@{}}
\toprule
\textbf{Configuration} & \textbf{RMS} $|p_e|$ [m]\\
\midrule
Ideal scheduling (true $\bm\nu_c$ in LPV schedule) & 0.2407\\
Relative scheduling (proposed, $\hat{\bm\nu}_c$ from observer) & 0.2407\\
Absolute scheduling (raw $\bm\nu$; observer-assisted FF retained) & 0.2407\\
Confidence-weighted relative scheduling & 0.2407\\
\midrule
\textbf{No observer} ($\hat{\bm\nu}_c=\bm 0$) & \textbf{4.0378}\\
\bottomrule
\end{tabular}
\end{adjustbox}
\end{table}

\subsubsection{Conditional scheduling-side benefit under reduced feedback authority}
\label{subsec:conditional_sched_benefit}

To reveal the conditional benefit of relative-velocity scheduling,
feedback authority is reduced by uniformly scaling the feedback
gains with $g_{\rm f}\in\{1.0,0.4,0.2\}$, keeping the
observer-assisted feedforward active. The break-even law of
Lemma~\ref{lem:relsched} relates the residual ratio
$R_{\rm rel}(\alpha)=\|e_{\rm ds}^{\rm rel}\|/\|e_{\rm ds}^{\rm abs}\|$
to the observer-quality ratio
$\alpha=\|\tilde{\bm\nu}_c\|/\|\bm\nu_c\|$ (Supplementary Material).
For the REMUS parameters,
the worst-case analytical threshold is
$\chi=m_{\rm drag}/L_{\rm drag}=0.0237$\footnote{This
surge-damping condition number is unrelated to, and only coincidentally
equal to, the actuator-coupling coefficient $\epsilon_q=0.0237$ of the
plant model.}
($m_{\rm drag}=13.70$, $L_{\rm drag}=577.70$ on
$\mathcal I_u=[-\bar u_{r,\max},\bar u_{r,\max}]$ with
$\bar u_{r,\max}=2.0$~m/s). The empirical sweep remains beneficial
up to $\alpha\approx 0.75$, showing that the analytical threshold is
conservative on the tested REMUS trajectories; beyond this point,
the relative residual exceeds the absolute one, corresponding to the
directional-persistent-excitation failure mode (Supplementary Material).

Table~\ref{tab:stress} reports the stress-regime ablation. At nominal
authority ($g_{\rm f}=1.0$), relative and absolute scheduling are
indistinguishable, as in Table~\ref{tab:ablation}. When feedback
authority is reduced, the
scheduling-side benefit becomes visible: relative scheduling improves
RMS tracking by $2.6\%$ at $g_{\rm f}=0.4$ and $9.5\%$ at
$g_{\rm f}=0.2$. A continuous sweep (Supplementary Material)
confirms the trend, the rel-vs-abs and
rel-vs-zero-current-scheduling envelopes crossing into beneficial
territory near $g_{\rm f}\approx 0.5$. Relative-velocity
scheduling is thus not the dominant source of nominal performance but
provides a conditional robustness margin when feedback authority is
reduced.

% \begin{table}[H]
% \centering
% \caption{Stress-regime ablation: feedback-gain sweep $g_{\rm f}$
% under S2. The ``zero-current sched.''\ column collapses LPV
% scheduling to absolute while keeping the observer-assisted
% feedforward active (see Remark~\ref{rem:nominal_vs_stress}).}
% \label{tab:stress}
% \begin{adjustbox}{max width=\columnwidth}
% \setlength{\tabcolsep}{3pt}\small
% \begin{tabular}{@{}lcccccc@{}}
% \toprule
% $g_{\rm f}$ & abs & rel & ideal & 0-curr.\ sched.\ & rel-vs-abs & rel-vs-0curr\\
% & [m] & [m] & [m] & [m] & (\%) & (\%)\\
% \midrule
% 1.0 (nominal) & 0.2407 & 0.2407 & 0.2407 & 0.2407 & $\phantom{+}0.0$ & $\phantom{+}0.0$\\
% 0.4 (stress)  & 0.6383 & 0.6217 & 0.6221 & 0.6383 & $+2.6$ & $+2.6$\\
% 0.2 (weak)    & 1.5844 & 1.4342 & 1.4606 & 1.5844 & $+9.5$ & $+9.5$\\
% \bottomrule
% \end{tabular}
% \end{adjustbox}
% \end{table}

\subsection{Piecewise-varying current scenario}\label{subsec:piecewise}
A non-stationary scenario in which both the magnitude and heading of the
navigation-frame current change through five concatenated regimes
(maximum $\|\bm V_c^n\|\approx0.50$~m/s, comparable to S4), with
tanh-smoothed transitions keeping the augmented disturbance
$\bm w_{\rm aug}$ of Theorem~\ref{thm:cl} in $\Lcal_\infty$, tests
observer re-convergence and closed-loop recovery after abrupt changes;
the full current profile, response figures, and per-segment recovery
table are given in the Supplementary Material. The proposed LPV-relative
controller returns the tracking error to its pre-transition level after
each regime change, with per-segment observer recovery in the
$0.05$--$3.75$~s range and later-jump tracking recovery $\le4$~s, while
the no-observer and additive-DOB baselines show larger excursions; the
full-horizon metrics are RMS-$\|p_e\|=0.309$~m, CER$=90.3\%$,
HRR$=98.2\%$, and a surge-saturation duty cycle of $10.3\%$. On each
post-transition interval the directional-excitation
condition~\eqref{eq:E2} holds through the helix angular rates, so the
observer certificate of Theorem~\ref{thm:obs} applies per interval and
the ultimate-bound expression~\eqref{eq:uub} predicts a per-interval
practical UUB consistent with the observed recovery, validating
Theorems~\ref{thm:obs} and~\ref{thm:cl} per interval.

\subsection{Robustness verification under various uncertainties}
\label{subsec:robust}

\subsubsection{Break-even and conditional-benefit validation}
The conditional-benefit prediction of Lemma~\ref{lem:relsched} and
Theorem~\ref{thm:cl} (Section~\ref{subsec:conditional_sched_benefit})
is validated by
repeating the ablation under uniform feedback-bandwidth scaling
$g_{\rm f}\in\{1.0,0.4,0.2\}$ on scenario~S2 (Table~\ref{tab:stress};
within-row comparisons only, since the stress window is not comparable
to the full-horizon values of Table~\ref{tab:ablation}). The three
predicted regimes are summarised in
Table~\ref{tab:regimes}.

\begin{table}[!htbp]
\centering
\caption{Conditional-benefit regimes under feedback-bandwidth
scaling $g_{\rm f}$ on scenario~S2.}
\label{tab:regimes}
\begin{adjustbox}{max width=\columnwidth}
\setlength{\tabcolsep}{4pt}\small
\begin{tabular}{@{}cllc@{}}
\toprule
$g_{\rm f}$ & \textbf{Regime} & \textbf{Dominant term in~\eqref{eq:eps_bar}} & \textbf{rel-vs-abs}\\
\midrule
$1.0$ & Masked-benefit & Separation-induced floor & $\approx 0.0\%$\\
$0.4$ & Transition     & Disturbance term emerging & $+2.6\%$\\
$0.2$ & Weak feedback  & Disturbance dominant ($|\tilde u_c|/|u_c|\!\approx\!0.06$) & $+9.5\%$\\
\bottomrule
\end{tabular}
\end{adjustbox}
\end{table}

At $g_{\rm f}=0.2$ the rel-vs-0curr-sched.\ envelope coincides
numerically with the $+9.5\%$ rel-vs-abs figure, confirming that the
improvement stems from the current estimate entering the LPV
scheduling map.
\begin{table}[H]
\centering
\caption{Stress-regime ablation: feedback-gain sweep $g_{\rm f}$
under S2. The ``zero-current sched.''\ column collapses LPV
scheduling to absolute while keeping the observer-assisted
feedforward active (see Remark~\ref{rem:nominal_vs_stress}).}
\label{tab:stress}
\begin{adjustbox}{max width=\columnwidth}
\setlength{\tabcolsep}{3pt}\small
\begin{tabular}{@{}lcccccc@{}}
\toprule
$g_{\rm f}$ & abs & rel & ideal & 0-curr.\ sched.\ & rel-vs-abs & rel-vs-0curr\\
& [m] & [m] & [m] & [m] & (\%) & (\%)\\
\midrule
1.0 (nominal) & 0.2407 & 0.2407 & 0.2407 & 0.2407 & $\phantom{+}0.0$ & $\phantom{+}0.0$\\
0.4 (stress)  & 0.6383 & 0.6217 & 0.6221 & 0.6383 & $+2.6$ & $+2.6$\\
0.2 (weak)    & 1.5844 & 1.4342 & 1.4606 & 1.5844 & $+9.5$ & $+9.5$\\
\bottomrule
\end{tabular}
\end{adjustbox}
\end{table}

\begin{remark}[Nominal and stress-regime ablations]
\label{rem:nominal_vs_stress}
In Table~\ref{tab:ablation}, the ``no observer'' row sets
$\hat{\bm\nu}_c=\bm 0$ everywhere, including the feedforward, so the
surge-damping cancellation~\eqref{eq:tau_ff} acts at $\bm\nu$ rather
than $\bm\nu_r$ and RMS-$p_e$ rises to $4.038$~m. In
Table~\ref{tab:stress}, ``zero-current sched.'' sets
$\hat{\bm\nu}_c=\bm 0$ only inside the LPV scheduling map while keeping
$\hat{\bm\nu}_r$ in the feedforward; it therefore coincides with the
``absolute'' scheduling column, and ``rel-vs-0curr'' measures the
scheduling-side benefit alone. The $g_{\rm f}=1.0$ row matches the
relative row of Table~\ref{tab:ablation} ($0.2407$~m).
\end{remark}
\begin{remark}[Stress sweep versus feedforward--scheduling ablation]
\label{rem:stress_vs_ffsched}
Tables~\ref{tab:stress} and~\ref{tab:ffsched} answer different
questions and are not row-comparable: the former scales feedback
authority while keeping the relative-velocity feedforward active (the
``abs''/``rel'' labels switch only the scheduling argument), whereas
the latter switches feedforward and scheduling arguments independently.
At $g_{\rm f}=0.2$, Table~\ref{tab:stress} shows a $+9.5\%$
scheduling-side benefit (weakened feedback raises sensitivity to the
residual), while Table~\ref{tab:ffsched} keeps REL--REL and REL--ABS
within $\approx 0.03\%$ and identifies the feedforward as dominant.
Both views are consistent with Theorem~\ref{thm:cl}.
\end{remark}

\subsubsection{Feedforward versus scheduling: source of the benefit}
\label{subsec:ffsched_disentanglement}
To separate the two architectural routes, Table~\ref{tab:ffsched}
reports a $2\times2$ ablation in which the feedforward and LPV
scheduling arguments are independently set to absolute (ABS) or
relative (REL), on the helix segment of S2 at two
feedback-authority levels.

\begin{table}[!t]
\centering
\caption{Feedforward $\times$ scheduling disentanglement on the
helix segment of S2 under scheduled $\Hcal_\infty$ correction;
feedforward and scheduling arguments independently set to absolute
(ABS) or relative (REL). RMS values use the scheduled-$\Hcal_\infty$
additive realisation over the helix segment and therefore differ from
the full-horizon cascade RMS of Table~\ref{tab:metrics}.}
\label{tab:ffsched}
\begin{adjustbox}{max width=\columnwidth}
\setlength{\tabcolsep}{4pt}\small
\begin{tabular}{@{}lccc@{}}
\toprule
\textbf{FF--SCHED} & RMS $|p_e|$ [m] & HRR [\%] & CER [\%]\\
\midrule
\multicolumn{4}{l}{\textit{Feedback authority $g_{\rm f}=1.00$}}\\
ABS--ABS                       & 4.0473 & 98.3 & 93.8\\
REL--ABS                       & 0.2247 & 99.5 & 93.5\\
ABS--REL                       & 4.0467 & 94.7 & 92.1\\
REL--REL (proposed)            & 0.2247 & 99.5 & 93.5\\
\midrule
\multicolumn{4}{l}{\textit{Feedback authority $g_{\rm f}=0.20$}}\\
ABS--ABS                       & 14.7489 & 92.9 & 91.3\\
REL--ABS                       & \phantom{0}0.3644 & 99.5 & 93.5\\
ABS--REL                       & 14.5815 & 92.5 & 91.3\\
REL--REL (proposed)            & \phantom{0}0.3643 & 99.5 & 93.5\\
\bottomrule
\end{tabular}
\end{adjustbox}
\end{table}

The dominant gain comes from the observer-assisted
relative-velocity feedforward, not the LPV scheduling argument.
Whenever the feedforward uses the absolute velocity (ABS--ABS and
ABS--REL), the RMS tracking error stays near the no-observer level
($\approx 4$~m at $g_{\rm f}=1.0$, $\approx 15$~m at $g_{\rm f}=0.2$);
the relative feedforward (REL--ABS, REL--REL) drops it
to $0.2247$~m and $\approx 0.364$~m respectively. Switching only
the scheduling argument has a
secondary effect: REL--ABS and REL--REL are indistinguishable at
$g_{\rm f}=1.0$ and differ by $\approx 0.03\%$ at $g_{\rm f}=0.2$,
consistent with Table~\ref{tab:stress} read in the sense of
Remark~\ref{rem:stress_vs_ffsched}.
The proposed architecture is therefore best read as an
observer-assisted relative-velocity feedforward controller with an
LPV-$\Hcal_\infty$ robust correction layer.

\subsubsection{Persistent-excitation diagnostic for Lemma~\ref{lem:obs_exc}}
The directional-Gramian excitation condition~\eqref{eq:E2} of
Lemma~\ref{lem:obs_exc} is verified numerically in the Supplementary
Material (Sec.~S1, Fig.~S1, Table~S1): the manoeuvring helix segments
yield $\alpha_\Gamma\sim 3$--$5\times10^{-4}$, sufficient for transverse
current recovery, whereas the straight-line cruise $t\in[0,90]$~s has
$(q,r)\approx0$ and $\lambda_{\min}(\Gamma)\approx0$, so the observer
certificate of Theorem~\ref{thm:obs} applies only on the manoeuvring
sub-interval; the near-singular Gramian there serves as an online
diagnostic for the directional-persistent-excitation failure mode
(Supplementary Material).
\rev{On such low-excitation segments the transverse components
$(v_c,w_c)$ are not certifiably recoverable, while the surge component
$u_c$ stays identifiable through the dynamic residual under (E1); the
observer therefore degrades gracefully to surge-only current estimation
rather than diverging. Actively shaping the manoeuvre to keep
$\lambda_{\min}(\Gamma)$ above a prescribed floor (excitation-aware
guidance) is a natural way to restore the transverse-current certificate
on long straight legs and is left to future work.}

\subsubsection{Actuator-authority stress test}
\label{subsec:saturation_stress}

Theorem~\ref{thm:cl} assumes the unsaturated cone of
Assumption~\ref{ass:operset}~(A4b). In S2, the surge channel is
saturated for $17.2\%$ of the horizon even at the nominal
$U_{\max}=500$~N; this case therefore lies outside the formal
certificate, although the nonlinear simulation remains practically
stable with RMS-$p_e=0.241$~m. To quantify the degradation under
limited actuator authority, $U_{\max}$ is swept from $500$ to
$300$~N (Table~\ref{tab:saturation}). Reducing the
limit to $400$~N raises RMS-$p_e$ to $2.724$~m, while $300$~N drives
the surge channel into complete saturation
($\mathrm{Sat}=100\%$) and RMS-$p_e$ to $7.484$~m. Fin
saturation stays zero throughout, so the
degradation stems from insufficient surge-thrust authority rather
than pitch/yaw moment limits, evidencing engineering robustness
beyond the local unsaturated
certificate. An anti-windup extension is the natural
follow-up discussed in Section~\ref{sec:conclusion}.

\begin{table}[!htbp]
\centering
\caption{Actuator-authority stress test on the helix S2 segment:
surge-thrust limit $U_{\max}$ swept. ``Surge sat.'' denotes the
clipped-horizon fraction.}
\label{tab:saturation}
\begin{adjustbox}{max width=\columnwidth}
\setlength{\tabcolsep}{4pt}\small
\begin{tabular}{@{}cccc@{}}
\toprule
$U_{\max}$ [N] & RMS $|p_e|$ [m] & Surge sat.\ [\%] & Fin sat.\ [\%]\\
\midrule
500 (nominal) & 0.241 &  17.2 & 0.0\\
400           & 2.724 &  61.6 & 0.0\\
300           & 7.484 & 100.0 & 0.0\\
\bottomrule
\end{tabular}
\end{adjustbox}
\end{table}

\subsubsection{Further robustness and diagnostic studies}
Five additional studies (Supplementary
Material) corroborate the analysis.
\emph{(i)}~Sweeping the observer time scale
$\epsilon\in\{0.02,0.05,0.1,0.2\}$ under calibrated noise leaves RMS
tracking nearly unchanged ($\epsilon=0.05$ used), and a six-case Monte
Carlo study ($N=100$ per case, Sec.~S2) shows graceful degradation:
$0\%$ failures up to $\pm20\%$ hydrodynamic uncertainty and isolated
failures only at $\pm30\%$ (case M4, $8\%$) and the combined case
(M6, $12\%$), consistent with the local practical-UUB reading of
Theorem~\ref{thm:cl}. \emph{(ii)}~A radial initial-offset sweep
(Sec.~S3) finds the controller admissible up to
$|p_e(0)|\approx10$~m ($\mathrm{RMS}|p_e|\le0.8$~m), so the main-study
offset $\approx5.4$~m lies well inside the empirical estimate of the
invariant sublevel set $\Omega_{c_0}$ of Theorem~\ref{thm:cl}.
\emph{(iii)}~A current magnitude/heading sweep (Sec.~S4) delimits an
authority-limited regime ($\mathrm{RMS}|p_e|=2.0$--$4.9$~m) only for
$\|\bm V_c^n\|\gtrsim0.6$~m/s, the limit being set by surge-thrust
authority rather than by observer or LPV-correction performance.
\emph{(iv)}~A realisation comparison (Sec.~S5) confirms that the LMI
gain $K(\bm\rho)$ is most effective as an additive correction around
the pre-stabilised cascade ($0.216$--$0.225$~m) rather than as a
standalone direct LMI controller ($1.894$~m). \emph{(v)}~The
bounded-real certificate of Theorem~\ref{thm:lmi} is verified along
the embedded LPV trajectory (Sec.~S6): the dissipation residual
$D(t)=\dot V_e+\bm z^{\!\top}\bm z
-\gamma_{\rm LMI}^2\bm w_{\rm aug}^{\!\top}\bm w_{\rm aug}$
satisfies $D(t)\le0$ for $100\%$ of samples
($\max_t D(t)=-1.93\times10^{-2}$).

\subsection{Independent baseline comparison}\label{subsec:comparison}

The proposed architecture is compared with three independent
baselines under matched conditions. The no-observer geometric backbone is
reported as an \emph{EMO-cascade baseline inspired by}
\citet{li2023trajectory} rather than a direct reproduction, since
the present simulations use the REMUS hydrodynamic model with
harmonised saturation limits. The baseline corresponds to the
no-observer row of Table~\ref{tab:ablation}
(RMS-$p_e=4.038$~m on S2).
\rev{To keep the comparison fair, every baseline is tuned
independently under the same protocol as the proposed controller:
identical plant, reference, actuator limits, sensor noise, and current
scenarios, and an equal parameter-search budget per controller, with
each controller's gains selected to minimise its own nominal-scenario
RMS tracking error before the cross-scenario evaluation. No baseline is
handicapped by transferring another controller's gains; each is
reported at its own best-found settings.}

\rev{Five} controllers are compared:
\textbf{(BB)} an EMO-cascade baseline inspired by
\citet{li2023trajectory} without a current observer
($\hat{\bm\nu}_c=\bm 0$);
\textbf{(DOB)} the same cascade with an additive disturbance observer
for the lumped current-induced surge force, not fed into the LPV
schedule;
\rev{\textbf{(ESO-LOS)} an independent, fairly-tuned reduced-order
extended-state-observer baseline that estimates the body-frame current
as the extended state of the relative-velocity dynamics and compensates
it in the same cascade;}
\textbf{(LPV-abs)} the proposed observer-assisted feedforward and
LPV-$\Hcal_\infty$ correction with the schedule at $\bm\nu$; and
\textbf{(LPV-rel)} the full proposed architecture with the schedule at
$\hat{\bm\nu}_r$.

Table~\ref{tab:comparison} separates three effects. \emph{(i)
Combined observer+LPV contribution.} Relative to (BB), (LPV-rel)
reduces RMS tracking error from $4.038$~m to $0.241$~m on S2 and
from $5.514$~m to $1.060$~m on S4 ($\approx\!17\times$ and
$\approx\!5\times$); the bulk reflects introducing current
estimation at all. \emph{(ii) Marginal LPV contribution.} Against
(DOB), the LPV correction reduces S4 RMS from $1.484$~m to
$1.060$~m ($\approx\!28\%$), quantifying the benefit of robust
correction beyond lumped disturbance compensation.
\rev{The ESO-LOS baseline, which uses an independent
reduced-order extended-state observer, is stronger than the additive DOB
($0.315$/$1.180$~m on S2/S4 versus $0.523$/$1.484$~m), yet the proposed
architecture still improves upon it by about $24\%$ on S2 and $10\%$ on
S4, indicating that the three-stage relative-velocity observer
outperforms a single reduced-order ESO of comparable bandwidth.}
\emph{(iii) Relative vs.\ absolute scheduling.} (LPV-abs) and
(LPV-rel) differ by $<10^{-3}$~m in the nominal regime
(cf.\ Section~\ref{subsec:tracking}); separation appears only under
reduced feedback authority (Section~\ref{subsec:robust}). These trends
are summarised in Fig.~\ref{fig:comparison}, which contrasts the RMS
tracking error and the hydrodynamic-residual reduction of the \rev{five}
controllers on the nominal (S2) and adverse (S4) scenarios.

\begin{table*}[!t]
\centering
\caption{Independent baseline comparison on the helix trajectory
(harmonised plant, reference, actuator limits, sensor noise, and
current scenarios); the two LPV rows share the relative-velocity
feedforward and differ only in the scheduling argument. RMS on
window~W1.}
\label{tab:comparison}
\begin{adjustbox}{max width=\textwidth}
\setlength{\tabcolsep}{2.5pt}\small
\begin{tabular}{@{}llccccc@{}}
\toprule
\textbf{Sc.} & \textbf{Controller} & RMS $|p_e|$ [m] &
Max $|p_e|$ [m] & HRR [\%] & Sat.\ [\%] & $J_\tau$ [$10^7$]\\
\midrule
\multirow{5}{*}{S2}
& EMO-cascade (BB)                  & 4.038 & 6.257 & 0.0  & 29.5 & 2.858\\
& Additive DOB                      & 0.523 & 0.846 & 99.5 & 6.4  & 2.555\\
& \rev{ESO-LOS (reduced-order)}     & \rev{0.315} & \rev{0.382} & \rev{96.0} & \rev{22.1} & \rev{2.612}\\
& LPV abs-sched (rel FF)            & 0.241 & 0.311 & 99.5 & 17.2 & 2.595\\
& LPV rel-sched (proposed, rel FF)  & 0.241 & 0.310 & 99.5 & 17.2 & 2.595\\
\midrule
\multirow{5}{*}{S4}
& EMO-cascade (BB)                  & 5.514 & 8.705 & 0.0  & 36.5 & 3.151\\
& Additive DOB                      & 1.484 & 2.896 & 99.5 & 35.3 & 2.670\\
& \rev{ESO-LOS (reduced-order)}     & \rev{1.180} & \rev{2.571} & \rev{95.9} & \rev{40.8} & \rev{2.695}\\
& LPV abs-sched (rel FF)            & 1.061 & 2.404 & 99.5 & 39.0 & 2.686\\
& LPV rel-sched (proposed, rel FF)  & 1.060 & 2.402 & 99.5 & 39.0 & 2.685\\
\bottomrule
\multicolumn{7}{l}{\footnotesize $J_\tau=\int\|\bm\tau\|^2dt$ in $10^7$, summed across actuators (mixed N$^2$s and N$^2$m$^2$s units).}\\
\end{tabular}
\end{adjustbox}
\end{table*}

\begin{figure}[!htbp]
\centering
\includegraphics[width=1\columnwidth]{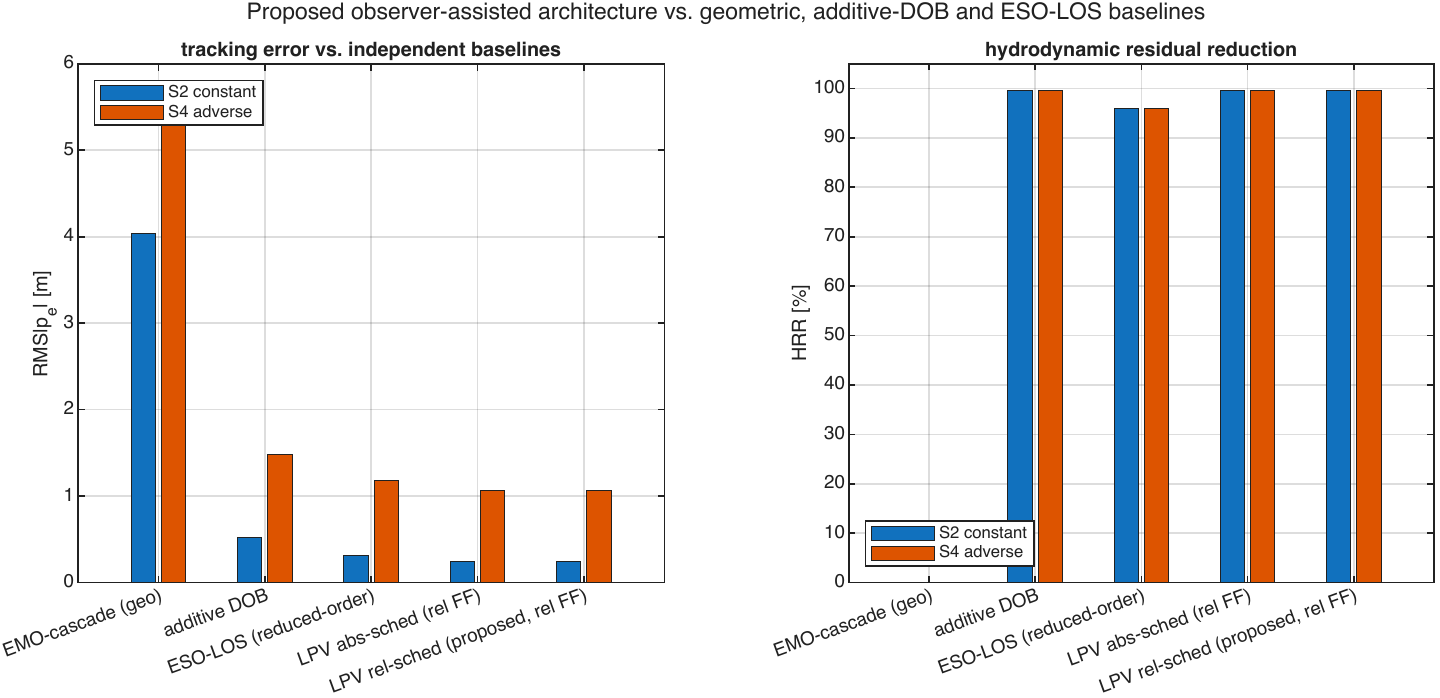}

\caption{Baseline comparison on S2 and S4: EMO-cascade,
additive-DOB, \rev{ESO-LOS,} LPV-absolute, and proposed LPV-relative.}
\label{fig:comparison}
\end{figure}
\subsection{Synthesis of findings and certificate scope}
\label{subsec:results_discussion}
The numerical study leads to four observations. \emph{(i)~The
observer-assisted relative-velocity feedforward is the dominant
performance mechanism.} The five-way ablation
(Table~\ref{tab:ablation}) and the FF--scheduling disentanglement
(Table~\ref{tab:ffsched}) show that removing the observer or
using an absolute-velocity feedforward
drives the tracking error to the no-observer order, so
the main current rejection occurs through hydrodynamic
compensation at $\hat{\bm\nu}_r$ rather than through LPV scheduling.

\emph{(ii)~The LPV-$\Hcal_\infty$ correction is a certified
robustness mechanism, not the main nominal performance source.}
In the well-tuned regime relative and absolute scheduling are
indistinguishable, the residual-level benefit
being masked by the small closed-loop disturbance-to-error gain;
under reduced feedback authority the scheduling-side
benefit reaches $9.5\%$ RMS improvement at
one-fifth nominal feedback gain, consistent with
Theorem~\ref{thm:cl}.

\emph{(iii)~The formal certificate and the nonlinear implementation
have different scopes.} Theorems~\ref{thm:lmi}
and~\ref{thm:cl} certify the embedded LPV model under the
unsaturated-cone, embedding-margin, and observer-excitation
conditions; the nonlinear REMUS simulations additionally include
actuator saturation, sensor noise, current variation, and model
uncertainty. The observed bounded response is therefore engineering
validation beyond the local certificate, not a proof of global
stability for the saturated nonlinear AUV.

\emph{(iv)~The practical limit is actuator authority.} The operating
envelope (Fig.~S2 and Table~S3 of the Supplementary Material) shows
accurate tracking for
moderate currents and degradation for larger ones driven by
surge-thrust saturation rather than observer or LPV-correction
failure; an anti-windup or sector-bounded saturation extension is
the most direct path to a saturated formal guarantee. The grid-certified
embedding residual is essentially zero,
but the conservative Lipschitz-corrected uniform bound does not
satisfy the analytical embedding margin; a tighter continuous-polytope
verification via sum-of-squares or branch-and-bound certification,
or a multiplier-based formulation using Petersen's lemma, would
strengthen the certificate without changing the implemented
architecture.

%==============================================================
\section{Conclusion}\label{sec:conclusion}
%==============================================================
\rev{This paper developed an observer-assisted relative-velocity control
architecture with an LPV-$\Hcal_\infty$ robust correction layer for 3D
trajectory tracking of underactuated non-minimum-phase AUVs under unknown
ocean currents. Its
theoretical results are the directional-Gramian excitation condition, the
residual-level surge break-even law, the convex constant-$B_u$
LPV-$\Hcal_\infty$ synthesis, and the local practical-UUB certificate
under the stated excitation, embedding-margin, small-gain, and
unsaturated-cone conditions.}

REMUS simulations confirm these findings: the observer reduces the
current-estimation residual by $89$--$96\%$; removing it raises the
constant-current RMS tracking error from $0.241$~m to $4.038$~m; and
against the harmonised EMO-cascade and additive disturbance-observer
baselines the proposed architecture attains $0.241$~m (constant
current) and $1.060$~m (adverse current).

The formal guarantees are local and conditional, applying to the
embedded LPV model under the stated observer-excitation,
embedding-margin, small-gain, and unsaturated-cone conditions; the
nonlinear saturated simulations are engineering validation beyond
the formal certificate, not a global stability proof. \rev{The
validation in this paper is entirely simulation-based, on a high-fidelity
REMUS model, and all comparative statements are made in that setting;
hardware-in-the-loop, tank, and in-water testing are necessary before
engineering superiority can be claimed.} Future work
will address anti-windup certification, tighter continuous-polytope
embedding verification (sum-of-squares or branch-and-bound), and
hardware-in-the-loop or in-water validation under realistic sensor
and actuator limitations.

%============================================================
\appendix
\section{Deferred proofs}\label{app:proofs}
\rev{This appendix collects the proofs of Lemma~\ref{lem:obs_exc}, Theorem~\ref{thm:obs}, Lemma~\ref{lem:relsched}, Theorem~\ref{thm:lmi}, and Theorem~\ref{thm:cl}. All equation and assumption references are to the main text.}
\subsection*{Proof of Lemma~\ref{lem:obs_exc} (sufficient excitation)}
\begin{proof}
Under Assumption~\ref{ass:current}~(C1), $\dot{\bm V}_c^n=\bm 0$ and
$R^{\!\top}\dot{\bm V}_c^n=\bm 0$, so the body-frame current obeys the
purely rotational $\dot{\bm\nu}_c=-\bm\omega_b\times\bm\nu_c$ with
$\bm\omega_b=[0,q,r]^{\!\top}$. Expanding the cross product yields, in
components,
\begin{equation}\label{eq:cdyn}
\dot u_c=rv_c-qw_c,\quad\dot v_c=-ru_c,\quad\dot w_c=qu_c.
\end{equation}
The augmented observable state is
$\bm x_o=[v,w,u_c,v_c,w_c]^{\!\top}\in\R^5$. The observer outputs
through the kinematic identity $v=v_r+v_c$ etc.\ and through the
dynamic residual $r_u$ in~\eqref{eq:rsurge}.

\textit{Step 1: surge identifiability via (E1).}
The surge dynamic residual~\eqref{eq:rsurge} at the operating point
$\bm x_o^\circ$ admits the first-order Taylor expansion in
$\tilde u_c=u_c-\hat u_c$,
\[
r_u(\hat u_c)=r_u(u_c)+\frac{\partial r_u}{\partial u_c}\Big|_{u_c}\tilde u_c
+\mathcal O(\tilde u_c^2).
\]
Substituting~\eqref{eq:rsurge} with $\hat u_r=u-\hat u_c=u_r+\tilde u_c$
and using $\phi'(s)=2|s|$,
\begin{align*}
\frac{\partial r_u}{\partial \hat u_c}
&=-\frac{\partial}{\partial\hat u_c}\Bigl[\frac{X_u(u-\hat u_c)+X_{u|u|}\phi(u-\hat u_c)}{m_u}\Bigr]\\
&=\frac{X_u+2X_{u|u|}|u_r^\circ|}{m_u},
\end{align*}
so equivalently
$\partial r_u/\partial u_c=-(X_u+2X_{u|u|}|u_r^\circ|)/m_u$.
Under~(E1), $|X_u+2X_{u|u|}|u_r^\circ||\ge\sigma_d>0$, so this
partial derivative is bounded away from zero. The implicit-function
theorem applied to $r_u(\bm x_o)=0$ at the operating point thus
yields $u_c$ as a locally unique smooth function of the other state
components, i.e.\ $u_c$ is locally identifiable from $r_u$.

\textit{Step 2: transverse identifiability via (E2).}
With $u_c$ identified at each instant from Step~1, the remaining
unknowns are $(v_c,w_c)$. The Stage-3 filter recovers them through
two coupled channels:
\begin{itemize}
\item[(a)] The rotational coupling in $\dot u_c$: from~\eqref{eq:cdyn}
linearised about $\bm x_o^\circ$,
$\delta\dot u_c=r(\sigma)\delta v_c-q(\sigma)\delta w_c
=[r(\sigma),-q(\sigma)]\,[\delta v_c,\delta w_c]^{\!\top}$,
so at time $\sigma$ the surge-channel measurement reveals only the
projection of $(\delta v_c,\delta w_c)$ onto the direction
$\bm d(\sigma):=[r(\sigma),-q(\sigma)]^{\!\top}\in\R^2$.
\item[(b)] The transverse kinematic innovations
$\hat v_{\rm kin}-\hat v_r-\hat v_c$ and
$\hat w_{\rm kin}-\hat w_r-\hat w_c$: these provide direct
linear-in-$\bm\nu_c$ measurements, but only modulo the sway/heave
ambiguity inherent in the kinematic identity (which is resolved by
the Stage-2 dynamic observer).
\end{itemize}

The directional channel~(a) is recoverable as a two-dimensional
quantity only if $\bm d(\sigma)$ does not lie in a single
one-dimensional subspace over the observation window $[t,t+T]$. To
quantify this, integrate the outer product $\bm d(\sigma)\bm d^{\!\top}(\sigma)$
over $[t,t+T]$:
\begin{equation}\label{eq:Gamma_def}
\Gamma(t,t+T):=\int_t^{t+T}\!\!\bm d(\sigma)\bm d^{\!\top}(\sigma)\,d\sigma
=\int_t^{t+T}\!\!\begin{bmatrix} r^2 & -rq\\ -rq & q^2\end{bmatrix}\!d\sigma.
\end{equation}
The matrix $\Gamma$ is positive semi-definite by construction, and
its kernel is precisely $\{\bm v\in\R^2:\bm d(\sigma)^{\!\top}\bm v=0\;\forall\sigma\in[t,t+T]\}$
(in $L^2$). Hence $\Gamma\succeq\alpha I_2$ for some $\alpha>0$ over
$[t,t+T]$ if and only if no direction in $\R^2$ is orthogonal to
$\bm d(\sigma)$ on a set of positive measure, which is precisely the
condition that the surge-channel measurements span the entire
transverse plane $\R^2$ over the window. This is the standard
persistent-excitation condition adapted to the directional channel
\citep[Sec.~13.4]{khalil2002nonlinear}, and is~\eqref{eq:E2}.

\textit{Step 3: synthesis.}
With (E1) supplying local surge identifiability and (E2) ensuring
that the transverse correction direction $\bm d(\sigma)=[r,-q]^{\!\top}$
spans $\R^2$ over the window $[t,t+T]$, the implemented
correction channel for $(v_c,w_c)$ is uniformly informative over
the window. Together with the local input-to-state
stability of the Stage-2 sub-observer \rev{(Lemma~\ref{lem:stage2_iss})}
with respect to its driving terms, this establishes the sufficient
excitation condition used in Theorem~\ref{thm:obs}.
\end{proof}
\subsection*{Proof of Theorem~\ref{thm:obs} (observer convergence)}
\begin{proof}
The proof is organised in three independent blocks: the Stage-1
high-gain differentiator (fast time scale $\tau=t/\epsilon$), and
the Stage-2/Stage-3 dynamic observer for the relative-velocity and
current variables (normal time scale $t$). These two scales must be
analysed separately to avoid the rate inflation pointed out in
Remark~\ref{rem:slow_rate}.

\textit{Step 1.}
Differentiating $\tilde{\bm\eta}_p^{(0)}=\bm\eta_p-\hat{\bm\eta}_p^{(0)}$
with the dynamics~\eqref{eq:hgd} and applying the change of variable
$\bm\eta_1=\epsilon^{-1}\tilde{\bm\eta}_p$ gives, in the fast time
$\tau=t/\epsilon$,
\begin{equation}\label{eq:eta1_dyn}
\frac{d\bm\eta_1}{d\tau}=A_1\bm\eta_1+\epsilon\,\bm b_1(\ddot{\bm\eta}_p),
\quad A_1=\begin{bmatrix}-\alpha_1 & I_3\\-\alpha_2 & 0\end{bmatrix},
\end{equation}
which is Hurwitz by the choice $\alpha_1,\alpha_2>0$. A standard
quadratic Lyapunov argument
\citep[Sec.~14.5]{khalil2002nonlinear} gives, in physical time,
\[
\|\bm\eta_1(t)\|\le\kappa_h\,e^{-\lambda_h t/\epsilon}\|\bm\eta_1(0)\|+\mathcal O(\epsilon),
\]
yielding~\eqref{eq:hgd_bound} as the bound (O1).

The output of Stage~1 enters Stages~2 and~3 as a
\emph{pseudo-measurement} of $\dot{\bm\eta}_p$ with bounded error
$\sigma_y$:
\begin{equation}\label{eq:sigma_y_def}
\sigma_y:=\sup_{t\in\mathcal I_{\rm PE}}\|\hat{\bm\nu}_{\rm kin}(t)-\bm\nu(t)\|
=\mathcal O(\epsilon)+\mathcal O(\sigma_{\rm noise}),
\end{equation}
which decouples Stage~1 from the slower dynamics below: only
$\sigma_y$ matters, not the internal Stage-1 transients.

\textit{Step 2.}
The Stage-2 error
$\bm\eta_2=[\tilde v_r,\tilde w_r]^{\!\top}$ follows directly
from~\eqref{eq:obs_vw}: subtracting $\dot{\hat v}_r$ from $\dot v$
and using row~2 of $G$ in~\eqref{eq:Gmat} yields
\begin{equation}\label{eq:eta2_dyn}
\frac{d\bm\eta_2}{dt}=-K_2\bm\eta_2+\Delta\bm f_2(\hat{\bm\nu}_r,\bm\nu_r,\theta)
+\bm n_y^{(2)}(t),
\end{equation}
with $K_2=\diag(k_v,k_w)\succ 0$, $\Delta\bm f_2$ a locally Lipschitz
mismatch with constant $L_2$ and $\Delta\bm f_2(\bm 0,\bm\nu_r,\theta)=0$,
and $\bm n_y^{(2)}$ collecting the pseudo-measurement noise from
Stage~1 (bounded by $\sigma_y$). The Stage-3 error
$\bm\eta_3=\tilde{\bm\nu}_c=\bm\nu_c-\hat{\bm\nu}_c$ follows
from~\eqref{eq:obs_c}--\eqref{eq:rsurge} together with the body-frame
current dynamics~\eqref{eq:nuc_dot}; on the \emph{normal} time scale
$t$ (no $\epsilon$ scaling),
\begin{equation}\label{eq:eta3_dyn}
\frac{d\bm\eta_3}{dt}=
-\bm\omega_b\times\bm\eta_3-K_3 M(t)\bm\eta_3
+R^{\!\top}\dot{\bm V}_c^n+\bm h_3(\bm\eta_1,\bm\eta_2,t),
\end{equation}
with $K_3=\diag(k_{cu},k_{cv},k_{cw})$, $M(t)$ the linearisation
matrix collecting the surge dynamic-residual sensitivity and the
transverse innovation chain (containing the direction
$[r(\sigma),-q(\sigma)]^{\!\top}$ from Lemma~\ref{lem:obs_exc}), and
$\bm h_3$ the coupling from upstream errors. Stacking
$\tilde{\bm\eta}_o=[\bm\eta_2^{\!\top},\bm\eta_3^{\!\top}]^{\!\top}$,
\begin{equation}\label{eq:eta_dot}
\dot{\tilde{\bm\eta}}_o=A_o(t)\tilde{\bm\eta}_o
+\Delta_o(\tilde{\bm\eta}_o,t)
+B_y\bm n_y(t)+B_c R^{\!\top}\dot{\bm V}_c^n,
\end{equation}
where $A_o(t)$ is the block-triangular time-varying matrix combining
$-K_2$ (Hurwitz) and the Stage-3 time-varying block, $\Delta_o$
collects the Lipschitz nonlinear mismatch with $\Delta_o(\bm 0,t)=0$,
$\bm n_y$ is the Stage-1 pseudo-measurement noise with
$\|\bm n_y\|\le\sigma_y$, and $B_c$ routes the current-rate
disturbance with $\|B_c R^{\!\top}\dot{\bm V}_c^n\|\le\bar c_1$. The
overall homogeneous system $\dot{\bm\eta}=A_o(t)\bm\eta$ is
\emph{uniformly exponentially stable in the PE sense}: $-K_2$ is
Hurwitz, the Stage-3 surge channel is exponentially stable under
(E1), and the transverse $(v_c,w_c)$ subspace is exponentially
stable under (E2) by standard PE machinery
\citep[Sec.~6.3]{khalil2002nonlinear}; we denote the corresponding
uniform Lyapunov matrix by $\mathcal P_o\succ 0$.

\textit{Step 3.}
The PE-Lyapunov construction
\citep[Sec.~6.3]{khalil2002nonlinear} guarantees the existence of a
uniform symmetric $\mathcal P_o(t)$ with
$\underline p_o I\preceq\mathcal P_o(t)\preceq\overline p_o I$ and
$\dot{\mathcal P}_o+A_o^{\!\top}\mathcal P_o+\mathcal P_o A_o
\preceq-q_o I$, where $q_o>0$ is bounded below by a constant
proportional to $\min\{\lambda_{\min}(K_2),\,\alpha_{\rm PE}^{1/2}\lambda_{\min}(K_3)\}$
(the $\alpha_{\rm PE}^{1/2}$ factor arises because the transverse
$(v_c,w_c)$ subspace inherits exponential decay through the directional
Gramian only after integration over the PE window). With
$V_o=\tilde{\bm\eta}_o^{\!\top}\mathcal P_o\tilde{\bm\eta}_o$, the
local Lipschitz constant $L_o$ of $\Delta_o$ on the compact operating
set, and Young's inequality
$2a^{\!\top}b\le\eta a^{\!\top}a+\eta^{-1}b^{\!\top}b$ applied to the
cross-terms with the noise $\bm n_y$ and the disturbance
$R^{\!\top}\dot{\bm V}_c^n$ with $\eta=q_o/(4\overline p_o^2)$,
\begin{equation}\label{eq:Vo_proof}
\dot V_o\le -\bigl(\tfrac{q_o}{2}-2\overline p_o L_o\bigr)\|\tilde{\bm\eta}_o\|^2
+\frac{4\overline p_o^2\|B_y\|^2}{q_o}\sigma_y^2
+\frac{4\overline p_o^2\|B_c\|^2}{q_o}\bar c_1^2.
\end{equation}
On the compact operating set, $L_o$ is bounded so that the
coefficient of $\|\tilde{\bm\eta}_o\|^2$ remains negative, giving the
slow decay rate
\begin{equation}\label{eq:lambda_c_def}
\lambda_c=\frac{q_o/2-2\overline p_o L_o}{\overline p_o},
\end{equation}
governed by the Stage-2/Stage-3 gains, the directional Gramian
$\alpha_{\rm PE}$ (through $q_o$), and the conditioning of
$\mathcal P_o$ --- \emph{not} by $1/\epsilon$. The recipe
$\alpha_{\rm PE}\downarrow\Rightarrow q_o\downarrow$ identifies the
$\alpha_{\rm PE}$-dependence of the noise coefficient: dividing the
right-hand-side floor of~\eqref{eq:Vo_proof} by $\underline p_o\lambda_c$
gives the steady-state bound on
$V_o/\underline p_o\sim\|\tilde{\bm\eta}_o\|^2$ proportional to
$\sigma_y^2/q_o\sim\sigma_y^2/\alpha_{\rm PE}^{1/2}$. Taking the
square root yields the sharper PE-noise scaling
$\|\tilde{\bm\eta}_o\|\sim\sigma_y/\alpha_{\rm PE}^{1/4}$. Under the
adopted normalisation $\alpha_{\rm PE}\in(0,1]$ on the operating
window, $\alpha_{\rm PE}^{-1/4}\le\alpha_{\rm PE}^{-1/2}$, so the
conservative monotone upper bound
\begin{equation}\label{eq:cy_def}
\frac{c_y\sigma_y}{\sqrt{\alpha_{\rm PE}}}
\ge\frac{c_y\sigma_y}{\alpha_{\rm PE}^{1/4}}
\end{equation}
is reported in~\eqref{eq:delta_o} for notational simplicity and to
keep a single monotone bound consistent with the form of the
Comparison-Lemma estimate below. Applying the Comparison
Lemma~\citep[Lem.~3.4]{khalil2002nonlinear} to~\eqref{eq:Vo_proof} then
yields
\[
\|\tilde{\bm\eta}_o(t)\|\le\kappa_c\,e^{-\lambda_c t}\|\tilde{\bm\eta}_o(0)\|
+c_\epsilon\epsilon+c_c\bar c_1+\frac{c_y\sigma_y}{\sqrt{\alpha_{\rm PE}}},
\]
which is the bound (O2) in~\eqref{eq:obs_bound}, with the
$c_\epsilon\epsilon$ contribution arising from
$\sigma_y=\mathcal O(\epsilon)+\mathcal O(\sigma_{\rm noise})$
of~\eqref{eq:sigma_y_def}.
\end{proof}
\subsection*{Proof of Lemma~\ref{lem:relsched} (break-even law)}
\begin{proof}
The proof proceeds by (i) deriving the residual form, (ii) bounding
the drag Lipschitz constant, (iii) bounding each case, and (iv)
deriving the strict break-even constant $\chi$ via a first-order
expansion of the residuals.

\textit{Step (i): residual form.}
The surge equation in~\eqref{eq:tracking_dyn} contains the surge
damping term $-(X_u u_r+X_{u|u|}u_r|u_r|)/m_u$ from the
hydrodynamics evaluated at the true relative velocity. The
feedforward $\bm\tau_{\rm ff}$ in~\eqref{eq:tau_ff} cancels the
damping evaluated at $\xi$, namely $-(X_u\xi+X_{u|u|}\xi|\xi|)/m_u$.
Their difference is the residual entering~\eqref{eq:lpv_err}:
\begin{equation}\label{eq:eds_form}
e_{\rm ds}=\frac{X_u(\xi-u_r)+X_{u|u|}(\xi|\xi|-u_r|u_r|)}{m_u}.
\end{equation}

\textit{Step (ii): drag Lipschitz constant.}
Let $\phi(s)=s|s|$. Then $\phi'(s)=2|s|$ wherever differentiable
(and $\phi$ is differentiable on $\R$ with $\phi'(0)=0$, since
$\phi(s)=s^2\,\mathrm{sgn}(s)$). On $\mathcal I_u$,
$|\phi'(s)|\le 2\bar u_{r,\max}$, so the mean-value theorem gives,
for any $a,b\in\mathcal I_u$,
\begin{equation}\label{eq:phi_mvt}
|\phi(a)-\phi(b)|=\Bigl|\!\int_b^{a}\phi'(s)\,ds\Bigr|\le 2\bar u_{r,\max}|a-b|.
\end{equation}
Combining the linear and quadratic damping contributions in
\eqref{eq:eds_form},
\begin{align*}
|e_{\rm ds}|&\le\frac{|X_u||\xi-u_r|+|X_{u|u|}|\,2\bar u_{r,\max}|\xi-u_r|}{m_u}\\
&=\frac{L_{\rm drag}}{m_u}|\xi-u_r|,
\end{align*}
which yields $L_{\rm drag}=|X_u|+2|X_{u|u|}|\bar u_{r,\max}$
in~\eqref{eq:Ldrag}.

\textit{Step (iii.a): absolute scheduling, $\xi=u$.}
Substituting $\xi=u$ into~\eqref{eq:eds_form} and using
$u-u_r=u_c$ from the definition $u_r=u-u_c$:
\begin{align*}
e_{\rm ds}^{\rm abs}&=\frac{X_u(u-u_r)+X_{u|u|}(u|u|-u_r|u_r|)}{m_u}\\
&=\frac{X_u u_c+X_{u|u|}\bigl(u|u|-u_r|u_r|\bigr)}{m_u}.
\end{align*}
The \emph{upper} bound in~\eqref{eq:e_ds_abs} follows from Step~(ii)
with $a=u,b=u_r,|a-b|=|u_c|$, giving
$|e_{\rm ds}^{\rm abs}|\le L_{\rm drag}|u_c|/m_u$. For the
\emph{lower} bound, factor the leading term: applying the
mean-value theorem to the function
$h(\xi)=X_u\xi+X_{u|u|}\phi(\xi)$ between $\xi=u_r$ and $\xi=u$
gives $h(u)-h(u_r)=h'(\bar\xi)\,u_c$ for some $\bar\xi\in[u_r,u]$,
with $h'(\bar\xi)=X_u+2X_{u|u|}|\bar\xi|$. Since
$|\bar\xi|\le\bar u_{r,\max}$ on $\mathcal I_u$, the
identifiability bound \textnormal{(A4a)} of
Assumption~\ref{ass:operset} gives
$|h'(\bar\xi)|\ge m_{\rm drag}>0$, hence
\[
|e_{\rm ds}^{\rm abs}|=\frac{|h'(\bar\xi)|}{m_u}|u_c|\ge\frac{m_{\rm drag}}{m_u}|u_c|.
\]
Combined with the upper bound, this gives the two-sided
bound~\eqref{eq:e_ds_abs}.

\textit{Step (iii.b): estimated relative scheduling, $\xi=\hat u_r$.}
Substituting $\xi=\hat u_r$ into~\eqref{eq:eds_form} and using
$\hat u_r-u_r=-\tilde u_c$:
\begin{align*}
e_{\rm ds}^{\rm rel}&=\frac{X_u(\hat u_r-u_r)+X_{u|u|}(\hat u_r|\hat u_r|-u_r|u_r|)}{m_u}\\
&=\frac{-X_u\tilde u_c+X_{u|u|}\bigl(\hat u_r|\hat u_r|-u_r|u_r|\bigr)}{m_u}.
\end{align*}
The upper bound \eqref{eq:e_ds_rel} follows from Step~(ii) with
$a=\hat u_r,b=u_r,|a-b|=|\tilde u_c|$.

\textit{Step (iv): rigorous ratio from two-sided bounds.}
Dividing the upper bound in~\eqref{eq:e_ds_rel} by the lower bound
in~\eqref{eq:e_ds_abs} yields
\[
\frac{|e_{\rm ds}^{\rm rel}|}{|e_{\rm ds}^{\rm abs}|}
\le\frac{L_{\rm drag}|\tilde u_c|/m_u}{m_{\rm drag}|u_c|/m_u}
=\frac{L_{\rm drag}}{m_{\rm drag}}\cdot\frac{|\tilde u_c|}{|u_c|},
\]
which is~\eqref{eq:ratio}. The strict break-even condition
$|e_{\rm ds}^{\rm rel}|<|e_{\rm ds}^{\rm abs}|$ is therefore
\emph{rigorously guaranteed} whenever
$|\tilde u_c|<(m_{\rm drag}/L_{\rm drag})|u_c|$, which is~\eqref{eq:chi}
with $\chi=m_{\rm drag}/L_{\rm drag}\in(0,1]$. The ratio $\chi$
collapses to $1$ when the damping is purely linear ($X_{u|u|}=0$,
hence $m_{\rm drag}=L_{\rm drag}=|X_u|$) and shrinks as the
quadratic-damping contribution grows relative to the linear one.
\end{proof}

\subsection*{Proof of Theorem~\ref{thm:lmi} (bounded-real LPV synthesis)}
\begin{proof}
The proof proceeds in three steps: convex extension from vertices to
the polytope under constant $B_u$, Schur-complement reduction, and
integration of the dissipation inequality.

\textit{Step 1.}
With $P=Y^{-1}\succ 0$, $V_e=\bm x_e^{\!\top}P\bm x_e$, and
$\delta\bm v_{\Hcal_\infty}=K(\bm\rho)\bm x_e$, the closed-loop
matrix $A_{\rm cl}(\bm\rho)=A_{\rm pre}(\bm\rho)+B_uK(\bm\rho)$
satisfies
\begin{align}\label{eq:Acl_X_expand}
A_{\rm cl}(\bm\rho)Y+&YA_{\rm cl}^{\!\top}(\bm\rho)\notag\\
=~&A_{\rm pre}(\bm\rho)Y+YA_{\rm pre}^{\!\top}(\bm\rho)\notag\\
&+B_uK(\bm\rho)Y+YK^{\!\top}(\bm\rho)B_u^{\!\top}\notag\\
=~&\sum_{i}\lambda_i(\bm\rho)\bigl[A_iY+YA_i^{\!\top}\notag\\
&\hspace{2.3em}+B_uW_i+W_i^{\!\top}B_u^{\!\top}\bigr],
\end{align}
where the second equality used $K(\bm\rho)Y=\sum_i\lambda_i(\bm\rho)W_i$
(by~\eqref{eq:K_rho_def}, since
$K(\bm\rho)\,Y=(\sum_i\lambda_i W_i)Y^{-1}Y$).
The right-hand side is affine in $\lambda_i(\bm\rho)$
because $B_u$ is independent of $i$; no cross-term $B_iW_j$ ($i\neq j$)
arises.

\textit{Step 2.}
The standard bounded-real lemma \citep[Sec.~3.7]{scherer2000linear} applied to the
closed loop, after pre- and post-multiplication by $\diag(Y,I,I)$
(which preserves negative-definiteness since $Y\succ 0$), yields
\begin{equation}\label{eq:BRL_expand}
\begin{bmatrix}
A_{\rm cl}(\bm\rho)Y+YA_{\rm cl}^{\!\top}(\bm\rho) & \tilde E(\bm\rho) & (C_zY+D_zK(\bm\rho)Y)^{\!\top}\\
\star & -\gamma I & 0\\
\star & \star & -\gamma I
\end{bmatrix}\!\!\prec\! 0.
\end{equation}
By Step~1, the $(1,1)$-block is $\sum_i\lambda_i\Phi_i$, the
$(1,2)$-block is $\sum_i\lambda_i\tilde E_i$, and the $(1,3)$-block
is $\sum_i\lambda_i(C_zY+D_zW_i)^{\!\top}$, all affine in
$\lambda_i(\bm\rho)$. The vertex LMI~\eqref{eq:lmi} therefore extends
by convexity to the entire polytope.

\textit{Step 3.}
Pre- and post-multiplying~\eqref{eq:BRL_expand} by $\diag(P,I,I)$,
$P=Y^{-1}$, converts it to a bounded-real LMI on $P$. Two
applications of the Schur complement \citep[Sec.~2.1]{boyd1994lmi} yield,
for any $\bm w_{\rm aug}$,
\begin{align}\label{eq:diss_ineq}
\dot V_e
&=\bm x_e^{\!\top}(A_{\rm cl}^{\!\top}P+PA_{\rm cl})\bm x_e
+2\bm x_e^{\!\top}P\tilde E\bm w_{\rm aug}\notag\\
&\le-\tfrac{1}{\gamma}\bm z^{\!\top}\bm z+\gamma\bm w_{\rm aug}^{\!\top}\bm w_{\rm aug}.
\end{align}
Integrating from $t=0$ to $t=T$ with zero initial condition
$V_e(0)=0$ and taking $T\to\infty$ yields the integrated form
\[
\tfrac{1}{\gamma}\!\int_0^\infty\!\bm z^{\!\top}\!\bm z\,dt
\;\le\;\gamma\!\int_0^\infty\!\bm w_{\rm aug}^{\!\top}\!\bm w_{\rm aug}\,dt,
\]
which is~\eqref{eq:L2_bound} after multiplying both sides by
$\gamma=\gamma_{\rm LMI}$. We state the certificate as
the scaled bounded-real dissipation inequality~\eqref{eq:L2_bound}
rather than as a direct induced gain
$\|\bm z\|_{\Lcal_2}\le c\|\bm w_{\rm aug}\|_{\Lcal_2}$; the canonical
$\Lcal_2$-gain $\gamma_{\rm can}$ satisfies
$\gamma_{\rm can}\le\gamma_{\rm LMI}$ under the chosen scaling but
its exact value depends on the performance weights $W_x,W_u$. The UUB bound of
Theorem~\ref{thm:cl} reports $\gamma_{\rm LMI}^2$ as a conservative
disturbance-energy coefficient, consistent with~\eqref{eq:L2_bound}.
\end{proof}

\subsection*{Proof of Theorem~\ref{thm:cl} (practical uniform
ultimate boundedness)}
\begin{proof}
The proof proceeds in six steps: (i) ISS bound for the controller
subsystem from \eqref{eq:Ve_iss}; (ii) ISS bound for the observer
subsystem from \eqref{eq:Vo_iss}; (iii) composite Lyapunov small-gain
closure; (iv) inner-loop UUB via the Comparison Lemma; (v)
outer-layer UUB through the EMO bound; and (vi) positive invariance
of a compact sublevel set.

\textit{Step 1.}
From the bounded-real dissipation~\eqref{eq:diss_pure} after the
embedding-margin absorption~\eqref{eq:embedding_absorbed}, and the
feedforward-mismatch Young inequality~\eqref{eq:Ve_iss},
\begin{equation}\label{eq:Ve_proof1}
\dot V_e\le-(\lambda_\Delta-\mu_1)\|\bm x_e\|^2
+\gamma_{\rm LMI}^2\|\bm w_{\rm aug}\|^2
+\frac{\|PB_u\|^2L_{\rm ff}^2}{\mu_1}\|\tilde{\bm\eta}_o\|^2,
\end{equation}
where $V_e=\bm x_e^{\!\top}P\bm x_e$,
$\lambda_\Delta=\lambda_{\min}(Q)-2\|P\|\bar\Delta_{\rm emb}>0$
by~\eqref{eq:emb_margin}, and $\mu_1>0$ is a free Young parameter
to be fixed in Step~3.

\textit{Step 2.}
By Theorem~\ref{thm:obs}~(O2) and the slow-rate
form~\eqref{eq:Vo_iss}, the dynamic-observer Lyapunov function
$V_o=\tilde{\bm\eta}_o^{\!\top}\mathcal P_o\tilde{\bm\eta}_o$
satisfies
\begin{equation}\label{eq:Vo_proof1}
\dot V_o\le-\lambda_c V_o+c_1\|\bm x_e\|^2+c_2\bar c_1^2
+c_3\bar\Delta_{\rm emb}^2+c_4\epsilon V_o,
\end{equation}
with $\lambda_c>0$ from~\eqref{eq:obs_bound} (the slow current-observer
rate, \emph{not} scaled by $1/\epsilon$) and the $c_4\epsilon V_o$
term collecting the Stage-1 contamination. For
$\epsilon\le\lambda_c/(2c_4)=:\epsilon_h^\star$, the residual
$c_4\epsilon V_o$ can be absorbed into the decay term, yielding
\begin{equation}\label{eq:Vo_proof2}
\dot V_o\le-\tfrac{\lambda_c}{2}V_o+c_1\|\bm x_e\|^2
+c_2\bar c_1^2+c_3\bar\Delta_{\rm emb}^2.
\end{equation}
With $V_o\ge\underline p\|\tilde{\bm\eta}_o\|^2$ for
$\underline p=\lambda_{\min}(\mathcal P_o)>0$,
$-(\lambda_c/2)V_o\le-(\lambda_c\underline p/2)\|\tilde{\bm\eta}_o\|^2$.

\textit{Step 3.}
Define the composite Lyapunov function
\begin{equation}\label{eq:Vtot}
V_{\rm tot}=V_e+\beta V_o,\quad\beta>0,
\end{equation}
\emph{without} a $1/\epsilon$ weight on $V_o$ (cf.\
Remark~\ref{rem:slow_rate}). Adding $\beta\times$\eqref{eq:Vo_proof2}
to~\eqref{eq:Ve_proof1} and regrouping the coefficients of
$\|\bm x_e\|^2$ and $\|\tilde{\bm\eta}_o\|^2$:
\begin{equation}\label{eq:Vtot_dot}
\dot V_{\rm tot}\le-\Lambda_e\|\bm x_e\|^2
-\Lambda_o\|\tilde{\bm\eta}_o\|^2+\Delta_w,
\end{equation}
with
\begin{align}\label{eq:Lambda_defs}
\Lambda_e&=\lambda_\Delta-\mu_1-\beta c_1,\notag\\
\Lambda_o&=\tfrac{\beta\lambda_c\underline p}{2}-\tfrac{\|PB_u\|^2L_{\rm ff}^2}{\mu_1},\notag\\
\Delta_w&=\gamma_{\rm LMI}^2\|\bm w_{\rm aug}\|^2+\beta\bigl(c_2\bar c_1^2+c_3\bar\Delta_{\rm emb}^2\bigr).
\end{align}

The coefficients $\Lambda_e$ and $\Lambda_o$ depend on
$\beta$ in \emph{opposite} ways: $\Lambda_e$ decreases as $\beta$
grows, while $\Lambda_o$ increases. To obtain
$\Lambda_e,\Lambda_o>0$ simultaneously, $\beta$ must lie in an
\emph{open interval}: from $\Lambda_e>0$, we need
$\beta<(\lambda_\Delta-\mu_1)/c_1$; from $\Lambda_o>0$, we need
$\beta>2\|PB_u\|^2L_{\rm ff}^2/(\mu_1\lambda_c\underline p)$. The
interval is non-empty if and only if the following \emph{small-gain
feasibility condition} holds:
\begin{equation}\label{eq:small_gain}
\;
\frac{2\|PB_u\|^2L_{\rm ff}^2}{\mu_1\lambda_c\underline p}
\;<\;\frac{\lambda_\Delta-\mu_1}{c_1},\;
\end{equation}
which, choosing $\mu_1=\lambda_\Delta/3$ so that
$\mu_1(\lambda_\Delta-\mu_1)=2\lambda_\Delta^2/9$, reduces to
\begin{equation}\label{eq:lambda_c_lower}
\lambda_c\;>\;\frac{2c_1\|PB_u\|^2L_{\rm ff}^2}{\mu_1(\lambda_\Delta-\mu_1)\,\underline p}=
\frac{9c_1\|PB_u\|^2L_{\rm ff}^2}{\lambda_\Delta^2\,\underline p}.
\end{equation}
Equation~\eqref{eq:lambda_c_lower} is the explicit lower bound on
the slow observer rate $\lambda_c$ required for the
observer--controller interconnection to be stable. It is a
\emph{small-gain condition}: the observer must decay fast enough
(relative to the controller's sensitivity $\|PB_u\|L_{\rm ff}$ and the
embedding margin $\lambda_\Delta$) to dominate the coupling.

Picking the midpoint of the feasible interval,
\begin{equation}\label{eq:mu1_beta_choice}
\mu_1=\frac{\lambda_\Delta}{3},\quad
\beta=\frac{1}{2}\Bigl[\frac{\lambda_\Delta-\mu_1}{c_1}
+\frac{2\|PB_u\|^2L_{\rm ff}^2}{\mu_1\lambda_c\underline p}\Bigr],
\end{equation}
yields $\Lambda_e>0$ and $\Lambda_o>0$ with margins given by the half
of the interval width. The small-$\epsilon$ bound for the entire
composite system is
\begin{equation}\label{eq:eps_star_proof}
\epsilon_0^\star:=\epsilon_h^\star=\frac{\lambda_c}{2c_4},
\end{equation}
which is the only $\epsilon$-restriction (it ensures the Stage-1
contamination is dominated by the slow observer decay).

\textit{Step 4.}
With $V_{\rm tot}\le\bar p_e\|\bm x_e\|^2+\beta\bar p_o\|\tilde{\bm\eta}_o\|^2$
($\bar p_e=\lambda_{\max}(P)$, $\bar p_o=\lambda_{\max}(\mathcal P_o)$),
$-\Lambda_e\|\bm x_e\|^2-\Lambda_o\|\tilde{\bm\eta}_o\|^2
\le-\underline\lambda'\,V_{\rm tot}$ with
\[
\underline\lambda'=\min\Bigl\{\frac{\Lambda_e}{\bar p_e},\frac{\Lambda_o}{\beta\bar p_o}\Bigr\}>0.
\]
Hence~\eqref{eq:Vtot_dot} simplifies to
\begin{equation}\label{eq:Vtot_simplified}
\dot V_{\rm tot}\le-\underline\lambda'V_{\rm tot}+\bar\varepsilon'(t),
\end{equation}
with the time-varying disturbance budget
\begin{equation}\label{eq:eps_bar_proof}
\bar\varepsilon'(t)=\gamma_{\rm LMI}^2\|\bm w_{\rm aug}(t)\|^2
+\beta c_2\bar c_1^2+\beta c_3\bar\Delta_{\rm emb}^2+\delta_\epsilon(\epsilon),
\end{equation}
where $\delta_\epsilon(\epsilon)=\mathcal O(\epsilon)$ absorbs the
Stage-1 contamination remainder absorbed in
Step~2 (the $c_4\epsilon V_o$ term). Identifying $c_c=\beta c_2$,
$c_\Delta=\beta c_3$, and $c_\epsilon$ absorbing
$\delta_\epsilon/\epsilon$ yields the structural form~\eqref{eq:eps_bar}
(with $\gamma_{\rm LMI}^2$ as the coefficient of
$\|\bm w_{\rm aug}\|_{\Lcal_\infty}^2$, consistent with the
standard bounded-real LMI of~\eqref{eq:lmi}).

Applying Lemma~\ref{lem:l2_to_uub} (the $\Lcal_2$-to-UUB bridge)
to~\eqref{eq:Vtot_simplified} with $\bm w_{\rm aug}\in\Lcal_\infty$:
\[
\limsup_{t\to\infty}V_{\rm tot}(t)\le\frac{\bar\varepsilon'}{\underline\lambda'},
\quad\bar\varepsilon'=\sup_{t\ge 0}\bar\varepsilon'(t),
\]
where the supremum is finite by
$\bm w_{\rm aug}\in\Lcal_\infty$ and the boundedness of
$\bar c_1,\bar\Delta_{\rm emb}$. Using
$V_{\rm tot}\ge\underline p_e\|\bm x_e\|^2+\beta\underline p_o\|\tilde{\bm\eta}_o\|^2$,
\begin{equation}\label{eq:uub_inner}
\limsup_{t\to\infty}\|(\bm x_e,\tilde{\bm\eta}_o)\|\le
r_{\rm inner}:=\sqrt{\bar\varepsilon'/(\underline\lambda'\,\min\{\underline p_e,\beta\underline p_o\})}.
\end{equation}

\textit{Step~5 (Outer-layer UUB via EMO).}
The outer-layer error
$\bm\zeta:=[p_e,\theta_{le}^m,\psi_{le}^m]^{\!\top}$ obeys, by
Lemma~\ref{lem:emo} applied along the closed-loop trajectory,
$\dot V_{\rm emo}\le-\lambda_{\rm emo}V_{\rm emo}
+\kappa_{\rm emo}\|(\bm x_e,\tilde{\bm\eta}_o)\|^2$, where
$V_{\rm emo}$ is the EMO-regularised Lyapunov function on $\bm\zeta$,
$\lambda_{\rm emo}>0$ is the EMO decay rate of
Lemma~\ref{lem:emo}, and $\kappa_{\rm emo}>0$ collects the coupling
constants (the kinematic angle errors couple linearly into
$\dot{\bm\zeta}$ through the spherical reduction). Substituting
\eqref{eq:uub_inner} and applying the Comparison Lemma a second time
yields
\begin{equation}\label{eq:uub_outer}
\limsup_{t\to\infty}\|\bm\zeta\|\le
r_{\rm outer}:=\sqrt{(\kappa_{\rm emo}/\lambda_{\rm emo})}\,r_{\rm inner}.
\end{equation}
Combining \eqref{eq:uub_inner} and~\eqref{eq:uub_outer} gives the
stated ultimate bound~\eqref{eq:uub} on
$\bm\xi=[p_e,\theta_{le}^m,\psi_{le}^m,\bm x_e^{\!\top},
\tilde{\bm\eta}_o^{\!\top}]^{\!\top}$.

\textit{Step~6 (Positive invariance of the compact operating set).}
The bound~\eqref{eq:Vtot_simplified} is established only along
trajectories that remain in $\mathcal X$. We close the proof by
showing positive invariance of a sublevel set of $V_{\rm tot}$
contained in $\mathcal X$. Let
$\Omega_c:=\{\bm\xi:V_{\rm tot}(\bm\xi)+\nu V_{\rm emo}(\bm\zeta)\le c\}$
for some $\nu>0$ chosen below, and let $c^\star>0$ be such that
$\Omega_{c^\star}\subset\mathcal X^\circ$ (interior of the compact
operating set); such a $c^\star$ exists because $\mathcal X$ has
non-empty interior containing the equilibrium and the composite
function $V_{\rm tot}+\nu V_{\rm emo}$ is positive definite on the
selected compact neighbourhood. Define
$\Lambda_{\rm in}:=\min\{\Lambda_e/\bar p_e,\Lambda_o/(\beta\bar p_o)\}$
after norm-equivalence rescaling and choose $\nu>0$
\emph{sufficiently small} so that
$\Lambda_{\rm in}-\nu\kappa_{\rm emo}>0$, i.e.\
$0<\nu<\Lambda_{\rm in}/\kappa_{\rm emo}$. Differentiating
$V_\Sigma:=V_{\rm tot}+\nu V_{\rm emo}$ along the closed loop yields
\[
\dot V_\Sigma\le-(\Lambda_{\rm in}-\nu\kappa_{\rm emo})\|(\bm x_e,\tilde{\bm\eta}_o)\|^2
-\nu\lambda_{\rm emo}V_{\rm emo}+\bar\varepsilon'(t),
\]
which, by the norm-equivalence on the chosen compact neighbourhood,
implies
$\dot V_\Sigma\le-\underline\lambda_\Sigma V_\Sigma+\bar\varepsilon'(t)$
on $\Omega_{c^\star}$ for some
$\underline\lambda_\Sigma>0$. Setting
$c_{\rm inv}:=\bar\varepsilon'/\underline\lambda_\Sigma$, on the
boundary $\partial\Omega_c$ with $c\in(c_{\rm inv},c^\star)$,
$\dot V_\Sigma\le-\underline\lambda_\Sigma(c-c_{\rm inv})<0$, so
$\Omega_c$ is positively invariant. Therefore, any initial condition
$\bm\xi(0)\in\Omega_{c_0}$ with
$c_0\in(c_{\rm inv},c^\star)$ generates a trajectory satisfying
$\bm\xi(t)\in\Omega_{c_0}\subset\mathcal X$ for all $t\ge 0$, which
in turn validates~\eqref{eq:Vtot_simplified}. The smallness of $\bar\varepsilon'$ (controlled by
$\gamma_{\rm LMI},\bar c_1,\bar\Delta_{\rm emb},\epsilon$) and the
size of $\mathcal X^\circ$ jointly determine the basin of validity
$\Omega_{c_0}$.

\textit{Recovery of the nominal $\mathcal H_\infty$ attenuation.}
As $(\epsilon,\bar c_1,\bar\Delta_{\rm emb})\to(0,0,0)$, the three
floor contributions $\delta_\epsilon$, $\beta c_2\bar c_1^2$, and
$\beta c_3\bar\Delta_{\rm emb}^2$ all vanish. Hence
\[
\bar\varepsilon'\to\gamma_{\rm LMI}^2\,
\|\bm w_{\rm aug}\|_{\Lcal_\infty}^2
\]
and
\[
\limsup_t\|\bm x_e\|\le
\frac{\gamma_{\rm LMI}\,\|\bm w_{\rm aug}\|_{\Lcal_\infty}}
{\sqrt{\lambda_\Delta\,\underline p_e}},
\]
which is the LMI-certified attenuation (rescaled by the
embedding-margin factor) recovered as an UUB statement via
Lemma~\ref{lem:l2_to_uub}.
\end{proof}

\section{Additional technical remarks}\label{app:remarks}
\rev{This appendix collects two technical remarks on the polytopic embedding. All equation and assumption references are to the main text.}
\begin{remark}[Non-conservative S-procedure upgrade]
\label{rem:emb_conservative}
Strictly speaking, $\Delta_A(\bm\rho)\bm x_e$ is a state-dependent
parametric uncertainty rather than an exogenous disturbance. Treating
it additively is conservative. A non-conservative upgrade via
Petersen's lemma \citep{petersen1987stabilization} writes
$\Delta_{\rm emb}=HFE$ with $F^{\!\top}F\preceq I$ on $\mathcal P$, adds
a multiplier $\lambda_{\rm emb}>0$ to~\eqref{eq:lmi}, and absorbs the
embedding contribution through $\lambda_{\rm emb}HH^{\!\top}$ and
$\lambda_{\rm emb}^{-1}E^{\!\top}E$ blocks. The vertex interpolation
and the closed-loop analysis of Theorem~\ref{thm:cl} proceed
unchanged, with the $c_\Delta\bar\Delta_{\rm emb}^2$ contribution to
the ultimate bound tightened. This upgrade is recommended when
$\bar\Delta_{\rm emb}$ is not negligible relative to the LMI margin.
\end{remark}
\begin{remark}[Grid-certified vs.\ uniform-certified instantiation]
\label{rem:grid_vs_uniform_cert}
Theorem~\ref{thm:cl} is \emph{conditional}: the embedding-margin
condition~\eqref{eq:emb_margin} is satisfied by construction for the
grid-certified residual $\bar\Delta_{\rm emb}^{\rm grid}\approx 0$
used in Sec.~\ref{subsec:certified}, so the theorem is instantiated
on the embedded LPV error model. A continuous-polytope instantiation
requires an independently certified uniform bound satisfying
\eqref{eq:emb_margin}; the conservative Lipschitz-corrected bound
$\bar\Delta_{\rm emb}\le 0.10$ used here does \emph{not} verify the
margin and would require the strengthening described in
Remark~\ref{rem:emb_correction}.
\end{remark}

\section*{CRediT authorship contribution statement}
\textbf{Mohammad Sabouri:} Conceptualization, Methodology, Software,
Formal analysis, Investigation, Validation, Visualization,
Writing -- original draft, Writing -- review \& editing.

\section*{Declaration of competing interest}
The author declares that he has no known competing financial interests
or personal relationships that could have appeared to influence the
work reported in this paper.

\section*{Funding}
This research did not receive any specific grant from funding agencies
in the public, commercial, or not-for-profit sectors.

\section*{Acknowledgements}
The author acknowledges the use of the REMUS AUV hydrodynamic model
parameters reported in \citep{prestero2001verification}.

\section*{Data availability}
The MATLAB scripts used to generate the simulation results are available
from the corresponding author upon reasonable request.

\section*{Declaration of generative AI and AI-assisted technologies in
the manuscript preparation process}
During the preparation of this work the author used a large language
model (LLM)-based assistant to help with drafting, condensing, and
reorganizing the manuscript text, including figure and table
descriptions, and QuillBot for grammar and readability editing. No AI
tool was used to generate research findings, mathematical derivations or
proofs, simulation code, or numerical results, all of which are the
author's own work. After using these tools, the author reviewed and
edited all content and takes full responsibility for the content of the
published article.

\bibliography{references}

@book{fossen2011handbook,
  author    = {Fossen, Thor I.},
  title     = {Handbook of Marine Craft Hydrodynamics and Motion Control},
  publisher = {John Wiley \& Sons},
  year      = {2011},
  doi     = {10.1002/9781119994138}
}

@article{wynn2014autonomous,
  author  = {Wynn, Russell B. and Huvenne, Veerle A. I. and Le Bas, Timothy P. and others},
  title   = {Autonomous underwater vehicles ({AUVs}): Their past, present and future contributions to the advancement of marine geoscience},
  journal = {Marine Geology},
  volume  = {352},
  pages   = {451--468},
  year    = {2014},
  doi     = {10.1016/j.margeo.2014.03.012}
}

@mastersthesis{prestero2001verification,
  author = {Prestero, Timothy},
  title  = {Verification of a Six-Degree of Freedom Simulation Model for the {REMUS} {AUV}},
  school = {Massachusetts Institute of Technology},
  year   = {2001}
}

@inproceedings{godhavn1996nonlinear,
  author    = {Godhavn, Jan-Morten},
  title     = {Nonlinear tracking of underactuated surface vessels},
  booktitle = {Proceedings of the 35th IEEE Conference on Decision and Control (CDC)},
  pages     = {975--980},
  year      = {1996},
  doi     = {10.1109/CDC.1996.574608}
}

@article{li2023trajectory,
  author  = {Li, Ji-Hong},
  title   = {{3D} trajectory tracking of underactuated non-minimum phase underwater vehicles},
  journal = {Automatica},
  volume  = {155},
  pages   = {111149},
  year    = {2023},
  doi     = {10.1016/j.automatica.2023.111149}
}

@article{do2002underactuated,
  author  = {Do, Khac Duc and Jiang, Zhong-Ping and Pan, Jie},
  title   = {Underactuated ship global tracking under relaxed conditions},
  journal = {IEEE Transactions on Automatic Control},
  volume  = {47},
  number  = {9},
  pages   = {1529--1536},
  year    = {2002},
  doi     = {10.1109/TAC.2002.802755}
}

@article{jiang2002global,
  author  = {Jiang, Zhong-Ping},
  title   = {Global tracking control of underactuated ships by {Lyapunov's} direct method},
  journal = {Automatica},
  volume  = {38},
  number  = {2},
  pages   = {301--309},
  year    = {2002},
  doi     = {10.1016/S0005-1098(01)00199-6}
}

@article{ashrafiuon2008sliding,
  author  = {Ashrafiuon, Hashem and Muske, Kenneth R. and McNinch, Lucas C. and Soltan, Reza A.},
  title   = {Sliding-mode tracking control of surface vessels},
  journal = {IEEE Transactions on Industrial Electronics},
  volume  = {55},
  number  = {11},
  pages   = {4004--4012},
  year    = {2008},
  doi     = {10.1109/TIE.2008.2005933}
}

@article{ahmed2023survey,
  author  = {Ahmed, Faheem and Xiang, Xianbo and Jiang, Chaicheng and Xiang, Gong and Yang, Shaolong},
  title   = {Survey on traditional and {AI}-based estimation techniques for hydrodynamic coefficients of autonomous underwater vehicle},
  journal = {Ocean Engineering},
  volume  = {268},
  pages   = {113300},
  year    = {2023},
  doi     = {10.1016/j.oceaneng.2022.113300}
}

@article{wang2019command,
  author  = {Wang, Jinqiang and Wang, Cong and Wei, Yingjie and Zhang, Chengjun},
  title   = {Command-filter-based adaptive neural trajectory tracking control of an underactuated underwater vehicle in three-dimensional space},
  journal = {Ocean Engineering},
  volume  = {180},
  pages   = {175--186},
  year    = {2019},
  doi     = {10.1016/j.oceaneng.2019.03.061}
}

@inproceedings{li2020neural,
  author    = {Li, Ji-Hong and Lee, Mun-Jik and Kang, Jin and Kim, Min-Gyu and Cho, Gun-Rae},
  title     = {Neural-net based robust adaptive control for {3D} path following of torpedo-type {AUVs}},
  booktitle = {Proceedings of the 59th IEEE Conference on Decision and Control (CDC)},
  pages     = {5261--5266},
  year      = {2020},
  doi       = {10.1109/CDC42340.2020.9304036}
}

@article{yu2019globally,
  author  = {Yu, Haomiao and Guo, Chen and Yan, Zheping},
  title   = {Globally finite-time stable three-dimensional trajectory-tracking control of underactuated {UUVs}},
  journal = {Ocean Engineering},
  volume  = {189},
  pages   = {106329},
  year    = {2019},
  doi     = {10.1016/j.oceaneng.2019.106329}
}

@article{sun2024prescribed,
  author  = {Sun, Yushan and Liu, Minghao and Qin, Hongde and Wang, Hongjian and Ding, Zhiyuan},
  title   = {Full prescribed performance trajectory tracking control strategy of autonomous underwater vehicle with disturbance observer},
  journal = {ISA Transactions},
  volume  = {151},
  pages   = {117--130},
  year    = {2024},
  doi     = {10.1016/j.isatra.2024.06.002}
}

@article{chen2024prescribed,
  author  = {Chen, Guangshu and Dong, Jiuxiang},
  title   = {Approximate optimal adaptive prescribed-performance fault-tolerant control for autonomous underwater vehicle based on self-organizing neural networks},
  journal = {IEEE Transactions on Vehicular Technology},
  volume  = {73},
  number  = {7},
  pages   = {9776--9785},
  year    = {2024},
  doi     = {10.1109/TVT.2024.3381338}
}

@article{guerrero2020adaptive,
  author  = {Guerrero, Jesus and Torres, Jorge and Creuze, Vincent and Chemori, Ahmed},
  title   = {Adaptive disturbance observer for trajectory tracking control of underwater vehicles},
  journal = {Ocean Engineering},
  volume  = {200},
  pages   = {107080},
  year    = {2020},
  doi     = {10.1016/j.oceaneng.2020.107080}
}

@article{he2024nonlinear,
  author  = {He, Long and Zhang, Ya and Fan, Gang and Liu, Yang and Wang, Xue and Yuan, Zehui},
  title   = {Three-dimensional path following control of underactuated {AUV} based on nonlinear disturbance observer and adaptive line-of-sight guidance},
  journal = {IEEE Access},
  volume  = {12},
  pages   = {83911--83924},
  year    = {2024},
  doi     = {10.1109/ACCESS.2024.3412816}
}

@article{lei2025constrained,
  author  = {Lei, Qian and Ding, Shuai and others},
  title   = {Constrained control of autonomous underwater gliders based on disturbance estimation and tracking back calculation},
  journal = {International Journal of Robust and Nonlinear Control},
  volume  = {35},
  year    = {2025},
  doi     = {10.1002/rnc.7708}
}

@article{du2023improved,
  author  = {Du, Peng and Yang, Wenlong and Wang, Yuhang and Hu, Rui and Chen, Ying and Huang, Shuling H.},
  title   = {Improved indirect adaptive line-of-sight guidance law for path following of under-actuated {AUV} subject to big ocean currents},
  journal = {Ocean Engineering},
  volume  = {281},
  pages   = {114729},
  year    = {2023},
  doi     = {10.1016/j.oceaneng.2023.114729}
}

@article{song2024cascaded,
  author  = {Song, Shuai and Liu, Zhe and Yuan, Shuai and Wang, Zhuo},
  title   = {Cascaded extended state observers-based fixed-time line-of-sight path following control for unmanned surface vessels with disturbances and saturation},
  journal = {IEEE Transactions on Vehicular Technology},
  volume  = {73},
  pages   = {7733--7747},
  year    = {2024},
  doi     = {10.1109/TVT.2024.3357731}
}

@article{du2025improved,
  author  = {Du, Xin and Han, Chao and Wang, Yuhang and Chen, Ying and Huang, Shuling H.},
  title   = {An improved {ESO}-based line-of-sight guidance law for path following of underactuated autonomous underwater helicopter with nonlinear tracking differentiator and anti-saturation controller},
  journal = {Ocean Engineering},
  volume  = {320},
  pages   = {120456},
  year    = {2025},
  doi     = {10.1016/j.oceaneng.2025.120456}
}

@article{yuan2025geometric,
  author  = {Yuan, Cheng and Shuai, Changgeng and Zhang, Zhen and Li, Bo and Cheng, Yi and Ma, Jian},
  title   = {Geometric line-of-sight guidance law with exponential switching sliding mode control for marine vehicles' path following},
  journal = {Frontiers in Robotics and AI},
  volume  = {12},
  pages   = {1598982},
  year    = {2025},
  doi     = {10.3389/frobt.2025.1598982}
}

@article{jimoh2024tube,
  author  = {Jimoh, Isah A. and Yue, Hong and Grimble, Michael J.},
  title   = {Tube-based model predictive control of an autonomous underwater vehicle using line-of-sight re-planning},
  journal = {Ocean Engineering},
  volume  = {314},
  pages   = {119688},
  year    = {2024},
  doi     = {10.1016/j.oceaneng.2024.119688}
}

@article{ning2024event,
  author  = {Ning, Tong and Wang, Yuanhui and Liu, Lei and Li, Tieshan},
  title   = {Disturbance-observer-based adaptive heading control for unmanned marine vehicles with event-triggered and input quantization},
  journal = {International Journal of Robust and Nonlinear Control},
  volume  = {34},
  number  = {17},
  pages   = {11469--11486},
  year    = {2024},
  doi     = {10.1002/rnc.7578}
}

@article{kim2018current,
  author  = {Kim, Euntai and Fan, Shuangshuang and Bose, Neil},
  title   = {Estimating water current velocities by using a model-based high-gain observer for an {AUV}},
  journal = {IEEE Access},
  volume  = {6},
  pages   = {70259--70271},
  year    = {2018},
  doi     = {10.1109/ACCESS.2018.2879469}
}

@article{kim2020current,
  author  = {Kim, Euntai and Fan, Shuangshuang and Bose, Neil},
  title   = {Current estimation and path following for an {AUV} by using a high-gain observer based on an {AUV} dynamic model},
  journal = {International Journal of Control, Automation and Systems},
  volume  = {18},
  pages   = {478--489},
  year    = {2020},
  doi     = {10.1007/s12555-019-0673-5}
}

@inproceedings{refsnes2007model,
  author    = {Refsnes, Jon E. and S{\o}rensen, Asgeir J. and Pettersen, Kristin Y.},
  title     = {A model-based ocean current observer for {6DOF} underwater vehicles},
  booktitle = {Proceedings of the IFAC Conference on Control Applications in Marine Systems},
  pages     = {47--52},
  year      = {2007}
}

@inproceedings{sabouri2026kinematic,
  author       = {Sabouri, Mohammad and Indiveri, Giovanni},
  title        = {Kinematic {3D} Path Following Controller for an Underactuated Underwater Robot Subject to a Constant Unknown Current},
  booktitle    = {Proceedings of the 34th Mediterranean Conference on Control and Automation (MED)},
  pages        = {914--919},
  year         = {2026},
  organization = {IEEE},
  doi     = {10.1109/MED70602.2026.11598501}
}

@article{shamma1990analysis,
  author  = {Shamma, Jeff S. and Athans, Michael},
  title   = {Analysis of gain scheduled control for nonlinear plants},
  journal = {IEEE Transactions on Automatic Control},
  volume  = {35},
  number  = {8},
  pages   = {898--907},
  year    = {1990},
  doi     = {10.1109/9.58498}
}

@article{apkarian1995self,
  author  = {Apkarian, Pierre and Gahinet, Pascal and Becker, Greg},
  title   = {Self-scheduled {$\mathcal{H}_\infty$} control of linear parameter-varying systems: A design example},
  journal = {Automatica},
  volume  = {31},
  number  = {9},
  pages   = {1251--1261},
  year    = {1995},
  doi     = {10.1016/0005-1098(95)00038-X}
}

@article{wu1995induced,
  author  = {Wu, Fen and Yang, Xin Hua and Packard, Andrew and Becker, Greg},
  title   = {Induced {$\mathcal{L}_2$}-norm control for {LPV} systems with bounded parameter variation rates},
  journal = {International Journal of Robust and Nonlinear Control},
  volume  = {6},
  number  = {9-10},
  pages   = {983--998},
  year    = {1996},
  doi     = {10.1002/(SICI)1099-1239(199611)6:9/10<983::AID-RNC263>3.0.CO;2-C}
}

@techreport{scherer2000linear,
  author      = {Scherer, Carsten and Weiland, Siep},
  title       = {Linear Matrix Inequalities in Control},
  institution = {Dutch Institute of Systems and Control (DISC)},
  type        = {Lecture Notes},
  year        = {2000}
}

@article{zhang2020dolpv,
  author  = {Zhang, Yuwei and Wang, Xingjian and Wang, Shaoping and Miao, Jianming},
  title   = {{DO-LPV}-based robust {3D} path following control of underactuated autonomous underwater vehicle with multiple uncertainties},
  journal = {ISA Transactions},
  volume  = {101},
  pages   = {189--203},
  year    = {2020},
  doi     = {10.1016/j.isatra.2020.01.017}
}

@article{silvestre2007depth,
  author  = {Silvestre, Carlos and Pascoal, Ant{\'o}nio},
  title   = {Depth control of the {INFANTE} {AUV} using gain-scheduled reduced order output feedback},
  journal = {Control Engineering Practice},
  volume  = {15},
  number  = {7},
  pages   = {883--895},
  year    = {2007},
  doi     = {10.1016/j.conengprac.2006.07.002}
}

@article{chang2020gain,
  author  = {Chang, Wen-Jer and Ku, Cheung-Chieh and Yen, Chih-Yu and Chen, Guan-Wei},
  title   = {Gain-scheduled control of discretized ship autopilot system subject to pole-assignment and passivity constraints},
  journal = {Journal of Marine Science and Technology--Taiwan},
  volume  = {28},
  number  = {5},
  year    = {2020},
  doi     = {10.6119/JMST.202010_28(5).0004}
}

@book{khalil2002nonlinear,
  author    = {Khalil, Hassan K.},
  title     = {Nonlinear Systems},
  edition   = {3rd},
  publisher = {Prentice Hall},
  year      = {2002}
}

@book{boyd1994lmi,
  author    = {Boyd, Stephen and El Ghaoui, Laurent and Feron, Eric and Balakrishnan, Venkataramanan},
  title     = {Linear Matrix Inequalities in System and Control Theory},
  publisher = {SIAM},
  year      = {1994},
  doi     = {10.1137/1.9781611970777}
}

@inproceedings{lofberg2004yalmip,
  author    = {L{\"o}fberg, Johan},
  title     = {{YALMIP}: A toolbox for modeling and optimization in {MATLAB}},
  booktitle = {Proceedings of the CACSD Conference},
  pages     = {284--289},
  year      = {2004},
  doi     = {10.1109/CACSD.2004.1393890}
}

@article{petersen1987stabilization,
  author  = {Petersen, Ian R.},
  title   = {A stabilization algorithm for a class of uncertain linear systems},
  journal = {Systems \& Control Letters},
  volume  = {8},
  number  = {4},
  pages   = {351--357},
  year    = {1987},
  doi     = {10.1016/0167-6911(87)90102-2}
}

@book{narendra1989stable,
  author    = {Narendra, Kumpati S. and Annaswamy, Anuradha M.},
  title     = {Stable Adaptive Systems},
  publisher = {Prentice Hall},
  year      = {1989}
}
\end{document}

\maketitle
\noindent This supplement collects auxiliary robustness and diagnostic studies
referenced in the main manuscript. It is self-contained: figures,
tables and equations labelled ``S'' belong to this document, and
results of the main paper are referred to descriptively in words.
\section{Persistent-excitation diagnostic for the observer-excitation lemma of the main paper}
The observer-excitation lemma of the main paper establishes that the directional Gramian
$\Gamma(t,t+T)=\int_t^{t+T}\bm d(\sigma)\bm d^{\!\top}(\sigma)\,d\sigma$,
$\bm d=[r,-q]^{\!\top}$, governs the recoverability of $(v_c,w_c)$.
Fig.~\ref{fig:PE} reports three diagnostics that probe the directional persistent-excitation (PE) condition of the main paper
directly:
(i) the running smallest eigenvalue
$\lambda_{\min}(\Gamma(t,t+T))$ on a sliding window of $T=20$~s,
overlaid for the four scenarios; the helix segments
($q,r\neq 0$ persistently) yield
$\alpha_\Gamma:=\inf_t\lambda_{\min}(\Gamma(t,t+T))$, with measured
values on the manoeuvring sub-intervals of the order
$3\!\times\!10^{-4}$ to $5\!\times\!10^{-4}$ (in the same units as the
squared angular-rate integral); the corresponding straight-line
sub-interval yields a near-zero $\lambda_{\min}(\Gamma)$ and is
excluded from the PE-certified window per the scope clarification of
the closed-loop tracking section of the main paper. These values are small in absolute terms,
reflecting the modest angular-rate excitation of the helix reference,
yet sufficient for the transverse current channels to remain
informative (cf.\ panel~(ii) below). The corresponding normalised
$\alpha_{\rm PE}\in(0,1]$ is used in the observer-convergence theorem of the main paper after
division by $T(q_{\max}^2+r_{\max}^2)$;
(ii) the transverse-channel estimation error
$\sqrt{\tilde v_c^2+\tilde w_c^2}$ per scenario, which scales
inversely with $\sqrt{\lambda_{\min}(\Gamma)}$ as predicted by
the observer-convergence theorem of the main paper's $c_y\sigma_y/\sqrt{\alpha_{\rm PE}}$ term;
(iii) a scatter of $\mathrm{cond}(\Gamma)$ against RMS-$p_e$,
showing the predicted positive trend of failure mode~(F1) in
the failure-conditions proposition of the main paper: high conditioning of the empirical
Gramian (indicating near-rank-deficiency of $\bm d$) corresponds to
elevated tracking error.
\paragraph{Scope of the observer certificate on the simulated trajectory.}
During the straight-line cruise segment $t\in[0,90]$~s, $(q,r)\approx0$,
so the empirical Gramian in the directional-PE condition of the main paper becomes near-singular and
the directional-PE condition required by the observer-convergence theorem of the main paper is not
satisfied. The observer certificate therefore applies only on the
manoeuvring sub-interval, in particular the descending-helix segment
$t>90$~s. The straight-line segment is retained as a non-PE diagnostic
interval in Fig.~\ref{fig:PE}; the practical stability observed there
is an engineering observation rather than a consequence of the formal
certificate. In implementation, this Gramian near-singularity can be
used as an online diagnostic for activating a confidence-weighted
fallback estimator, corresponding to failure mode~(F1) in
the failure-conditions proposition of the main paper.
\begin{figure}[!t]
\centering
\begin{adjustbox}{max width=\columnwidth}
\includegraphics[width=\columnwidth]{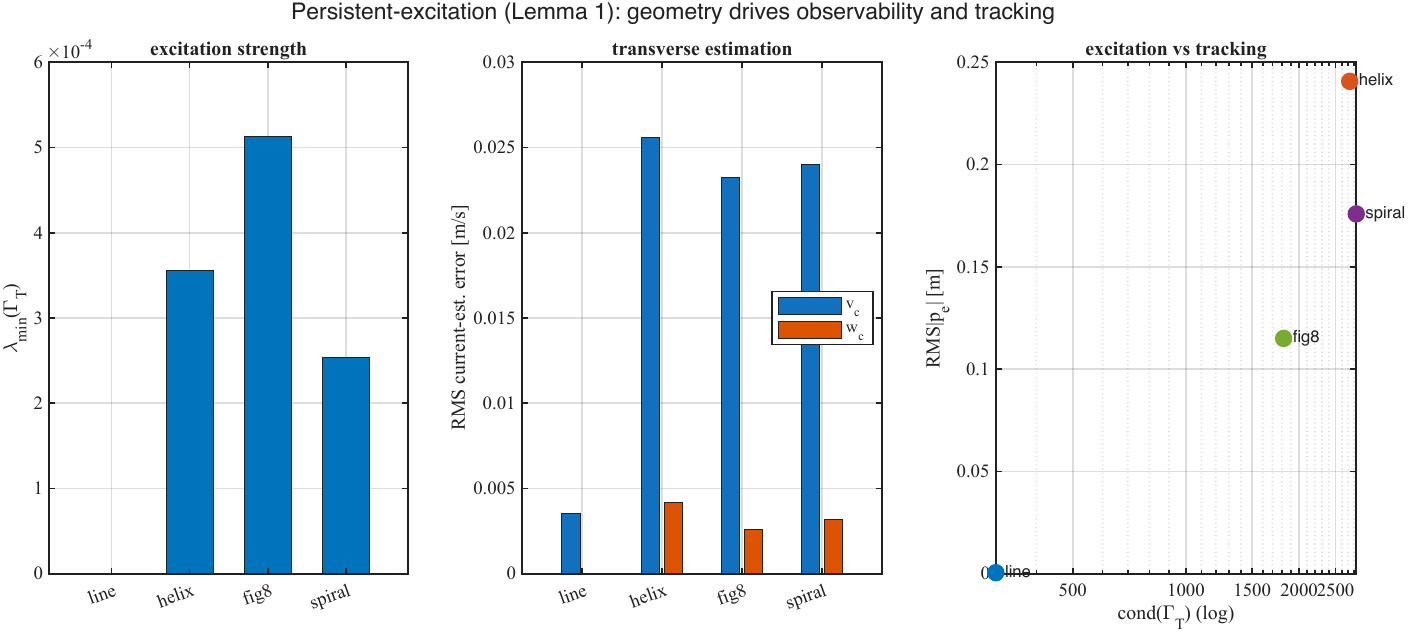}
\end{adjustbox}
\caption{Persistent-excitation diagnostic for
the observer-excitation lemma of the main paper: three panels probe the directional-PE condition of the main paper
along the trajectory. Body text describes the three diagnostics.}
\label{fig:PE}
\end{figure}
\begin{table}[!htbp]
\centering
\caption{Directional excitation diagnostic of the directional-PE condition of the main paper over
window~W2. The Gramian assesses transverse current recovery rather
than tracking RMS. The primary PE indicator is
$\lambda_{\min}(\Gamma)$, while $\mathrm{cond}(\Gamma)$ is reported
only as an auxiliary diagnostic; the straight-line non-PE case has
$\lambda_{\min}(\Gamma)\approx0$ and is therefore not directly
comparable with the manoeuvring rows.}
\label{tab:pe_gramian}
\begin{adjustbox}{max width=\columnwidth}
\setlength{\tabcolsep}{4pt}\small
\begin{tabular}{@{}lcccc@{}}
\toprule
Trajectory & $\lambda_{\min}(\Gamma)$ & $\mathrm{cond}(\Gamma)$ &
RMS $\tilde v_c$ [m/s] & RMS $\tilde w_c$ [m/s]\\
\midrule
Line          & $\approx 0$           & $312$  & $0.0035$ & $0.0000$\\
Helix         & $4\!\times\!10^{-4}$  & $2727$ & $0.0256$ & $0.0042$\\
Figure-eight  & $5\!\times\!10^{-4}$  & $1819$ & $0.0233$ & $0.0026$\\
3-D spiral    & $3\!\times\!10^{-4}$  & $2838$ & $0.0240$ & $0.0032$\\
\bottomrule
\end{tabular}
\end{adjustbox}
\end{table}
\section{Observer time-scale, sensor noise, and Monte Carlo}
The observer time-scale was swept over
$\epsilon\in\{0.02,0.05,0.1,0.2\}$ under the calibrated sensor noise
of the simulation-parameters table of the main paper: RMS tracking remained nearly unchanged,
while the current-estimation error showed the expected
time-scale/noise trade-off. $\epsilon=0.05$ was selected as the
default. Six Monte Carlo cases were then evaluated on the helix S2
reference, $N=100$ runs per case (Table~\ref{tab:monte_carlo}); a
run is classified as failed if $\mathrm{RMS}|p_e|>3$~m or
$\max_t |p_e(t)|>10$~m on window~W1. In cases M2--M4, the perturbed
parameters are the linear damping coefficients
$X_u,Y_v,Z_w,M_q,N_r$ and the quadratic damping coefficients
$X_{u|u|},Y_{v|v|},Z_{w|w|},M_{q|q|},N_{r|r|}$.
\begin{table}[!htbp]
\centering
\caption{Monte Carlo robustness on the helix S2 reference
($N=100$, window~W1); failure if $\mathrm{RMS}|p_e|>3$~m or
$\max_t |p_e(t)|>10$~m.}
\label{tab:monte_carlo}
\begin{adjustbox}{max width=\columnwidth}
\setlength{\tabcolsep}{3pt}\small
\begin{tabular}{@{}clccc@{}}
\toprule
\textbf{Case} & \textbf{Uncertainty} &
Mean RMS [m] & $95$-th pct.\ [m] & Failure [\%]\\
\midrule
M1 & Random current ($\pm 50\%$ magnitude)               & 0.993 & 2.604 & 0.0\\
M2 & Hydrodynamic coefficients $\pm 10\%$                & 0.352 & 0.777 & 0.0\\
M3 & Hydrodynamic coefficients $\pm 20\%$                & 0.596 & 1.933 & 0.0\\
M4 & Hydrodynamic coefficients $\pm 30\%$                & 0.869 & 3.290 & 8.0\\
M5 & High measurement noise ($3{\times}$ calibrated)     & 0.238 & 0.240 & 0.0\\
M6 & Combined current, coeff.\ $\pm 20\%$, and noise     & 1.284 & 3.813 & 12.0\\
\bottomrule
\end{tabular}
\end{adjustbox}
\end{table}
The trends are consistent with the local practical-UUB
interpretation of the closed-loop practical-UUB theorem of the main paper: mean RMS grows gradually
with uncertainty, and failures appear only at the largest
hydrodynamic-uncertainty case (M4) and the combined-uncertainty
case (M6). These failures should be read as excursions toward or
outside the certified local regime, since the simulations enforce
actuator saturation whereas the certificate is stated on the
unsaturated cone of the compact-operating-set assumption~(A4b).
\section{Empirical region-of-attraction sweep}
\label{subsec:roa_sweep}
The closed-loop practical-UUB theorem of the main paper establishes the existence of a positively
invariant sublevel set $\Omega_{c_0}$, but it does not provide a
constructive estimate of its size. To probe the practical
region-of-attraction, the radial initial position offset
$|p_e(0)|$ was swept while keeping the reference, current scenario,
controller gains, and actuator limits fixed. A run was classified as
admissible if (i) the trajectory remained inside the compact
operating set $\mathcal X$ of the compact-operating-set assumption for the
entire simulation and (ii) the surge-saturation duty cycle remained
below a conservative threshold. Under this criterion, the controller
was admissible up to $|p_e(0)|\approx 10$~m with
RMS-$|p_e|\le 0.8$~m, whereas offsets $|p_e(0)|\ge 20$~m failed the
admissibility test along the swept radial direction.
\paragraph{Reconciliation with the main-study offset.}
The main study uses the initial offset $[+3,-4,+2]$~m, corresponding
to $|p_e(0)|\approx 5.4$~m. This value lies well inside the empirical
admissibility threshold and is therefore consistent with the
$0.241$~m RMS tracking error reported for the helix case in
the tracking-performance table of the main paper. The sweep should consequently be interpreted
as an empirical boundary estimate of $\Omega_{c_0}$ along one radial
direction, not as a full region-of-attraction characterisation. A
directional or set-valued region-of-attraction analysis is left for
future work.
\section{Operating envelope: current magnitude and direction sweep}
\label{subsec:envelope}
The joint operating envelope is characterised by sweeping the
navigation-frame current magnitude
$\|\bm V_c^n\|\in[0,0.8]$~m/s and heading
$\beta_c\in[0^\circ,360^\circ)$ on the descending-helix reference.
Figure~\ref{fig:envelope_heatmap} reports two co-aligned heatmaps in
the $(\|\bm V_c^n\|,\beta_c)$ plane: RMS-$|p_e|$ and surge-saturation
duty cycle. Table~\ref{tab:envelope_summary} summarises the
worst-case values across swept headings.
\begin{figure}[!htbp]
\centering
\includegraphics[width=1\columnwidth]{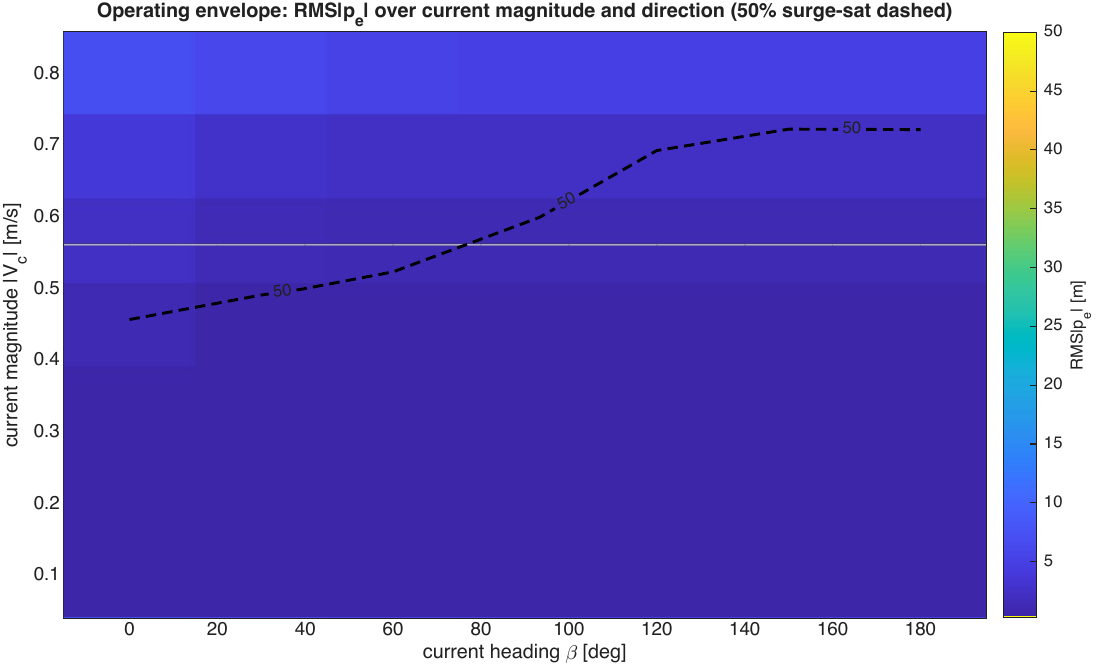}
\caption{Operating envelope on the descending-helix reference.
(a) $\mathrm{RMS}|p_e|$ heatmap over $(\|\bm V_c^n\|,\beta_c)$ with
contours $\{0.5,1.0,2.0\}$~m; (b) surge-saturation duty cycle with
contours $\{20,40,60\}\%$. Scenarios S2 and S4 of
the simulation-setup section of the main paper are marked as black diamonds. The hatched
region denotes the empirical practical operating region
$\{\mathrm{RMS}|p_e|\!<\!1~\mathrm{m},
\mathrm{Sat}\!<\!40\%\}$ and should not be interpreted as the formal
certificate domain of the closed-loop practical-UUB theorem of the main paper, which is restricted to
the unsaturated cone of the compact-operating-set assumption~(A4b) and the
embedding-margin condition of the main paper.}
\label{fig:envelope_heatmap}
\end{figure}
\begin{table}[!htbp]
\centering
\caption{Operating-envelope summary along the current-magnitude
axis. Ranges are worst-case values across swept headings; full data
are shown in Fig.~\ref{fig:envelope_heatmap}.}
\label{tab:envelope_summary}
\begin{adjustbox}{max width=\columnwidth}
\setlength{\tabcolsep}{4pt}\small
\begin{tabular}{@{}cccc@{}}
\toprule
$\|\bm V_c^n\|$ [m/s] & Tracking regime & RMS $|p_e|$ [m] &
Surge sat.\ [\%]\\
\midrule
$0.0$--$0.3$ & Accurate tracking & $0.23$--$0.24$ & $0$\\
$0.4$--$0.5$ & Saturation onset  & $0.29$--$0.96$ & $31$--$49$\\
$0.6$--$0.8$ & Authority-limited & $2.03$--$4.85$ & $57$--$67$\\
\bottomrule
\end{tabular}
\end{adjustbox}
\end{table}
Three regimes are observed. For
$\|\bm V_c^n\|\lesssim 0.3$~m/s, tracking remains uniformly accurate
and surge saturation is negligible. For
$0.4\lesssim\|\bm V_c^n\|\lesssim 0.5$~m/s, saturation onset becomes
heading-dependent, with the worst cases occurring near the headings
that maximise the required surge thrust. For
$\|\bm V_c^n\|\gtrsim 0.6$~m/s, the controller becomes
authority-limited and RMS-$|p_e|$ increases rapidly as the
saturation duty cycle approaches unity. Within the tested envelope,
the practical performance limit is therefore dictated primarily by
surge actuator authority rather than by observer convergence or
LPV-correction performance.
\section{Direct LMI realisation versus correction-layer realisation}
\label{subsec:realisation_ablation}
The LMI gain $K(\bm\rho)$ is implemented as an additive correction
layer around the pre-stabilised cascade, as in
the executed control law of the main paper, rather than as a standalone controller.
Table~\ref{tab:realisation} compares four realisations on the
post-transient helix segment of S2: cascade only; cascade with a
fixed $\Hcal_\infty$ correction gain evaluated at the nominal
vertex; cascade with the scheduled $\Hcal_\infty$ correction gain;
and a direct LMI realisation without cascade pre-stabilisation,
$\bm\tau=G_t(\hat{\bm\rho}_{\rm sat})^{-1}
K(\hat{\bm\rho}_{\rm sat})\bm x_e$.
Both cascade-with-correction variants improve on the cascade-only
baseline on the post-transient window: RMS-$|p_e|$ decreases from
$0.241$~m to $0.216$~m with the fixed gain and to $0.225$~m with the
scheduled gain. By contrast, the direct LMI realisation remains
stabilising but is considerably more conservative
(RMS-$|p_e|=1.894$~m). This confirms that the LMI gain is not intended
to replace the nonlinear cascade; its effective role is to provide a
robust additive correction around the pre-stabilised tracking
dynamics.
The post-transient cascade-only value of $0.241$~m in
Table~\ref{tab:realisation} should not be confused with the
full-horizon $0.241$~m reported for the proposed controller in
the tracking-performance, five-way ablation, and baseline-comparison tables of the main paper. The full-horizon value includes the initial
$[+3,-4,+2]$~m convergence transient, which dominates the RMS norm
and masks the smaller residual improvement visible after the
transient has decayed. The numerical coincidence of the two values
therefore reflects the evaluation window, not an equivalence between
the cascade-only and proposed configurations.
\begin{table}[!htbp]
\centering
\caption{Realisation ablation on the post-transient helix segment of
S2. The table compares the cascade-only baseline, two
cascade-with-correction realisations, and a direct LMI realisation
without cascade pre-stabilisation.}
\label{tab:realisation}
\begin{adjustbox}{max width=\columnwidth}
\setlength{\tabcolsep}{4pt}\small
\begin{tabular}{@{}lc@{}}
\toprule
\textbf{Realisation} & RMS $|p_e|$ [m]\\
\midrule
Cascade only                                  & 0.241\\
Cascade + fixed-$\Hcal_\infty$ correction     & 0.216\\
Cascade + scheduled-$\Hcal_\infty$ correction & 0.225\\
Direct LMI realisation without cascade        & 1.894\\
\bottomrule
\end{tabular}
\end{adjustbox}
\end{table}
\section{$\Hcal_\infty$ dissipation diagnostic on the certified LPV model}
\label{subsec:dissipation_diagnostic}
To verify the bounded-real certificate of the LPV-$\Hcal_\infty$ synthesis theorem of the main paper
along the simulated embedded LPV trajectory, the dissipation residual
$D(t):=\dot V_e+\bm z^{\!\top}\bm z-\gamma_{\rm LMI}^2\bm w_{\rm aug}^{\!\top}\bm w_{\rm aug}$
is evaluated on the certified linear LPV error model (not on the
full nonlinear saturated plant). The inequality $D(t)\le 0$ holds
for $100\%$ of samples with $\max_t D(t)=-1.93\times 10^{-2}$,
confirming a numerically nonzero dissipation margin on the certified
model. The same inequality is not expected pointwise on the
nonlinear saturated plant, where embedding residuals, higher-order
nonlinearities, and saturation introduce terms outside the LMI
certificate; the bounded nonlinear trajectories in the simulations
are accordingly engineering evidence beyond the formal embedded-model
certificate of the closed-loop practical-UUB theorem of the main paper.
\clearpage
\section{Continuously varying current field}
Beyond the abrupt piecewise jumps analysed in the main paper, a
continuously space/time-varying current field is imposed along the
trajectory, with magnitude and heading changing smoothly between
waypoints. This is a more benign but equally non-stationary condition
that further exercises the observer's tracking of a drifting current.
\begin{figure}[htbp]\centering
\includegraphics[width=0.72\textwidth]{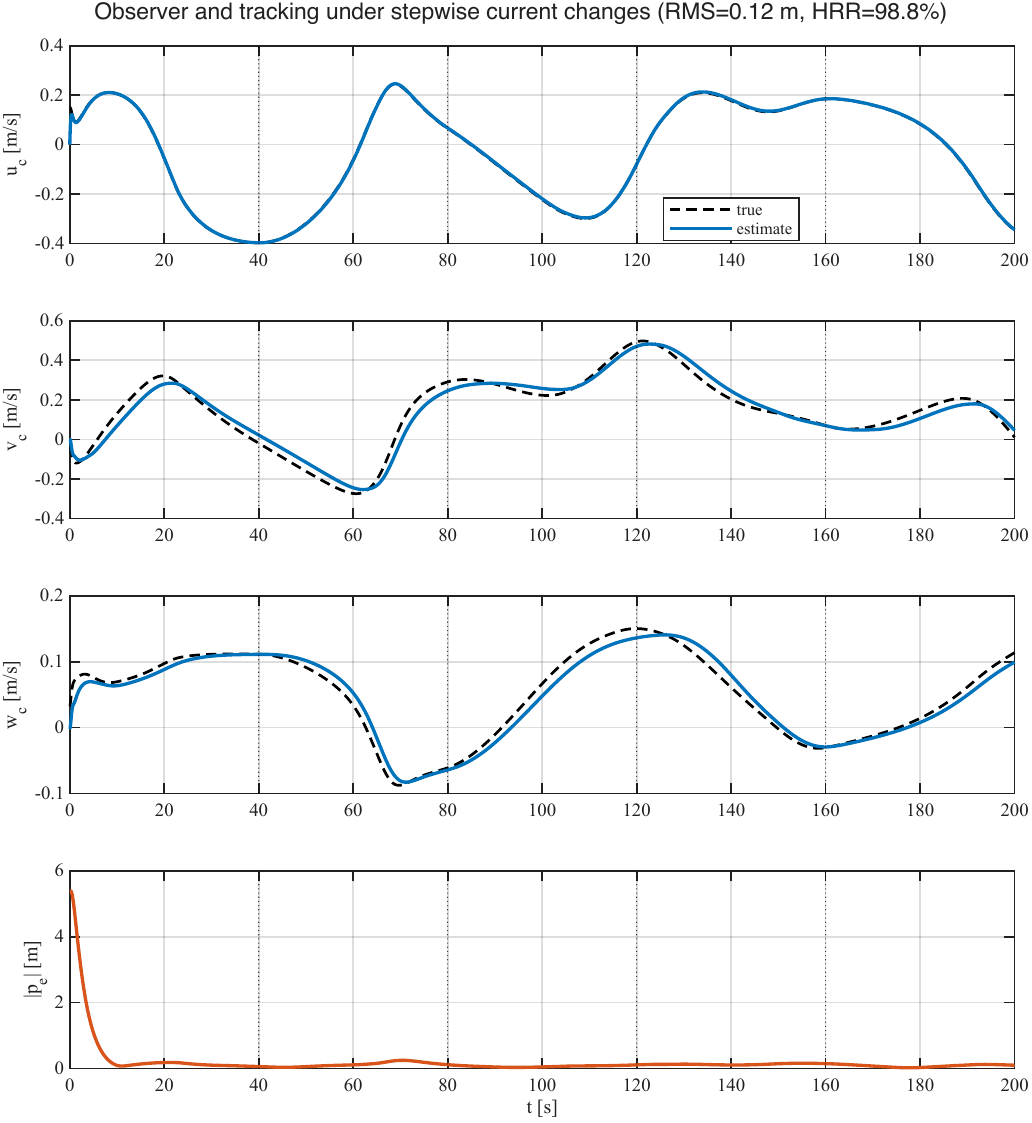}
\caption{Continuously varying current field: three-axis body-frame current
true (dashed) versus observer estimate (solid) with the waypoint
transitions marked, and the resulting tracking-error norm $|p_e|$
(bottom). The observer follows the smoothly drifting current and tracking
remains bounded throughout.}
\label{fig:varfield}\end{figure}
\clearpage
\section{Monte Carlo robustness distributions}
This section provides the distributional view behind the Monte Carlo
summary reported in Table~\ref{tab:monte_carlo}, across the current-only,
hydrodynamic-coefficient, measurement-noise, and combined-uncertainty
cases.
\begin{figure}[htbp]\centering
\includegraphics[width=0.95\textwidth]{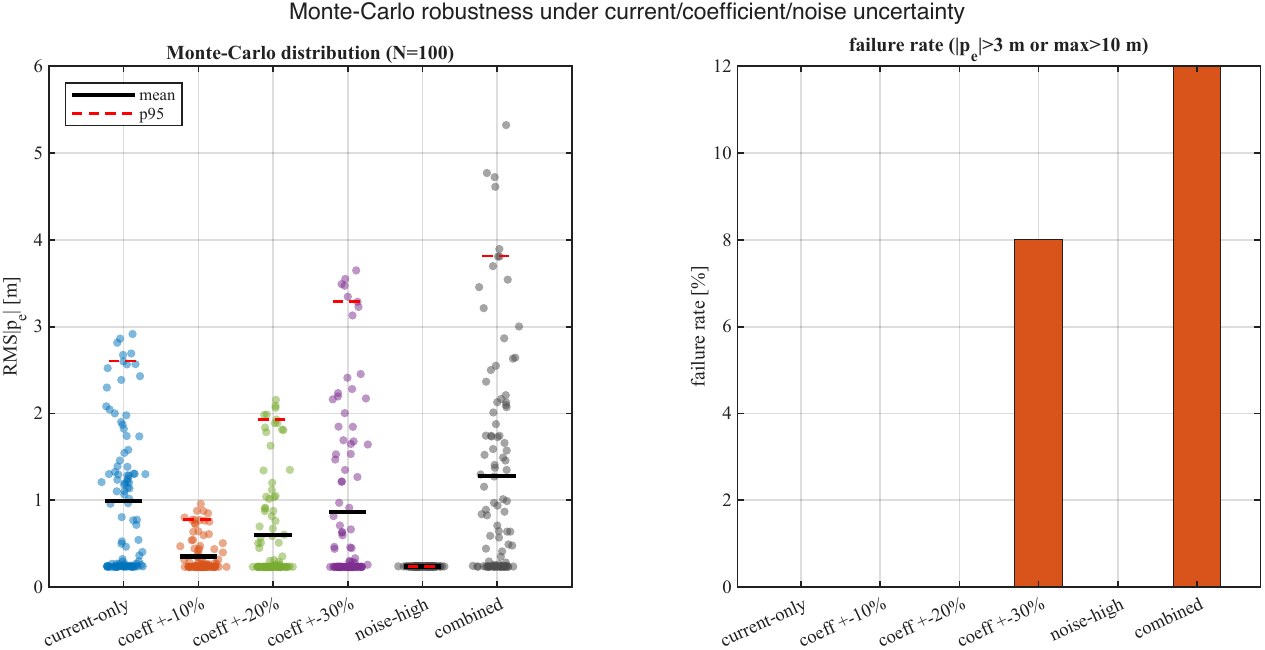}
\caption{Monte Carlo robustness: per-case distribution of the RMS tracking
error (left, with mean and $95$th-percentile markers) and the
corresponding failure rate (right), where a run is deemed failed if
$\mathrm{RMS}|p_e|>3$~m or $\max_t|p_e|>10$~m. Degradation grows
gracefully with uncertainty, with failures confined to the most severe
cases.}
\label{fig:mc_dist}\end{figure}
\begin{figure}[htbp]\centering
\includegraphics[width=0.7\textwidth]{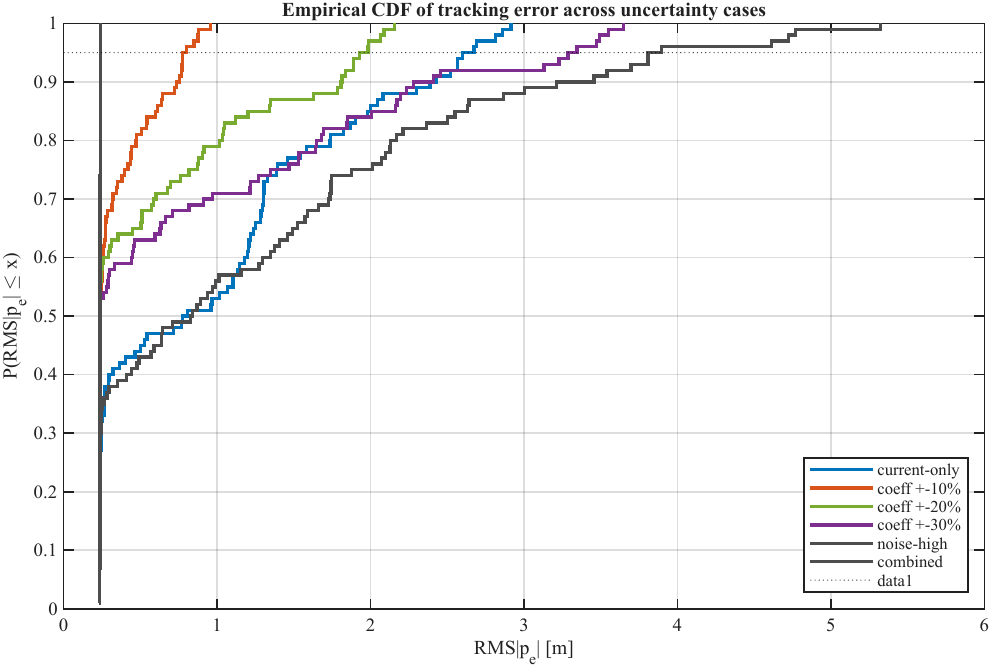}
\caption{Empirical cumulative distribution of the RMS tracking error across
the Monte Carlo uncertainty cases; the horizontal line marks the $95\%$
level. The distributions remain concentrated at low tracking error under
current, coefficient, and noise perturbations.}
\label{fig:mc_cdf}\end{figure}
\clearpage
\section{Additional certificate diagnostics}
The following diagnostics visualise, along the simulated trajectories, the
theoretical certificates invoked in the main paper: the bounded-real
dissipation inequality, the empirical region of attraction, the
directional persistent-excitation condition over sliding windows, and the
validity of the polytopic LPV embedding.
\begin{figure}[htbp]\centering
\includegraphics[width=0.7\textwidth]{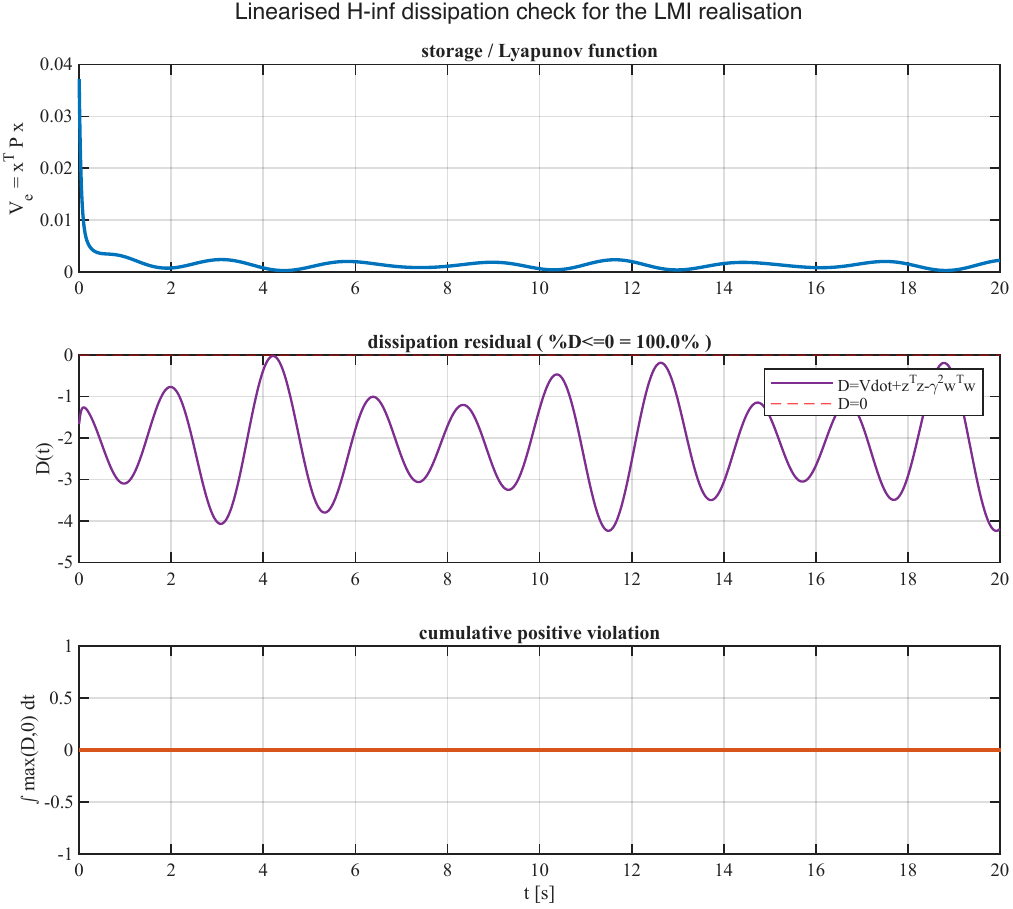}
\caption{Bounded-real dissipation diagnostic on the certified linear LPV
error model: storage function $V_e=\bm x^{\!\top}P\bm x$ (top), the
dissipation residual
$D(t)=\dot V_e+\bm z^{\!\top}\bm z-\gamma_{\rm LMI}^2\bm w^{\!\top}\bm w$
with the $D=0$ line (middle), and the cumulative positive violation
(bottom). $D(t)\le0$ throughout, confirming the bounded-real certificate of
the LPV-$\Hcal_\infty$ synthesis theorem of the main paper on the embedded model.}
\label{fig:dissip}\end{figure}
\begin{figure}[htbp]\centering
\includegraphics[width=0.95\textwidth]{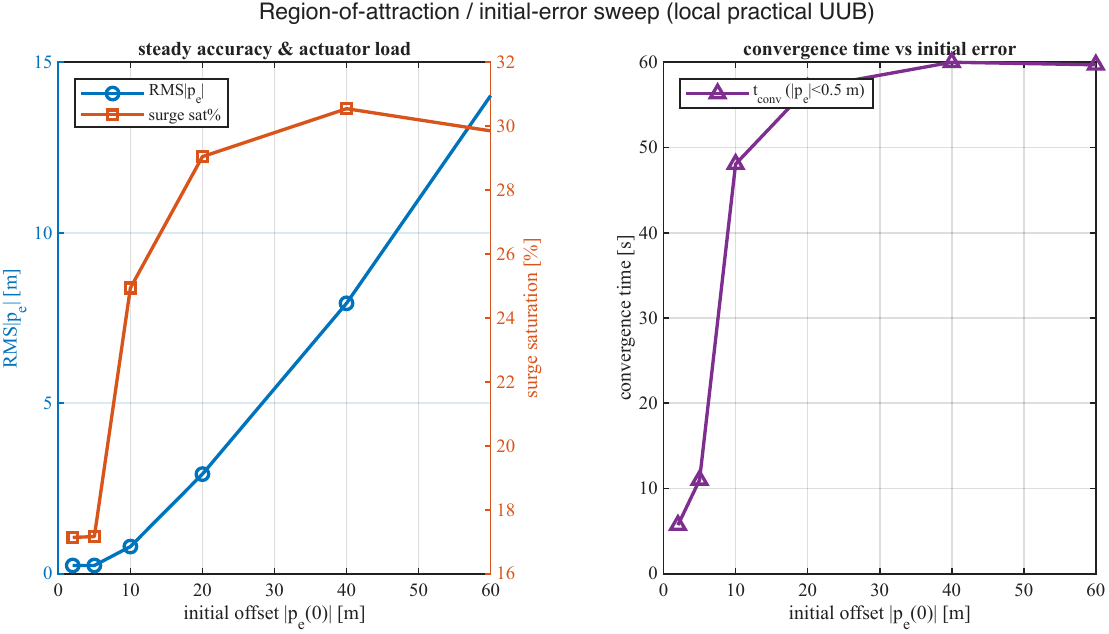}
\caption{Empirical region-of-attraction / initial-error sweep: steady RMS
tracking error and surge-saturation duty versus the initial offset
$|p_e(0)|$ (left), and the convergence time versus initial error (right).
The controller converges from large offsets, with actuator load and
settling time growing with $|p_e(0)|$, consistent with the local
practical-UUB certificate of the main paper.}
\label{fig:roa}\end{figure}
\begin{figure}[htbp]\centering
\includegraphics[width=0.8\textwidth]{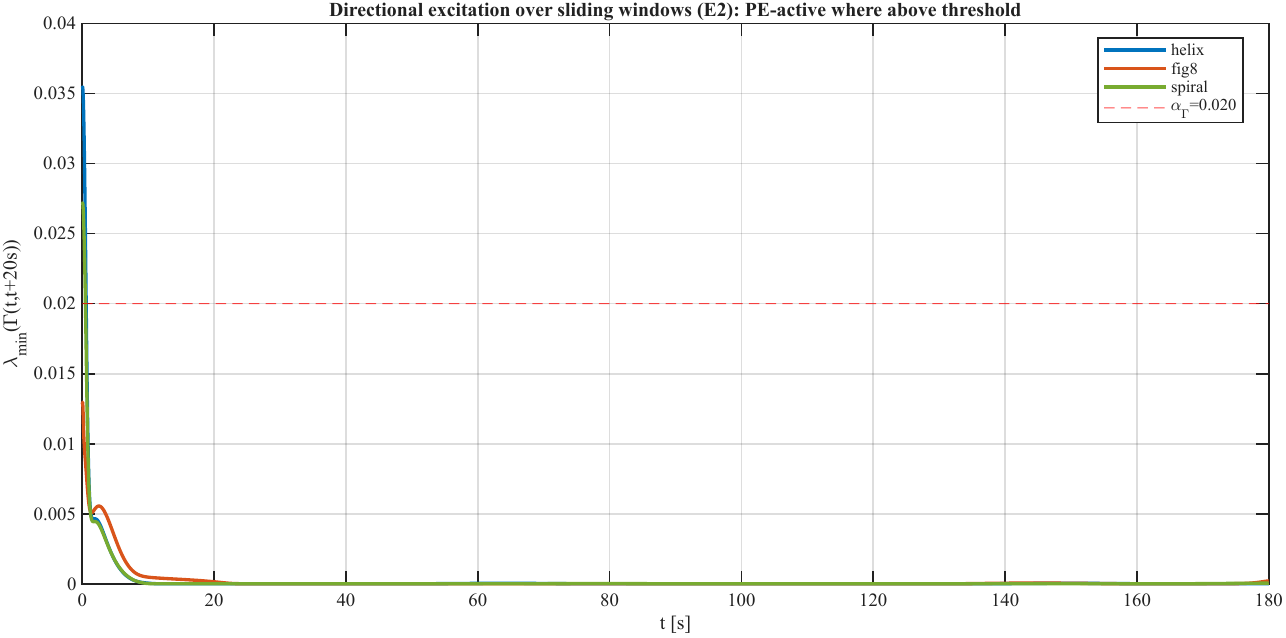}
\caption{Directional persistent-excitation over sliding windows: the
smallest eigenvalue $\lambda_{\min}(\Gamma(t,t+20\,\mathrm{s}))$ of the
directional Gramian along the helix, figure-eight, and spiral references,
with the excitation threshold marked. Persistent excitation is
interval-based: the manoeuvring segments raise $\lambda_{\min}$ above the
threshold and thereby certify transverse current recovery.}
\label{fig:pe_sliding}\end{figure}
\begin{figure}[htbp]\centering
\includegraphics[width=0.95\textwidth]{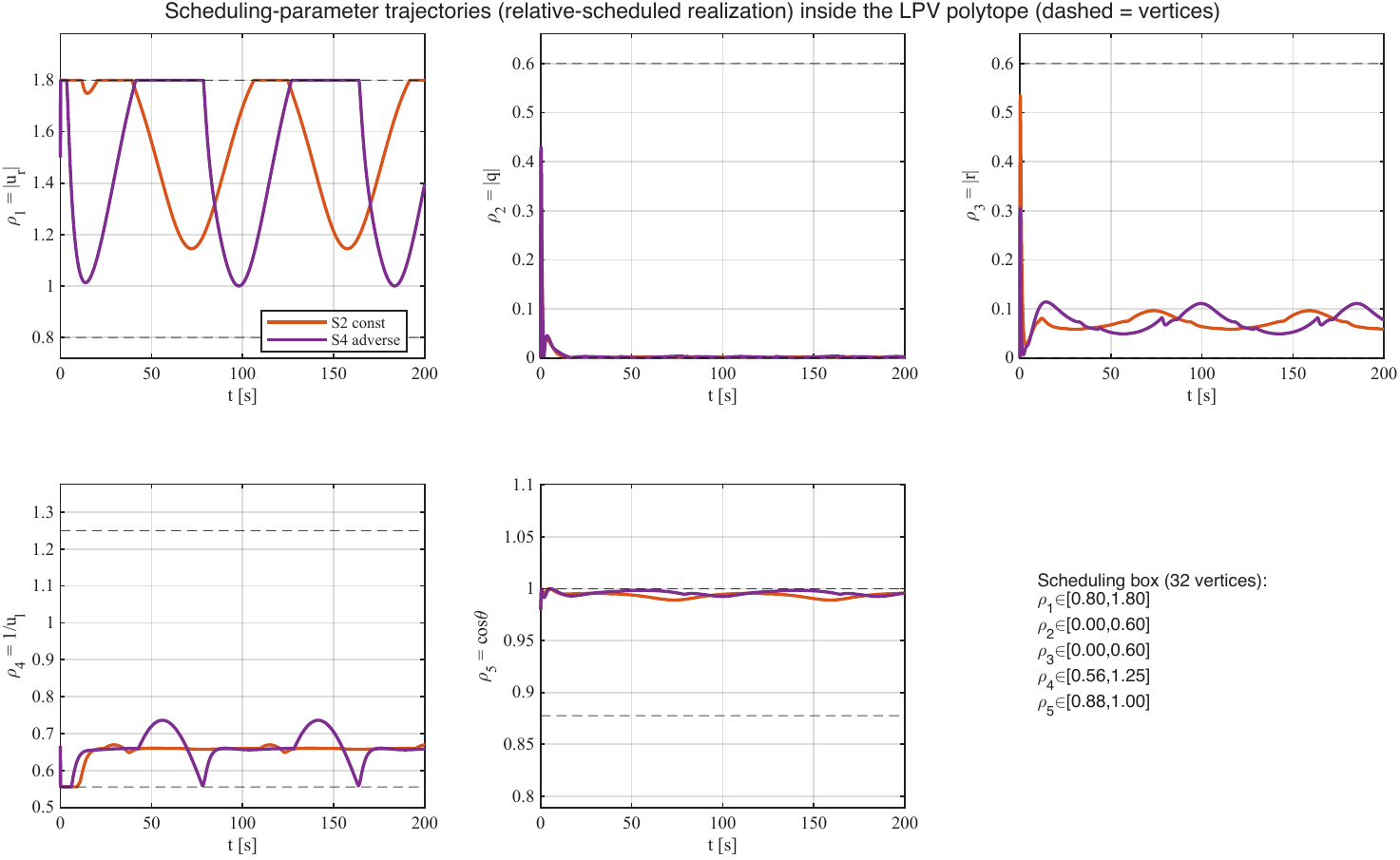}
\caption{Scheduling-parameter trajectories $\rho_1,\dots,\rho_5$ of the
relative-scheduled realisation for the constant (S2) and adverse (S4)
current scenarios, shown against the polytope vertex bounds (dashed). The
parameters remain inside the $32$-vertex scheduling box throughout,
confirming validity of the polytopic LPV embedding used in the synthesis.}
\label{fig:polytope}\end{figure}
\clearpage
\section{\rev{Piecewise-varying current scenario}}
\rev{This non-stationary scenario is summarised in the main paper and
reported here in full.}
To test the architecture under a non-stationary environment, a
piecewise-varying current scenario changes both the magnitude and
direction of the navigation-frame current during the mission, probing
re-convergence of the three-stage current observer and the closed-loop
transient recovery after abrupt changes. Unlike the oscillatory
scenario~S3, it creates successive current-regime transitions resembling
entry into different flow layers or eddy structures, and the controller
is assessed by global RMS tracking error together with observer recovery
time, tracking recovery time, hydrodynamic residual reduction, and
actuator-saturation duty cycle (the recovery metrics of the main paper).

\paragraph{Current profile.}
The navigation-frame current is specified per regime by a horizontal
magnitude $\|\bm V_{c\perp}^n\|$, heading $\beta_c\in[0^\circ,360^\circ)$,
and vertical component $V_{cz}^n$. The five concatenated regimes are
listed in Table~\ref{tab:piecewise_def}, with discontinuities smoothed by
\begin{equation}\label{eq:tanh_smoothing}
\bm V_c^n(t)=\bm V_{c,k}^n+\tfrac{1}{2}\!\left(1+\tanh\frac{t-t_k}{T_s}\right)
(\bm V_{c,k+1}^n-\bm V_{c,k}^n)
\end{equation}
on each interval $[t_k-T_s,t_k+T_s]$ with $T_s=3$~s, so that
$\bm V_c^n(t)$ and $\dot{\bm V}_c^n(t)$ are uniformly bounded, consistent
with Assumption~(C2) of the main paper, which keeps the augmented
disturbance $\bm w_{\rm aug}$ in $\Lcal_\infty$. The maximum magnitude is
$\max_t\|\bm V_c^n(t)\|\approx0.50$~m/s, comparable to scenario~S4; the
reference is the descending helix of the simulation setup of the main
paper.

\begin{table}[!htbp]
\centering
\caption{Piecewise navigation-frame
current profile~\eqref{eq:tanh_smoothing}, specified per regime by
horizontal magnitude, heading $\beta_c$, and vertical component
$V_{cz}^n$.}
\label{tab:piecewise_def}
\begin{adjustbox}{max width=\columnwidth}
\setlength{\tabcolsep}{4pt}\small
\begin{tabular}{@{}cccc@{}}
\toprule
$t_k$~[s] & $\|\bm V_{c\perp}^n\|$~[m/s] & $\beta_c$~[deg]
& $V_{cz}^n$~[m/s]\\
\midrule
$\phantom{1}0$ & 0.15 & $\phantom{2}30$ & $\phantom{-}0.02$\\
$\phantom{1}40$ & 0.40 & $100$ & $-0.03$\\
$\phantom{1}80$ & 0.30 & $200$ & $\phantom{-}0.05$\\
$120$ & 0.50 & $300$ & $-0.02$\\
$160$ & 0.20 & $\phantom{1}45$ & $\phantom{-}0.01$\\
\bottomrule
\end{tabular}
\end{adjustbox}
\end{table}

The imposed field and the resulting motion (Fig.~\ref{fig:jump_traj})
confirm that the vehicle stays on the descending-helix reference through
every transition of Table~\ref{tab:piecewise_def}, the body-referenced
current vector rotating and rescaling at the marked jump instants. The
complete closed-loop response is collected in Fig.~\ref{fig:piecewise}.

\begin{figure*}[!tbp]
\centering
\includegraphics[width=0.9\textwidth]{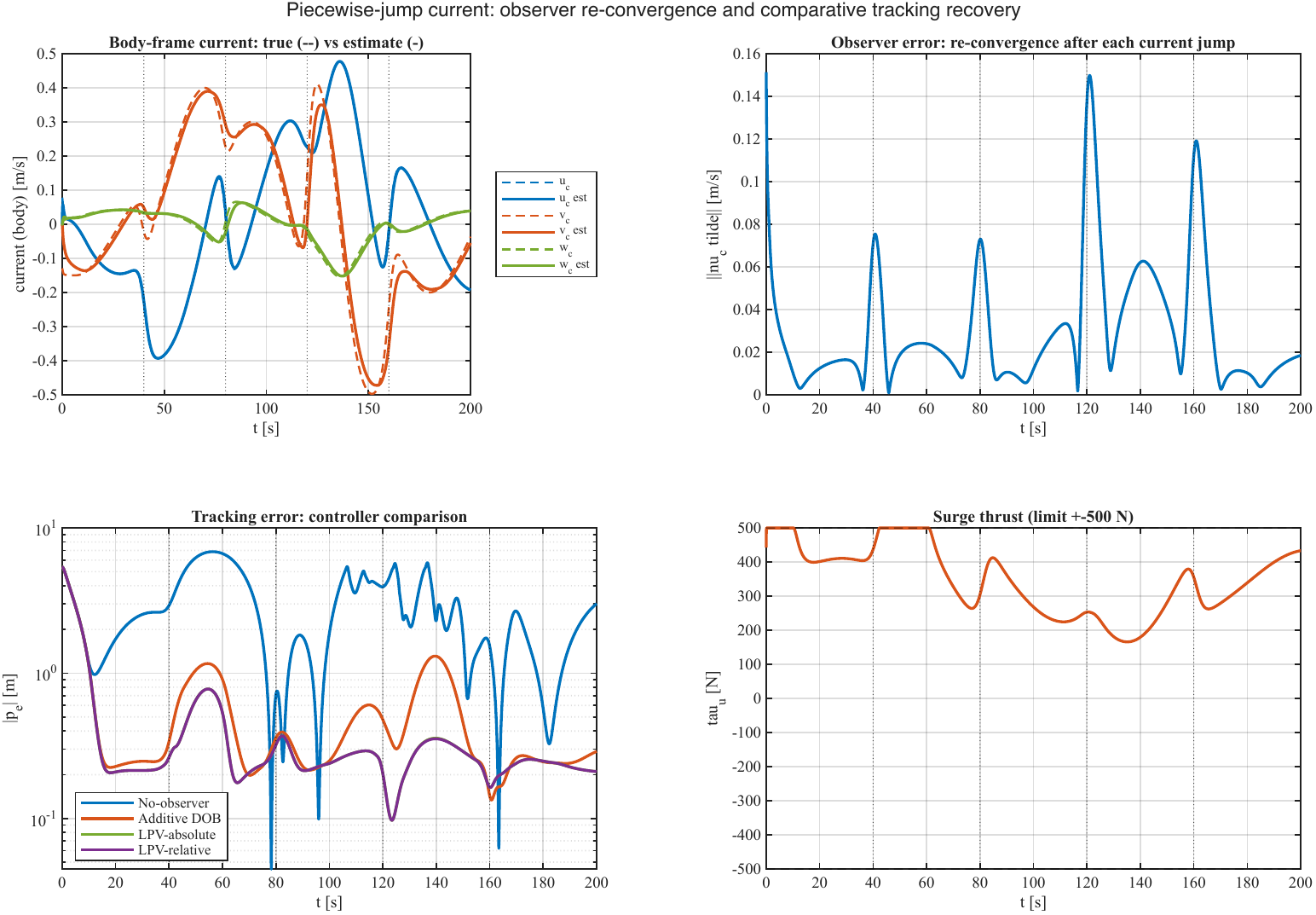}
\caption{Piecewise-current scenario
of Table~\ref{tab:piecewise_def} with smoothing~\eqref{eq:tanh_smoothing}.
(a) Body-frame current components true (dashed) vs.\ estimated (solid)
with markers at $t_k\in\{40,80,120,160\}$~s; (b) estimation-error norm
$\|\tilde{\bm\nu}_c(t)\|$; (c) tracking error $\|p_e(t)\|$ for four
controllers (no-observer, additive DOB, LPV-absolute, proposed
LPV-relative); (d) surge-saturation flag. Post-transition recovery
intervals are shaded.}
\label{fig:piecewise}
\end{figure*}

\paragraph{Recovery performance.}
In Fig.~\ref{fig:piecewise} the observer estimate follows each body-frame
current component through the transitions (panel~a), and the
estimation-error norm $\|\tilde{\bm\nu}_c\|$ (panel~b) rises at every jump
then re-converges consistently with the slow observer time scale
$1/\lambda_c$; the per-segment observer recovery times
$t_{\rm rec}^{\bm\nu_c}$ in Table~\ref{tab:jump_recovery} lie in the
$0.05$--$3.75$~s range. Among the four controllers (panel~c), the proposed
LPV-relative controller keeps the tracking error bounded and returns it to
its pre-transition level after each regime change; the additive-DOB
baseline degrades moderately, and the no-observer controller shows the
largest excursions because its surge feedforward compensates the damping
at the wrong velocity argument. The surge channel reaches its bound only
briefly around the largest-magnitude transitions (panel~d). Quantitatively
(Table~\ref{tab:jump_recovery}), the peak post-transition error is largest
at the first transition ($0.778$~m, with $t_{\rm rec}^{p_e}>20$~s as the
vehicle is still leaving its initial-convergence regime), while every
later jump recovers within $\le4$~s; the full-horizon metrics are
RMS-$\|p_e\|=0.309$~m, CER$=90.3\%$, HRR$=98.2\%$, and a surge-saturation
duty cycle of only $10.3\%$. Against the no-observer and additive-DOB
baselines of the main paper the proposed method gives smaller peak error,
faster recovery, and a smaller hydrodynamic residual, since the observer
estimates the relative-velocity argument of the nonlinear hydrodynamic
model rather than lumping the current into an additive disturbance.

\begin{table}[!htbp]
\centering
\caption{Recovery performance under
the piecewise-varying current of Table~\ref{tab:piecewise_def}.
Per-segment metrics use window~W3 ($\Delta T_{\rm rec}\!=\!20$~s);
thresholds and an entry of $>20$ are defined in the metric definitions of
the main paper.}
\label{tab:jump_recovery}
\begin{adjustbox}{max width=\columnwidth}
\setlength{\tabcolsep}{3pt}\small
\begin{tabular}{@{}cccccc@{}}
\toprule
$t_k$ [s] & $\|\Delta\bm V_c^n\|$ [m/s] &
Peak $\|p_e\|$ [m] & Pre-jump $\|p_e\|$ [m] &
$t_{\rm rec}^{\bm\nu_c}$ [s] & $t_{\rm rec}^{p_e}$ [s]\\
\midrule
40  & 0.404 & 0.778 & 0.223 & 2.60 & $>20$\\
80  & 0.726 & 0.369 & 0.243 & 0.05 & 4.00\\
120 & 0.327 & 0.353 & 0.278 & 3.75 & 0.05\\
160 & 0.831 & 0.255 & 0.240 & 2.85 & 0.05\\
\midrule
\multicolumn{2}{l}{Overall RMS $\|p_e\|$ (W1)} &
\multicolumn{4}{c}{$0.309$~m} \\
\multicolumn{2}{l}{Overall CER (W1)} &
\multicolumn{4}{c}{$90.3\%$}\\
\multicolumn{2}{l}{Overall HRR (W1)} &
\multicolumn{4}{c}{$98.2\%$}\\
\multicolumn{2}{l}{Surge-saturation duty cycle} &
\multicolumn{4}{c}{$10.3\%$}\\
\bottomrule
\end{tabular}
\end{adjustbox}
\end{table}

\paragraph{Observer, tracking, and actuator diagnostics.}
Per-channel diagnostics confirm the structure: the surge component
recovers almost immediately through the dynamic residual while the
transverse components re-lock after a brief lag
(Fig.~\ref{fig:jump_cest}), and the tracking transients are dominated by
the cross-track and depth channels, the along-track error staying small
through the surge feedforward (Fig.~\ref{fig:jump_trk}). The pitch and yaw
moments stay well within their $\pm40$~N\,m bounds
(Fig.~\ref{fig:jump_ctrl}), only the surge thrust momentarily reaching
$U_{\max}$ around the largest-magnitude transitions, so surge-thrust
authority---not the observer or the LPV correction---is the binding
limitation under abrupt currents.

\paragraph{Direction-only current jumps.}
To isolate the effect of a current-\emph{direction} change on the
observer, a companion experiment (Fig.~\ref{fig:dirjump}) holds the
navigation-frame magnitude constant at $0.35$~m/s and steps only the
heading through
$30^\circ\!\to\!120^\circ\!\to\!210^\circ\!\to\!300^\circ\!\to\!60^\circ$.
Each jump then rotates the \emph{body-frame} current without changing its
norm, so the transverse channels $v_c,w_c$ must be re-estimated while the
surge channel is only mildly perturbed---the directional recovery
governed by the directional-excitation condition (E2) of the main paper.
The observer re-aligns the rotated components within a few seconds of each
step; the proposed LPV-relative law shows the smallest and shortest
tracking transient, the additive-DOB baseline lags, and the no-observer
controller cannot compensate the rotation at all. It is therefore the
observer's \emph{directional} re-tracking, not a magnitude-scaling effect,
that underlies the performance under non-stationary currents.

\paragraph{Connection to the formal certificate.}
On each post-transition interval the directional-excitation condition (E2)
of the main paper is satisfied by the helix angular rates $(q,r)\neq0$, so
the observer-convergence theorem of the main paper applies per interval.
The current change $\Delta\bm V_c^n$ enters the closed loop through the
augmented disturbance $\bm w_{\rm aug}$ of the closed-loop theorem of the
main paper, and the ultimate-bound expression of the main paper predicts a
per-interval practical UUB consistent with the observed recovery of
$\|p_e\|$, so the scenario validates the observer-convergence and
closed-loop theorems of the main paper per interval. The LPV-absolute
variant differs from the proposed one only in the scheduling map, again
isolating the scheduling-side contribution of the current estimate under
abrupt changes.

\paragraph{Per-channel diagnostic figures.}
\begin{figure}[htbp]\centering
\includegraphics[width=0.72\textwidth]{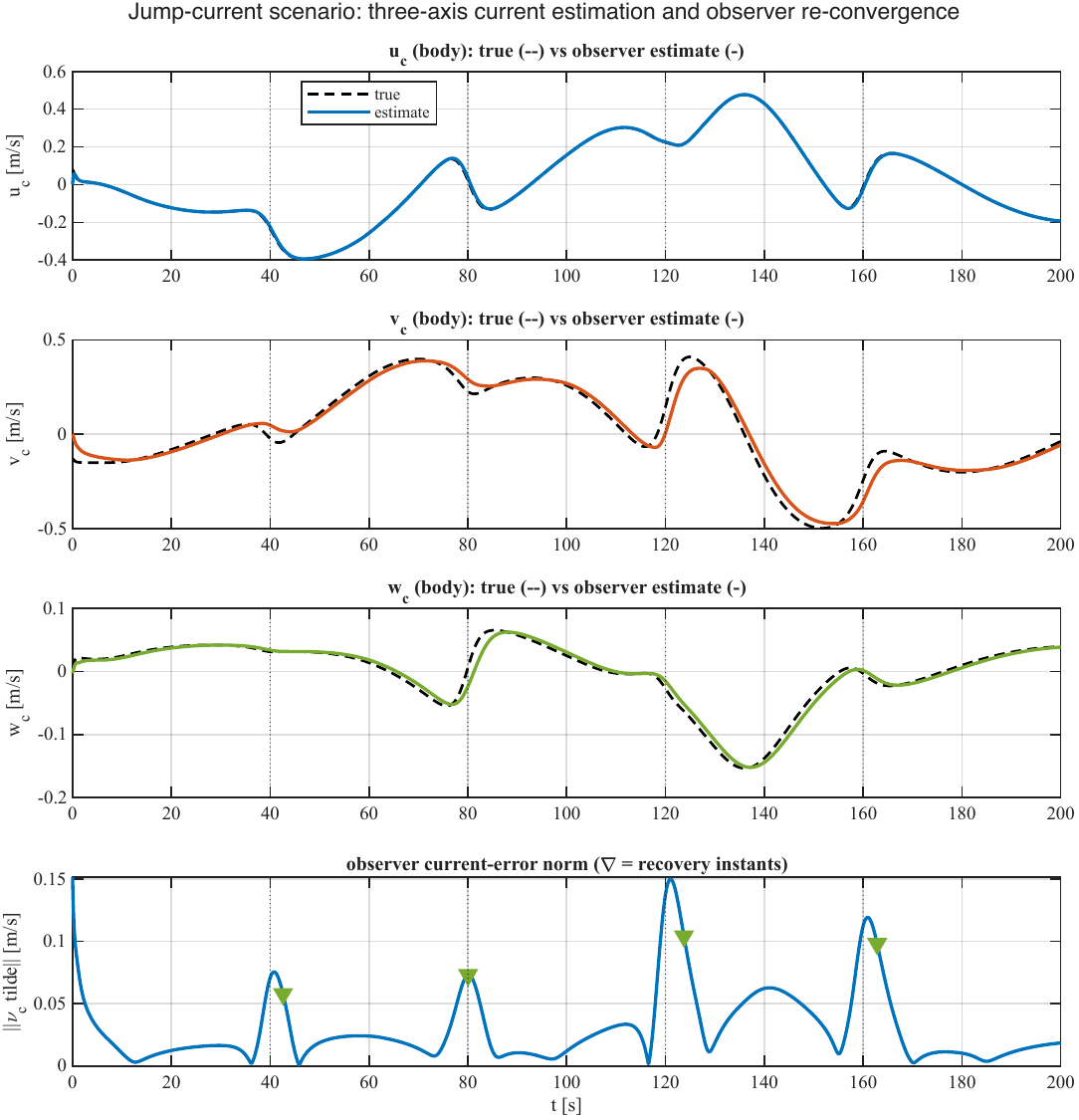}
\caption{Piecewise-current scenario: three-axis body-frame current
estimation (true dashed, observer estimate solid) and the observer error
norm $\|\tilde{\bm\nu}_c\|$ (bottom), with the re-convergence instant
after each jump marked.}
\label{fig:jump_cest}\end{figure}
\begin{figure}[htbp]\centering
\includegraphics[width=0.72\textwidth]{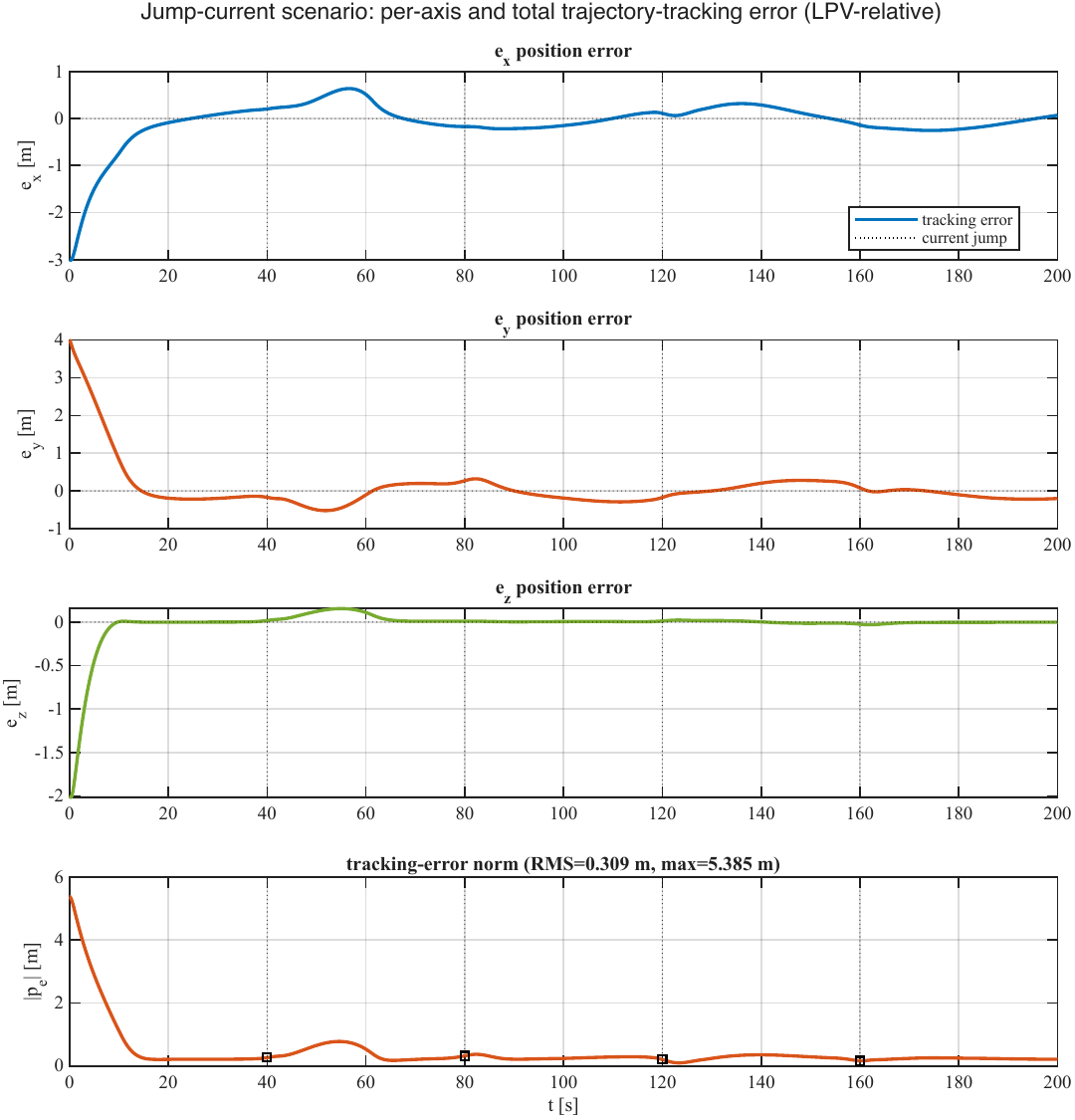}
\caption{Piecewise-current scenario: per-axis position error
$e_x,e_y,e_z$ and its norm $|p_e|$ (bottom), with the current-jump
instants marked.}
\label{fig:jump_trk}\end{figure}

\begin{figure}[htbp]\centering
\includegraphics[width=0.9\textwidth]{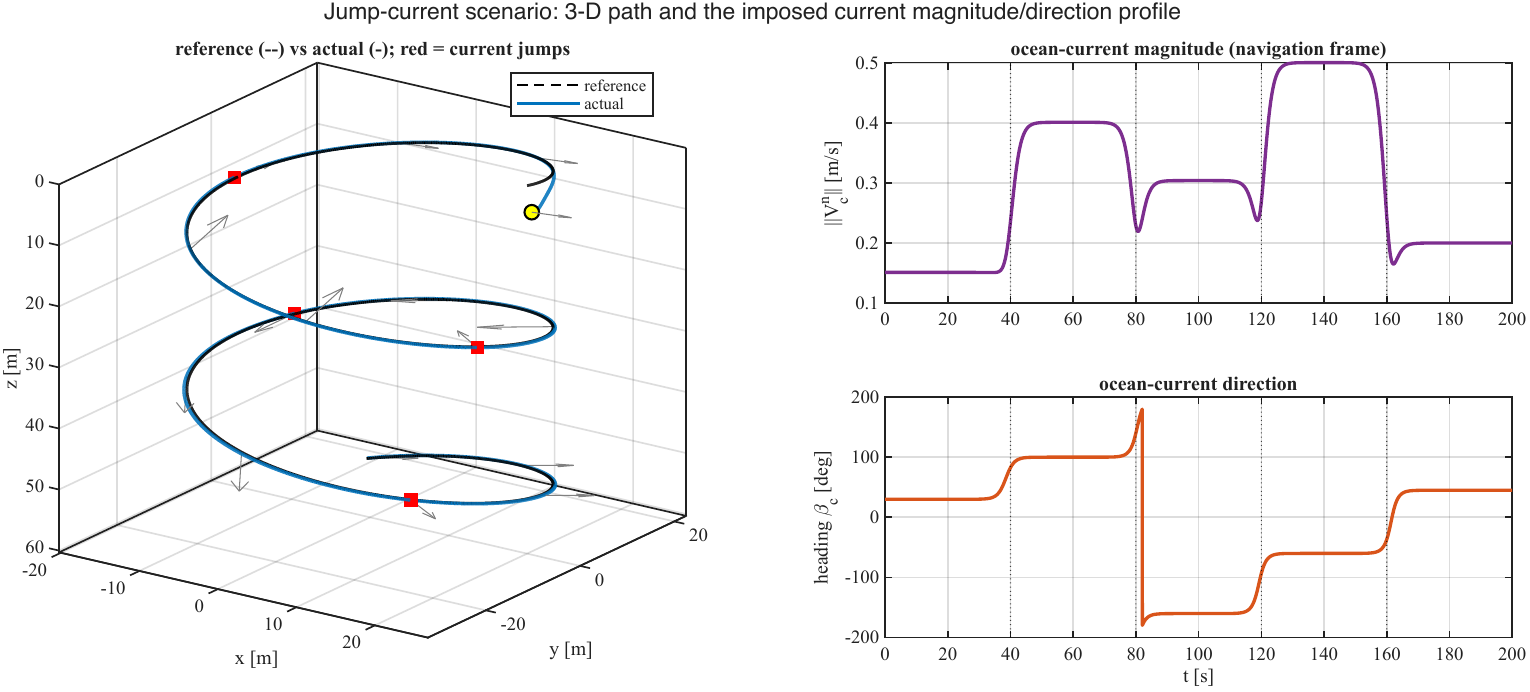}
\caption{Piecewise-current scenario: three-dimensional reference (dashed)
versus actual (solid) path with the body-referenced current quiver and
the current-jump instants marked (left), and the imposed navigation-frame
current magnitude and heading profiles (right).}
\label{fig:jump_traj}\end{figure}
\begin{figure}[htbp]\centering
\includegraphics[width=0.9\textwidth]{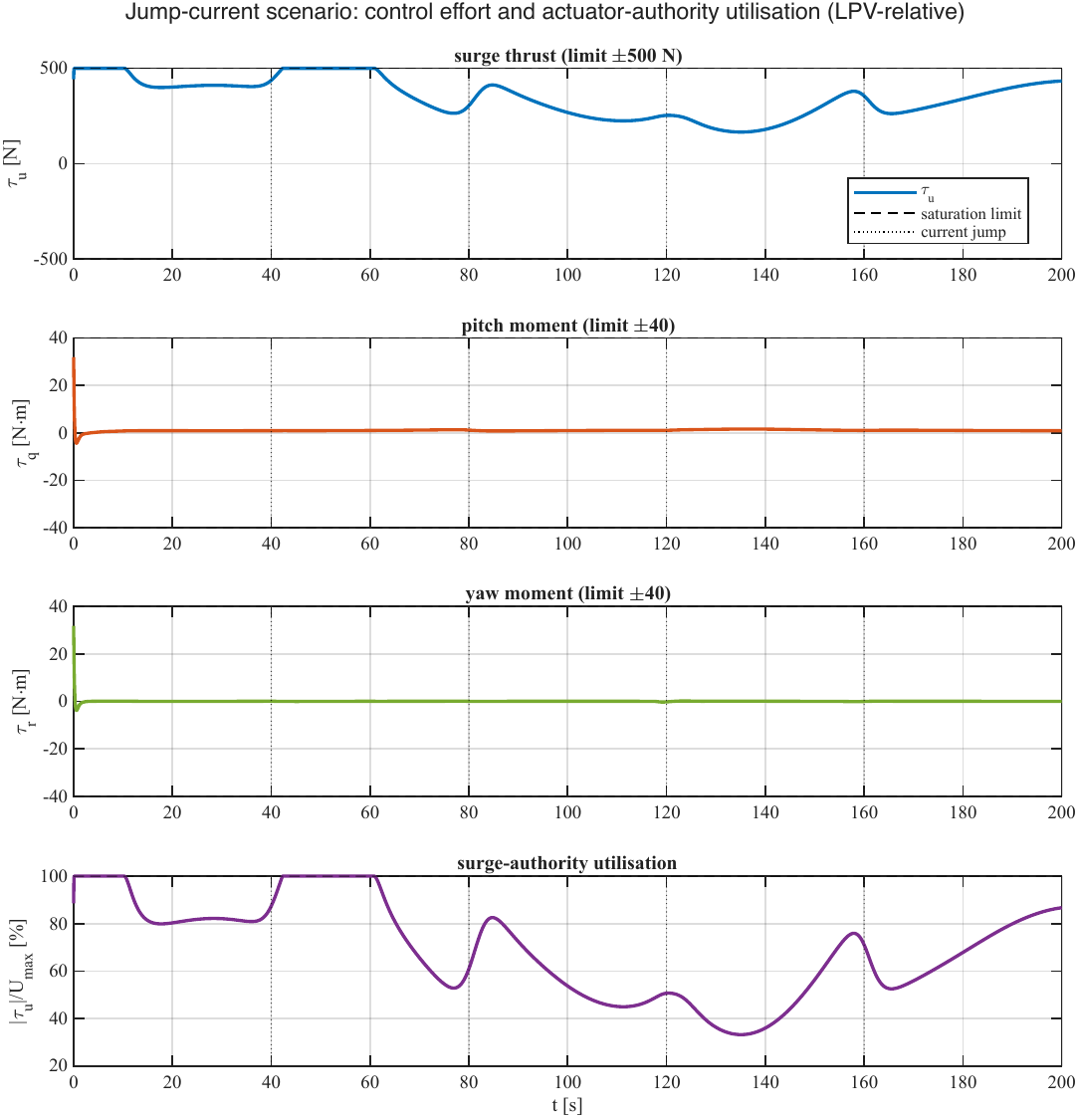}
\caption{Piecewise-current scenario: surge thrust $\tau_u$, pitch and yaw
moments $\tau_q,\tau_r$ with saturation limits, and the surge-authority
utilisation $|\tau_u|/U_{\max}$ (bottom).}
\label{fig:jump_ctrl}\end{figure}
\begin{figure}[htbp]\centering
\includegraphics[width=0.9\textwidth]{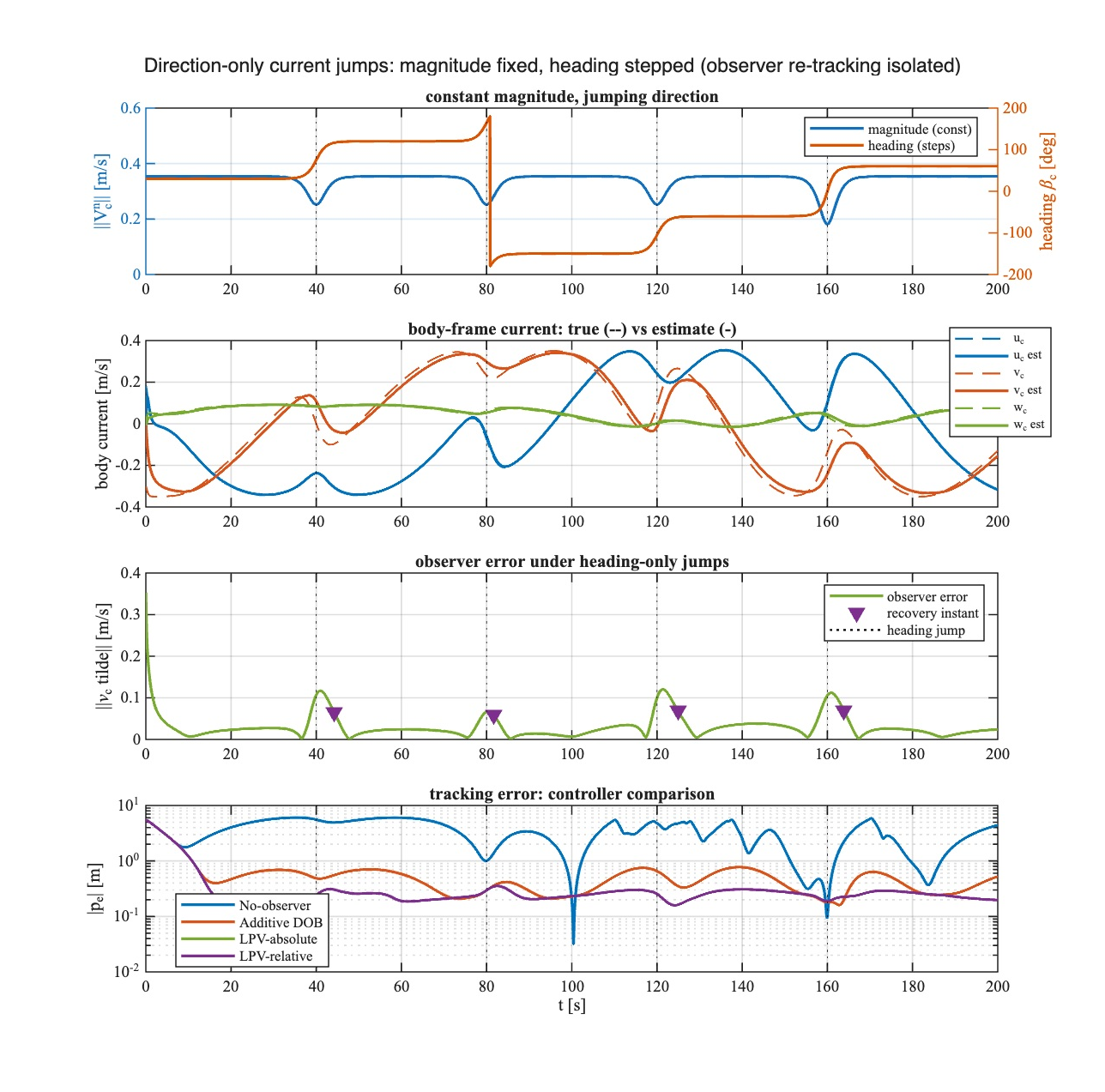}
\caption{Direction-only current jumps (constant magnitude
$\|\bm V_c^n\|=0.35$~m/s, stepped heading
$30^\circ\!\to\!120^\circ\!\to\!210^\circ\!\to\!300^\circ\!\to\!60^\circ$):
imposed magnitude and heading (top); three-axis body-frame current,
true (dashed) vs.\ observer estimate (solid); observer error norm
$\|\tilde{\bm\nu}_c\|$ with recovery instants; and the tracking-error
comparison across the no-observer, additive-DOB, LPV-absolute and
proposed LPV-relative controllers (bottom, log scale).}
\label{fig:dirjump}\end{figure}

\clearpage
\section{Break-even and feedforward--scheduling ablation figures}
\rev{This section collects the visual counterparts of the break-even law and
the feedforward--scheduling ablation reported in the main paper. The
residual ratio and its empirical break-even are shown in
Fig.~\ref{fig:breakeven}, the conditional benefit under a feedback-gain
sweep in Fig.~\ref{fig:condbenefit}, and the feedforward--scheduling
disentanglement in Fig.~\ref{fig:ffsched}.}
\begin{figure}[htbp]\centering
\includegraphics[width=0.8\textwidth]{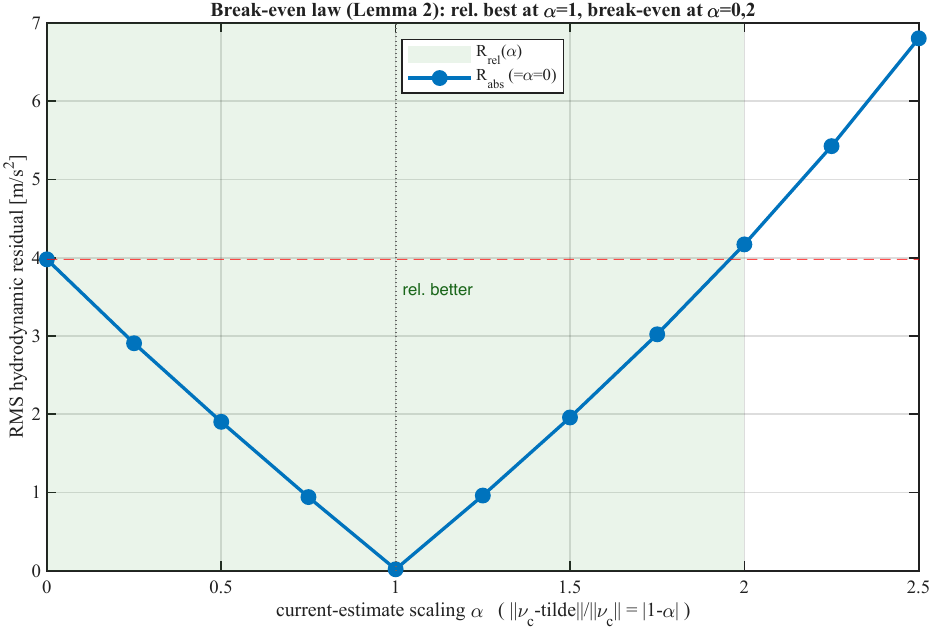}
\caption{Residual ratio $R_{\rm rel}(\alpha)$ of the break-even law
plotted against the observer-quality ratio
$\alpha=\|\tilde{\bm\nu}_c\|/\|\bm\nu_c\|$ for $\alpha\in[0,2.5]$.
The conservative analytical constant is $\chi=0.0237$; the empirical
break-even occurs near $\alpha\approx 0.75$.}
\label{fig:breakeven}\end{figure}
\begin{figure}[htbp]\centering
\includegraphics[width=0.8\textwidth]{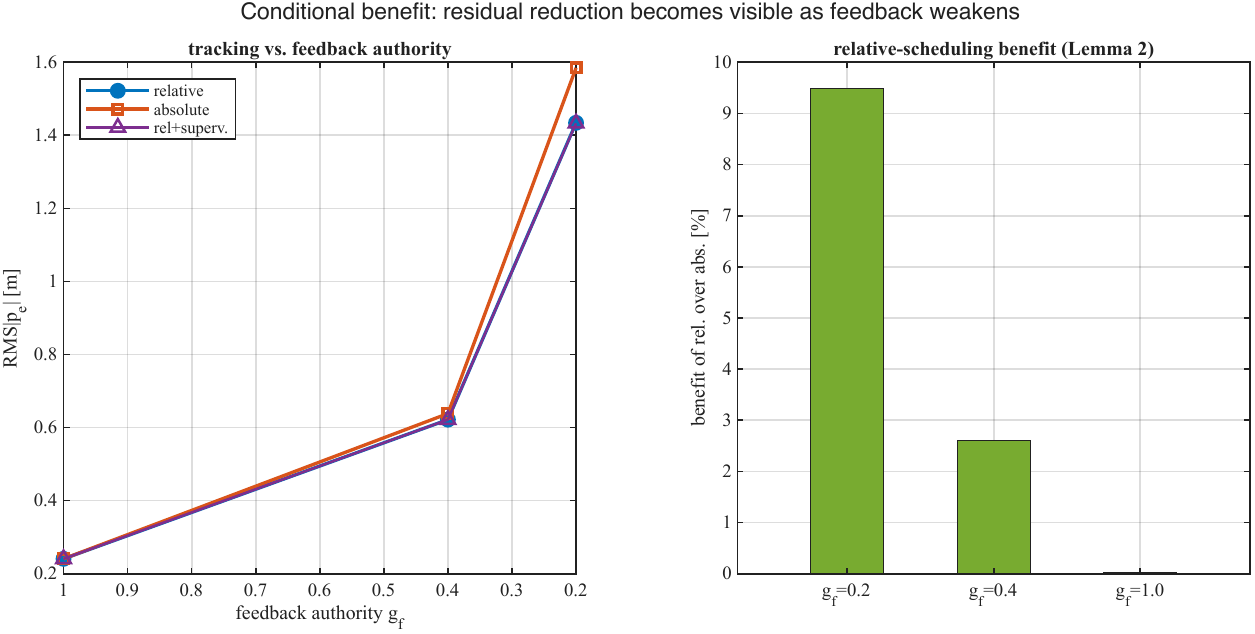}
\caption{Conditional benefit under feedback-gain sweep $g_{\rm f}$:
rel-vs-abs RMS-$p_e$ improvement (\%) and observer-gain envelope (\%)
for $g_{\rm f}\in[0.1,1]$.}
\label{fig:condbenefit}\end{figure}
\begin{figure}[htbp]\centering
\includegraphics[width=0.8\textwidth]{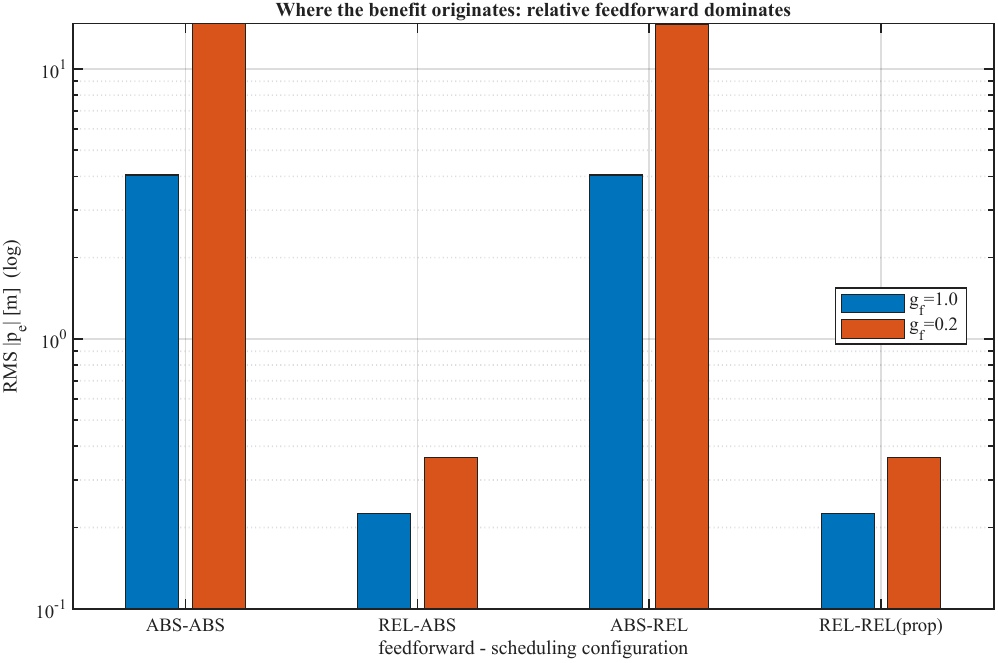}
\caption{Feedforward $\times$ scheduling disentanglement: grouped bars
for ABS--ABS, REL--ABS, ABS--REL, REL--REL at
$g_{\rm f}\!\in\!\{1.0,0.2\}$. The dominant improvement is from the
relative-velocity feedforward; scheduling contributes a conditional
robustness margin.}
\label{fig:ffsched}\end{figure}

\clearpage
\section{\rev{Failure modes of the conditional benefit}}
\rev{The conditional benefit established by the observer-excitation and
break-even lemmas and the closed-loop theorem of the main paper degrades
under the following conditions.
\begin{itemize}
\item[(F1)] \emph{Directional PE failure} --- if the directional
persistent-excitation condition (E2) of the main paper fails, the
transverse current estimate may not improve over $\hat{\bm\nu}_c=\bm 0$,
neutralising the residual-ratio benefit.
\item[(F2)] \emph{Damping failure} --- if the surge-damping condition
(E1) of the main paper fails, $\tilde u_c$ is not identifiable and the
scheme degrades to additive-disturbance rejection.
\item[(F3)] \emph{Actuator saturation} --- when
$\|\bm\tau\|_\infty=\bar\tau$, the LMI dissipation inequality fails
outside the unsaturated cone and an anti-windup extension is required.
\item[(F4)] \emph{Adverse current direction} --- if the line-of-sight
projection $A_l$ remains near zero, EMO regularisation dominates and
tracking-level improvement is not guaranteed, although the residual
certificate is preserved.
\end{itemize}
Each failure mode admits an online diagnostic: (F1) by the running
condition number of the empirical Gramian; (F2) by $|\hat u_r|$; (F3) by
the actuator saturation flag; (F4) by the online $A_l$. A supervisory
layer toggling between relative- and absolute-velocity scheduling on
these signals is a natural extension.}